\documentclass[10pt, a4paper, reqno]{article}
\usepackage{fullpage}
\usepackage{mathtools}

\usepackage{enumitem}
\usepackage{hyperref}
\usepackage{amsthm,amsfonts,amssymb,euscript,mathrsfs}

\usepackage{color}
\usepackage{xcolor}

\newcommand{\R}{\mathbb{R}}

\newcommand{\Z}{\mathbb{Z}}
\newcommand{\N}{\mathbb{N}}
\newcommand{\M}{\mathcal{M}}
\newcommand{\E}{\mathcal{E}}

\renewcommand{\d}{\mathrm{d}}
\usepackage{stmaryrd}
\date{}
\theoremstyle{plain}

\newtheorem{lem}{Lemma}[section]
\newtheorem{thm}{Theorem}[section]
\newtheorem{prop}{Proposition}[section]
\newtheorem{coro}{Corollary}[section]
\newtheorem{mydef}{Definition}[section]
\newtheorem{remark}{Remark}[section]
\numberwithin{equation}{section}

\newcommand{\ffi}{\varphi}

\newcommand{\e}{\varepsilon}
\newcommand{\dr}{\partial}

\newcommand{\dive}{\mathrm{div}}

\newcommand{\Ll}{\mathscr{L}}
\newcommand{\tr}{\mathrm{tr}}

\newcommand{\Lb}{\underline{L}}
\newcommand{\Pol}{\mathrm{Pol}}
\newcommand{\D}{\mathbf{D}}

\newcommand{\GO}[1]{O\left( #1 \right)}

\renewcommand{\l}{\left\|}
\renewcommand{\r}{\right\|}

\newcommand{\half}{\frac{1}{2}}
\newcommand{\enstq}[2]{\left\{#1~\middle|~#2\right\}}

\newcommand{\saut}{\par\leavevmode\par}

\newcommand{\h}{\mathfrak{h}}
\newcommand{\F}{\mathfrak{F}}

\newcommand{\dy}{\Delta}
\newcommand{\spe}[1]{\mathrm{spect}\left( #1 \right)}
\newcommand{\supp}[1]{\mathrm{supp}\left( #1\right)}

\newcommand{\BG}{\mathrm{BG}}

\newcommand{\tBox}{\Tilde{\Box}}

\newcommand{\nabb}{\mbox{$\nabla \mkern-13mu /$\,}}
\newcommand{\lap}{\mbox{$\Delta \mkern-13mu /$\,}}

\newcommand{\pth}[1]{\left( #1 \right)}
\renewcommand{\a}{\alpha}
\renewcommand{\b}{\beta}
\newcommand{\si}{\sigma}
\newcommand{\la}{\lambda}

\title{Geometric optics approximation \\ for the Einstein vacuum equations}
\author{Arthur Touati
\thanks{Institut des Hautes Etudes Scientifiques, Bures-sur-Yvette, France (\href{mailto:touati@ihes.fr}{touati@ihes.fr})}
}

\begin{document}

\maketitle

\begin{abstract}
We show the stability of the geometric optics approximation in general relativity by constructing a family $(g_\lambda)_{\lambda\in(0,1]}$ of high-frequency metrics solutions to the Einstein vacuum equations in 3+1 dimensions without any symmetry assumptions. In the limit $\lambda\to 0$ this family approaches a fixed background $g_0$ solution of the Einstein-null dust system, illustrating the backreaction phenomenon. We introduce a precise second order high-frequency ansatz and identify a generalised wave gauge as well as polarization conditions at each order.  The validity of our ansatz is ensured by a weak polarized null condition satisfied by the quadratic non-linearity in the Ricci tensor. The Einstein vacuum equations for $g_\lambda$ are recast as a hierarchy of transport and wave equations, and their coupling induces a loss of derivatives. In order to solve it, we take advantage of a null foliation associated to $g_0$ as well as a Fourier cut-off adapted to our ansatz. The construction of high-frequency initial data solving the constraint equations is the content of our companion paper \cite{Touati2023a}.
\end{abstract}

\tableofcontents

\section{Introduction}

\subsection{Presentation of the result}\label{section presentation}

In the present article, we are interested in oscillating solutions of the Einstein vacuum equations
\begin{align}
R_{\a\b}(g)=0, \label{EVE'}
\end{align}
where $R_{\a\b}(g)$ is the Ricci tensor of the Lorentzian metric $g$ defined on a manifold $\mathcal{M}$. Our main goal is to extend Choquet-Bruhat's construction from \cite{ChoquetBruhat1969} in order to fully justifies the geometric optics approximation for \eqref{EVE'}. Choquet-Bruhat constructs approximate oscillating vacuum spacetimes, i.e a family $(g_\la)_{\la\in(0,1]}$ of metrics oscillating at frequency $\la^{-1}$ satisfying $R_{\a\b}(g_\la)=\GO{\la}$. We improve this result by constructing a similar family satisfying $R_{\a\b}(g_\la)=0$. We give a rough version of our main result.
\begin{thm}[Rough version of Theorem \ref{theo main}]\label{theo rough'}
There exists a family of metrics $(g_\lambda)_{\lambda\in(0,1]}$ of the form
\begin{align}
g_\lambda = g_0 + \lambda g^{(1)} \left( \frac{u_0}{\lambda}\right) + \GO{\lambda^2}\label{ansatz mega rough'}
\end{align}
solving \eqref{EVE'} in generalised wave gauge on $[0,1]\times \R^3$ where $g_0$ is a given background metric close to Minkowski, $u_0$ a solution of $g_0^{-1}(\d u_0, \d u_0)=0$ and $g^{(1)}$ is a periodic tensor in the variable $\frac{u_0}{\lambda}$.
\end{thm}
In addition to a more refined high-frequency ansatz compared to the one of \cite{ChoquetBruhat1969}, we need to choose a gauge in order to deal with the invariance by diffeomorphism of \eqref{EVE'}. We choose to work in a generalised wave gauge, i.e we rewrite \eqref{EVE'} as a system of quasi-linear wave equations for the metric coefficients. The initial data for this system are given by our companion paper \cite{Touati2023a} which provides oscillating solutions of the constraint equations on $\R^3$. 

\saut
As it can be seen on \eqref{ansatz mega rough'}, the sequence $(g_\la)_{\la\in(0,1]}$ constructed in this article illustrates the backreaction phenomenon, i.e it converges weakly to a non-vacuum spacetime in the high-frequency limit $\la\to0$, denoted by $g_0$ in Theorem \ref{theo rough'}:
\begin{equation}\label{behaviour}
\begin{aligned}
g_\la & \to g_0,\quad\text{uniformly in $L^\infty$,}
\\ \dr g_\lambda& \rightharpoonup \dr g_0,\quad\text{weakly in $L^2_{loc}$.} 
\end{aligned} 
\end{equation}
The background metric $g_0$ is a generic solution of the Einstein-null dust system 
\begin{equation*}
\left\{
\begin{aligned}
R_{\alpha\beta}(g_0) & = F_0^2 \dr_\alpha u_0 \dr_\beta u_0, 
\\ g_0^{-1}(\d u_0, \d u_0) & =0, 
\\ 2 g_0^{\rho\sigma}\dr_\rho u_0 \dr_\sigma F_0 + (\Box_{g_0} u_0) F_0 & = 0. 
\end{aligned}
\right.
\end{equation*}
Therefore, our result is in agreement with Burnett's conjecture from \cite{Burnett1989}, which states that the closure of the set of vacuum spacetimes for the topology associated with the convergence \eqref{behaviour} identifies with the set of kinetic spacetimes solving the massless Einstein-Vlasov system. 

\saut
From the rough ansatz \eqref{ansatz mega rough'}, we can also deduce that Theorem \ref{theo rough'} is a low-regularity existence result for the Einstein vacuum equations. Indeed, by differentiating \eqref{ansatz mega rough'} twice, we obtain $\l g_\la \r_{H^2}\sim \la^{-1}$. This implies that a direct application of the best existence result for \eqref{EVE'}, the celebrated bounded $L^2$ curvature theorem from \cite{Klainerman2015}, to the oscillating data provided by \cite{Touati2023a} would only give a time of existence of order $\la$. Therefore, our result requires a specific construction exploiting the structure of the high-frequency ansatz \eqref{ansatz mega rough'}. Such a construction allows us to show existence on a time interval uniform in $\la$.

\saut
The remainder of this introduction discusses further the main two aspects of our result.
\begin{itemize}
\item In Section \ref{section GO}, we present the standard strategy of geometric optics and describes briefly the new challenges posed by the semi-linear terms and the gauge freedom of \eqref{EVE'}. We also discuss in depth the construction of \cite{ChoquetBruhat1969}.
\item In Section \ref{section Burnett intro}, we compare our result with the state of the art on Burnett's conjecture.
\end{itemize}

\subsection{Geometric optics}\label{section GO}

From a physical point of view, the strategy of geometric optics aims at describing the appearance of rays from wave propagation. The word "optics" reminds us of the very first motivations of physicists to describe the propagation of light from the Maxwell equations. From a mathematical point of view, geometric optics aims at constructing oscillating solutions to non-linear hyperbolic system of equations. We refer to \cite{Rauch2012} and \cite{Metivier2009} for two very rich presentations of the field of geometric optics.

\subsubsection{Towards weakly non-linear geometric optics and transparency}

Consider a generic first order non-linear hyperbolic system
\begin{align}
A_0(u)\dr_t u + \sum_{i=1}^d A_i(u) \dr_i u & = F(u),\label{generic hyp system}
\end{align}
where $u:\R^{1+d}\longrightarrow \R^N$ is the unknown, $A_0$ is a positive definite $N\times N$ matrix, the $A_i$'s are symmetric $N\times N$ matrices, and $F$ is a $\R^N$-valued function. One can wonder what happens to oscillating initial data of the form 
\begin{align*}
u_\la(0,x)= \exp\pth{i\frac{\ffi_0(x)}{\la}} f_0(x) , 
\end{align*}
where $f_0:\R^{d}\longrightarrow \R^N$ is the initial amplitude, $\ffi_0:\R^d\longrightarrow\R$ is a scalar function and $\la>0$ is a small wavelength. In the case where \eqref{generic hyp system} is linear, i.e when $A_0$ and the $A_i$'s don't depend on $u$ and $F$ is a linear function of $u$, this question has been settled in \cite{Lax1957}. There, it is shown that the solution $u_\la$ can be written as a so called WKB ansatz
\begin{align}
u_\la =  \exp\pth{i\frac{\ffi}{\la}} \pth{ f^{(0)} + \la f^{(1)} + \la^2f^{(2)} + \cdots },\label{ansatz u la}
\end{align}
where the phase $\ffi$ satisfies the eikonal equation $\det\pth{ A_0 \dr_t \ffi + \sum_{i=1}^d A_i \dr_i\ffi }=0$, and each $f^{(i)}$ satisfies a polarization condition and is transported along the rays, i.e along the level surfaces of $\ffi$. The whole point of geometric optics is to understand what survives when one goes to the non-linear system. In this case, one expects the waves to interact and to create harmonics. Therefore, a more refined ansatz than \eqref{ansatz u la} is required
\begin{align*}
u_\la \sim \la^p \sum_{n\geq 0}\la^n U_n\pth{ t,x , \frac{\ffi(t,x)}{\la}},
\end{align*}
where the profiles $U_n(t,x,\theta)$ are periodic in $\theta$. The prefactor $\la^p$ measures the amplitude of the oscillations, and the value of $p$ dictates the role played by the non-linear interaction. If $p$ is large, the first profile is still linearly propagated. However if one decreases $p$, we reach the regime of weakly non-linear geometric optics when non-linear terms enter the transport equation for the first profile. We refer to \cite{Hunter1983,Joly1992,Joly1993} for standard results on this regime. For quadratic interactions as in the Einstein vacuum equations, the threshold for this regime is $p=1$. Looking at \eqref{ansatz mega rough'}, this article deals with the weakly non-linear regime for the Einstein vacuum equations. 

\saut
For some systems, the interaction between two waves vanishes due to the particular structure of the non-linear terms: the transport equations are then linear, even in the weakly non-linear regime. This situation is known as transparency in the geometric optics literature, see \cite{Joly2000}. As discussed in depth in \cite{Lannes2013}, transparency is closely related to the famous null condition for quadratic interaction introduced in \cite{Christodoulou1986,Klainerman1986}. As it is well known, the Einstein vacuum equations in wave gauge don't satisfy the null condition but rather a weak form of it, which still plays a major role in both Choquet-Bruhat's approximate construction and the present article's exact construction.

\subsubsection{Approximate geometric optics in general relativity}

The application of geometric optics to general relativity and the Einstein vacuum equations is motivated by a potential non-linear description of gravitational waves. The standard description of gravitational waves is based on the linearized gravity setting, where one considers a metric of the form $g=m+h$ and linearizes \eqref{EVE'} around the Minkowski metric $m$. Assuming the Lorentz gauge and the TT-gauge, one can consider plane wave solutions to the linearized equations with two degrees of freedom corresponding to two polarizations. See Chapter 35 of \cite{Misner1973} for more details.

\saut
Though being historically crucial, linearized gravity fails to describe gravitational waves in a more physical context. First, the restriction to small perturbations $h$ can't, by definition, describe strong gravitational fields. Moreover, one cannot describe the energy of a gravitational wave since any quadratic expression in the perturbation $h$ is neglected. This has the following consequence: the background spacetime on which the gravitational wave propagates doesn't see the wave. As first shown in \cite{ChoquetBruhat1969}, geometric optics provides a mathematical framework that goes beyond linearized gravity. Note that other averaging schemes have been considered in \cite{Brill1964,Isaacson1968,Isaacson1968a} to describe the energy of a gravitational wave.

\saut
We now describe the construction of \cite{ChoquetBruhat1969} (see also Chapter 11 of \cite{ChoquetBruhat2009}). The author considers WKB ansatz for the metric of the spacetime, i.e introduces a small parameter $\lambda>0$ and construct $g_\lambda$ of the form
\begin{align}
g_\lambda & = g_0 + \lambda g^{(1)}\left( \frac{u_0}{\lambda}\right) + \lambda^2 g^{(2)}\left( \frac{u_0}{\lambda}\right)\label{ansatz CB}
\end{align}
satisfying $R_{\mu\nu}(g_\lambda)=\GO{\lambda}$. As in the present article, the notation $g^{(i)}\left( \frac{u_0}{\lambda}\right)$ means that $g^{(i)}$ are symmetric 2-tensor depending on an extra phase argument $\frac{u_0}{\lambda}$.  Since the Ricci tensor involves at most two derivatives of the metric, we have formally $R_{\mu\nu}(g_\lambda) = \frac{1}{\lambda} R_{\mu\nu}^{(-1)}+ R_{\mu\nu}^{(0)} + \GO{\lambda}$. The goal of Choquet-Bruhat's construction is to identify conditions on $g_0$, $g^{(1)}$ and $g^{(2)}$ so that $R_{\mu\nu}^{(-1)}=R_{\mu\nu}^{(0)}=0$. These conditions are the following. To ensure $R_{\mu\nu}^{(-1)}=0$, one must have $g_0^{-1}(\d u_0, \d u_0)=0$ and the polarization condition
\begin{align}
g_0^{\mu\nu} \left( \dr_\mu u_0 g^{(1)}_{\alpha\nu} - \half \dr_\alpha u_0 g^{(1)}_{\mu\nu} \right) = 0.\label{pola g1 intro'}
\end{align}
To ensure $R_{\mu\nu}^{(0)}=0$, one must have the following transport equation for $g^{(1)}$
\begin{align}
2\dr^\alpha u_0 \D_\alpha g^{(1)}_{\mu\nu} + (\D^\alpha \D_\alpha u_0 ) g^{(1)}_{\mu\nu} & = 0. \label{eq g1 intro'}
\end{align}
Moreover, one needs $R_{\mu\nu}(g_0)=\tau \dr_\mu u_0 \dr_\nu u_0$ with $\tau>0$ depending quadratically on $g^{(1)}$ and $g^{(2)}$ to satisfy a polarization condition with a RHS also depending quadratically on $g^{(1)}$. Note that the polarization condition \eqref{pola g1 intro'} is propagated by the transport equation \eqref{eq g1 intro'} and corresponds to the two degrees of freedom of the TT-gauge from linearized gravity. Moreover, the equation for $g_0$ shows how this strategy successfully describes the impact of the wave on the background spacetime. Finally, as announced above, the special structure of the Einstein vacuum equations implies transparency, i.e that gravitational waves do not suffer self-distortion, which translates as the linearity of the transport equation \eqref{eq g1 intro'}.

\saut
Before we describe the exact construction of this article, let us mention that WKB ansatz have also been considered in the linearized gravity setting in \cite{Andersson2021} in order to describe new geometric effect such as the spin Hall effect for gravitational waves, by analogy with the spin Hall effect for light (see \cite{Oancea2020a} from the same authors).

\subsubsection{Exact geometric optics in general relativity}\label{section exact GO}

The article \cite{ChoquetBruhat1969} constructs first order approximate high-frequency vacuum spacetimes. Two questions arise in the spirit of the geometric optics: first, can one construct approximate vacuum spacetimes of any order? That is, for a given $N\in\N$, can one construct $(g_\la)_{\la\in(0,1]}$ such that $R_{\mu\nu}(g_\lambda)=\GO{\lambda^N}$? \textit{A priori}, an iterative procedure similar to Choquet-Bruhat's construction or the author's might give a positive answer, though a very technical one. Of greater importance is the second question: given Choquet-Bruhat's $g_\la$, can one show stability of the geometric optics approximation? That is, can one show that $g_\la$ stays close to an exact vacuum spacetime for a fixed time interval? This question is for instance raised in \cite{Metivier2009}, and Theorem \ref{theo rough'} and its precise version Theorem \ref{theo main} provide a positive answer, thus fully justifying the geometric optics approximation for the Einstein vacuum equations. Note that this is far from automatic, and in many cases approximate solutions of arbitrary order are strongly unstable, see for example \cite{Lebeau2001}.

\saut
Let us now describe the new challenges posed by the construction of exact high-frequency vacuum spacetimes compared with Choquet-Bruhat's construction or with standard geometric optics results. The most striking difference is the need for a more refined high-frequency ansatz, since satisfying $R^{(-1)}_{\mu\nu}=R^{(0)}_{\mu\nu}=0$ is not enough anymore. The ansatz \eqref{ansatz CB} is strictly included in the ansatz \eqref{ansatz mega rough'}, and the $\GO{\la^2}$ in \eqref{ansatz mega rough'} hides some key terms such as a non-oscillating remainder, which absorbs every terms in the Ricci tensor of order $\la^2$ or higher. Compared with \cite{ChoquetBruhat1969}, we also provide the oscillating behaviour of $g^{(1)}$ and $g^{(2)}$, thus describing the creation of harmonics. In order to control the creation of harmonics, we heavily rely on the exact structure of the quadratic non-linearity in \eqref{EVE'}. This is already present in \cite{ChoquetBruhat1969} for second harmonic but we extend it to the third harmonic in $R^{(1)}_{\mu\nu}$.

\saut
The second difference is the need of a gauge choice, not needed if one only constructs approximate solutions to \eqref{EVE'}. We choose to work in a generalised wave gauge, which is part of our geometric optics construction. Concretely, the standard wave gauge $\Box_g x^\a=0$ condition is replaced by $\Box_{g_\la}x^\a=\GO{\la}$ where the $\GO{\la}$ contains oscillating terms determined during the construction. This is the major difference with standard geometric optics results such as \cite{Joly1993}, where the systems of equations at stake are by nature already hyperbolic, whereas for \eqref{EVE'} hyperbolicity precisely comes at a price of a gauge choice. Note that geometric optics for semi-linear gauge invariant systems such as the Yang-Mills equations have been studied in \cite{Jeanne2002}. The situation here is different since \eqref{EVE'} are quasi-linear. Dealing with the oscillations in the gauge and using them to absorb some non-linear interactions is the main formal challenge of this article.

\saut
The third difference is linked to well-posedness, and as for the second, is totally absent from \cite{ChoquetBruhat1969}. As explained at the end of Section \ref{section presentation}, the scaling of weakly non-linear geometric optics implies that Theorem \ref{theo rough'} corresponds to a regime for which well-posedness is not proved, since the $H^2$ norm of $g_\la$ diverges in the high-frequency limit, while only the $H^1$ norm is bounded. This loss of one derivative manifests itself in the coupled wave-transport system between the non-oscillating remainder and $g^{(2)}$. Regaining the missing derivative is the main analytical challenge of this article. A similar issue already occurs in \cite{Huneau2018a}, where the elliptic gauge  is of great help. Here, we introduce Fourier multipliers adapted to the high-frequency expansion of the metric and rely on the background null foliation.

\subsection{Burnett's conjecture}\label{section Burnett intro}

We now turn to the second motivation of our result, the so-called Burnett conjecture, introduced by Burnett in \cite{Burnett1989}. In this article, he consideres a family of metrics $(h_\lambda)_\lambda$ converging to a metric $h_0$ as in \eqref{behaviour}, with the $h_\la$'s all satisfying \eqref{EVE'}. Burnett then asks: what is $R_{\mu\nu}(h_0)-\half R(h_0)(h_0)_{\mu\nu}$? In other words, what effective stress-energy tensor can be produced via weak limits of vacuum spacetimes? He proposes the following double conjecture: 
\begin{itemize}
\item the metric $h_0$ solves the massless Einstein-Vlasov system
\begin{equation}\label{Einstein Vlasov}
\left\{
\begin{aligned}
R_{\alpha\beta}(h_0) & = \int_{h_0^{-1}(p,p)=0} f(x,p) p_\alpha p_\beta \d\mu_{h_0}, 
\\ p^\alpha \dr_\alpha f - p^\alpha p^\beta \Gamma(g_0)^\rho_{\alpha\beta} \dr_{p^\rho}f & =0, 
\end{aligned}
\right.
\end{equation} 
where $\d\mu_{h_0}$ is the measure on the tangent bundle,
\item conversely, any solution of \eqref{Einstein Vlasov} is the limit in the sense \eqref{behaviour} of a sequence of vacuum spacetimes $(h_\lambda)_\lambda$.
\end{itemize}
Before we review the literature, let us explain how the direct and indirect parts of the conjecture interact. The direct conjecture deals with given sequences of solutions to \eqref{EVE'} and tries to characterize their limits' Ricci tensor. The singular nature of these vacuum solutions implies that their own existence is not trivial. Therefore, by answering the indirect conjecture one actually constructs concrete examples of sequences exhibiting the singular behaviour required by the direct conjecture. 

\saut
The direct conjecture has first been adressed in \cite{Green2011}, where it is proved that the effective stress-energy tensor is traceless and satisfies the weak energy condition. This work has led to strong debates among cosmologists, see \cite{Buchert2015,Green2015}, on the role played by backreaction in cosmology. 

\saut
In a series of papers, Huneau and Luk prove Burnett's conjecture in $\mathbb{U}(1)$ symmetry, a symmetry under which the 3+1 Einstein vacuum equations reduce to the 2+1 Einstein-wave map system (see \cite{Moncrief1986}). In this setting, the direct part of the conjecture is proved by the means of microlocal defect measure in \cite{Huneau2019} (see also \cite{Guerra2021}), while the indirect conjecture is proved for $N$ null dusts (a discretized multiphase version of \eqref{Einstein Vlasov}) in \cite{Huneau2018a} and will be fully adressed in the forthcoming \cite{Huneau2023}. The two articles \cite{Huneau2018a,Huneau2023} are both based on high-frequency ansatz in the spirit of geometric optics, but highly benefit from the $\mathbb{U}(1)$ symmetry which allows the construction of an elliptic gauge.

\saut
In \cite{Luk2020}, Luk and Rodnianski provide the first result without any symmetry assumption. Based on their work on the low-regularity characteristic Cauchy problem in double null gauge (see \cite{Luk2012,Luk2017}), they prove both sides of Burnett's conjecture. Their proof is not based on a high-frequency ansatz, and thus can handle loss of strong convergence not only by oscillation but also by concentration. It also sheds new light on the Einstein-null dusts system and on the formation of trapped surfaces. They impose extra regularity along the angular directions compared with the null directions associated to the double null gauge. This has the following consequence: \cite{Luk2020} only proves Burnett's conjecture for two null dusts.

\saut
Our present geometric optics construction proves Burnett's indirect conjecture without any symmetry assumption but only for one null dust. Therefore, it is somehow an intermediate result between Huneau and Luk's approach and Luk and Rodnianski's. However, thanks to \cite{Huneau2018a} and multiphase geometric optics results such as \cite{Joly1993}, there is hope that the present proof can be extended to an arbitrary number of null dusts, since the generalised wave gauge does not single out any null direction.

\subsection{Acknowledgments}

The author would like to thank his advisor Cécile Huneau for all the helpful suggestions and encouragement. The author is also very grateful to the anonymous referees for their constructive inputs that considerably improve the initial manuscript.

\subsection{Outline of the article}

The remainder of this article is structured as follows.
\begin{itemize}
\item In Section \ref{section statement}, we introduce the usual analytical material, define our class of initial data and state the main result on the existence of the sequence $(g_\lambda)_\lambda$, i.e Theorem \ref{theo main}.
\item In Section \ref{section strategy}, we present the strategy of the proof of Theorem \ref{theo main}.
\item Sections \ref{section high-frequency ansatz} to \ref{section solving EVE} are devoted to the proof of Theorem \ref{theo main}.
\end{itemize}

\section{Statement of the results}\label{section statement}

In this section, we first introduce some material, define the background spacetime and finally state our main result.

\subsection{Preliminaries}

\subsubsection{Notations and function spaces}\label{section notations}

We will work on the manifold $\mathcal{M}\vcentcolon = [ 0,1]\times \R^3$. For $t\in[0,1]$, we denote by $\Sigma_t$ the hypersurface $\{t\}\times \R^3$. We introduce the following standard conventions and notations.
\begin{itemize}
\item The greek indices go from 0 to 3 and will always refer to the usual global coordinates system $(t,x^1,x^2,x^3)$ on $\mathcal{M}$, while latin lower case indices go from 1 to 3 and correspond to the three spatial directions.
\item Let $T$ and $S$ be two symmetric 2-tensors and $g$ a Lorentzian metric on $\mathcal{M}$. We define a scalar product by
\begin{align*}
|T\cdot S|_g = g^{\alpha\beta}g^{\mu\nu} T_{\alpha\mu} S_{\beta\nu},
\end{align*}
with associated norm $|T|^2_g = |T\cdot T|_g$.  The trace of $T$ with respect to $g$ is defined by $\tr_g T = g^{\alpha\beta}T_{\alpha\beta}$.
\item If $f$ is a scalar function on $\mathcal{M}$, we denote by $\nabla f$ a spatial derivative of $f$ and by $\dr f$ a time derivative $\dr_t f$ or $\nabla f$. The wave operator associated to a Lorentzian metric is defined by
\begin{align*}
\Box_g f & = g^{\alpha\beta} \left( \dr_\alpha\dr_\beta f - \Gamma^\mu_{\alpha\beta}\dr_\mu f \right),
\end{align*}
where $\Gamma^\mu_{\alpha\beta}$ are the Christoffel symbols. The principal part of this operator is denoted $\tBox_g$, i.e $\tBox_g f = g^{\alpha\beta}  \dr_\alpha\dr_\beta f$.
\end{itemize}
On each slice $\Sigma_t$ we consider the usual function spaces $L^p$ and $W^{k,p}$, by which we always mean $L^p(\Sigma_t)$ and $W^{k,p}(\Sigma_t)$ for some $t$ depending on the context.  We will also use the weighted Sobolev spaces defined below.
\begin{mydef}[Weighted Sobolev spaces]
For $1\leq p < + \infty$, $\delta\in\R$ and $k\in\N$ we define the space $W^{k,p}_\delta$ as the completion of $C^\infty_c$ for the norm
\begin{align*}
\l u\r_{W^{k,p}_\delta} & = \sum_{0\leq |\alpha|\leq k}\l  \langle x  \rangle^{\delta+|\alpha|} \nabla^\alpha u   \r_{L^p},
\end{align*}
where the $L^p$ norm is defined with the Euclidean volume element. We extend this definition to tensors of any type by summing over all components in coordinates.  Some special cases are $H^k_\delta \vcentcolon = W^{k,2}_\delta$ and $L^p_\delta \vcentcolon = W^{0,p}_\delta$.
\end{mydef}

We also define the following $L^\infty$-based spaces.
\begin{mydef}
For $k\in\N$ and $\delta\in\R$ we define $C^k_\delta$ as the completion of $C^\infty_c$ for the norm
\begin{align*}
\l u\r_{C^{k}_\delta} & = \sum_{0\leq |\alpha|\leq k}\l  \langle x  \rangle^{\delta+|\alpha|} \nabla^\alpha u   \r_{L^\infty},
\end{align*}
\end{mydef}

Let us recall some usual facts about these spaces (see \cite{ChoquetBruhat2009} for the proofs).

\begin{prop}\label{prop WSS chap 2}
Let $s,s',s_1,s_2,m\in\N$, $\delta,\delta',\delta_1,\delta_2,\beta\in\R$ and $1\leq p<+\infty$.
\begin{enumerate}
\item[(i)] If $s\leq \min(s_1,s_2)$, $s<s_1+s_2 - \frac{3}{p}$ and $\delta<\delta_1+\delta_2+\frac{3}{p}$ we have the continuous embedding
\begin{align*}
W^{s_1,p}_{\delta_1}\times W^{s_2,p}_{\delta_2} \subset W^{s,p}_\delta.
\end{align*}
\item[(ii)] If $m<s-\frac{3}{p}$ and $\beta\leq\delta + \frac{3}{p}$ we have the continuous embedding
\begin{align*}
W^{s,p}_\delta \subset C^m_\beta.
\end{align*}
\end{enumerate}
\end{prop}

\saut
There won't be any issue of regularity in the time variable so if $f$ is a scalar function defined on $\mathcal{M}$, then $\l f \r_X \leq C$ without any precision on $t$ will always mean
\begin{align*}
\sup_{t\in[0,1]} \l f\r_{X(\Sigma_t)}\leq C,
\end{align*}
where $X$ is one of the function spaces defined above.

\subsubsection{High-frequency notations}\label{section high-frequency notations}

In this section we introduce some notations regarding the oscillating behaviour of all the quantities involved, such as scalar functions, metrics, and more generally tensor fields defined on $\mathcal{M}$. 
\begin{itemize}
\item Such a quantity is said to be high-frequency or oscillating if it admits an expansion in powers of the small parameter $\lambda>0$ with coefficients of the form
\begin{align}
\mathrm{T} \left( \frac{u_0}{\lambda}\right) f \label{T u0 lambda f},
\end{align}
where $f$ is defined on $\mathcal{M}$ and $\mathrm{T}$ is a linear combination of trigonometric functions, i.e elements of
\begin{align}
\enstq{\theta\in \R\longmapsto\sin(k\theta)}{k\in\N} \cup \enstq{\theta\in \R\longmapsto\cos(k\theta)}{k\in\N} .\label{ensemble fcts trigo}
\end{align}
The phase $u_0$ is defined in the next section, as a solution of the eikonal equation for the background metric $g_0$. The phase $\frac{u_0}{\la}$ will be considered as the fundamental harmonic or first harmonic, while the phase $\frac{ku_0}{\la}$ is the $k$-th harmonic.
\item When considering a high-frequency quantity such as a tensor $S$, we denote by $S^{(i)}$ the coefficients of $\lambda^i$ in the expansion defining $S$, which thus expands formally as
\begin{align*}
S = \sum_{i\in\Z}\lambda^i S^{(i)}.
\end{align*}
Note that $S^{(i)}$ is a tensor of the same type as $S$. Moreover,  if $j\in \Z$ we define $S^{(\geq j)}$ by \[S^{(\geq j)} = \sum_{k\geq j} \lambda^{k-j}S^{(k)}.\]  This allows us to clearly truncate high-frequency expansions at a fixed order as in
\begin{align*}
S= \sum_{k\leq j-1}\lambda^k S^{(k)} + \lambda^j S^{(\geq j)}.
\end{align*}
\item In terms of derivation, if $Q(t,x)=\mathrm{T}\pth{\frac{u_0(t,x)}{\la}}q(t,x)$ is an oscillating quantity, then we introduce the two following notations
\begin{align*}
\dr_\theta Q = \mathrm{T}'\pth{\frac{u_0}{\la}}q, \qquad \tilde{\dr}_\a Q = \mathrm{T}\pth{\frac{u_0}{\la}}\dr_\a q.
\end{align*}
In particular, we have
\begin{align*}
\dr_\a Q = \frac{\dr_\a u_0}{\la} \dr_\theta Q + \tilde{\dr}_\a Q.
\end{align*}
\end{itemize}

\subsection{The background spacetime}\label{section BG}

The background quantities are defined on $\M$ and are composed of a metric $g_0$, an optical function $u_0$ and a density $F_0$. They solve the Einstein-null dust system on $\M$:
\begin{align}
R_{\mu\nu}(g_0) & = F_0^2 \dr_\mu u_0 \dr_\nu u_0, \label{Ricci g0}
\\g_0^{-1} (\d  u_0 , \d u_0) & = 0, \label{eq u0}
\\-2L_0 F_0+(\Box_{g_0}u_0)F_0&=0,\label{eq F0}
\end{align}
where $L_0$ is the spacetime gradient of $u_0$, i.e
\begin{align}
L_0 = - g_0^{\alpha\beta}\dr_\alpha u_0 \dr_\beta.\label{def L0}
\end{align}
By \eqref{eq u0}, $L_0$ is null and geodesic, i.e 
\begin{align*}
g_0(L_0,L_0)=0 \quad \text{and} \quad \D_{L_0}L_0=0,
\end{align*}
where $\D$ denotes throughout this article the covariant derivative associated to $g_0$. We also assume that $g_0$ satisfies in the coordinates $(t,x^1,x^2,x^3)$ the standard wave condition
\begin{align}
g_0^{\mu\nu}\Gamma(g_0)^\alpha_{\mu\nu}=0.\label{wave condition g0}
\end{align}
Note that this condition implies that the wave operator associated to $g_0$ reduces to its principal part, i.e $\Box_{g_0}=\tBox_{g_0}$. A standard computation shows that under the wave condition \eqref{wave condition g0} the Einstein type equation \eqref{Ricci g0} rewrites as
\begin{align}
\tBox_{g_0}(g_0)_{\alpha\beta} & = P_{\alpha\beta}(g_0)(\dr g_0,\dr g_0)-2F^2_0\dr_\alpha u_0\dr_\beta u_0,\label{eq g0}
\end{align}
where $P_{\alpha\beta}(g_0)(\dr g_0,\dr g_0)$ is a quadratic non-linearity (see \eqref{quadratic non-linearity} for its exact expression).  We make the following assumptions on the background quantities $(g_0,u_0,F_0)$.  In what follows, $N\geq 10$, $\delta>-\frac{3}{2}$ and $\e>0$ is our smallness threshold.
\begin{itemize}
\item \textbf{Assumptions on the metric} $g_0$. There exists a constant $\e>0$ such that for all $t\in[0,1]$
\begin{align}
\l g_0 - m \r_{H^{N+1}_\delta} + \l \dr_t g_0  \r_{H^{N}_{\delta+1}} + \l \dr_t^2 g_0  \r_{H^{N-1}_{\delta+2}} \leq \e\label{estim g0}
\end{align} 
where $m$ is the Minkowski metric on $\mathcal{M}$. For convenience, we also assume that $\dr_t$ is the unit normal to $\Sigma_0$ of $g_0$, and that the second fundamental form of $\Sigma_0$ is traceless. The former assumption simplifies the construction of the initial data for the high-frequency metric (see Section \ref{section ID spacetime metric}) and the latter is only used in \cite{Touati2023a} to simplify the construction of high-frequency solutions to the constraint equations. The estimates \eqref{estim g0} corresponds to the asymptotically flat setting.
\item \textbf{Assumptions on the optical function} $u_0$. We assume that there exists a constant non-zero vector field $\mathfrak{z}=(\mathfrak{z}_1,\mathfrak{z}_2,\mathfrak{z}_3)$ such that 
\begin{align}
\l \nabla u_0 - \mathfrak{z}  \r_{H^{N}_{\delta+1}}\leq \e,\label{estim u0-z}
\end{align}
where $\nabla u_0 = (\dr_1 u_0, \dr_2 u_0, \dr_3 u_0)$ is the spatial Euclidean gradient of $u_0$. This implies that the level sets of $u_0$ restricted to any $\Sigma_t$ are asymptotically planes in $\R^3$ and that if $\e$ is small enough there exists $c>0$ such that
\begin{align}
\inf_{x\in\R^3}|\nabla u_0|(x) > c.\label{estim inf nabla u0}
\end{align}
This in turn implies that $u_0$ has no critical point. Finally, we assume that $L_0$ defined by \eqref{def L0} is future-directed and since $\dr_t$ is assumed to be the unit normal of $g_0$ to $\Sigma_0$ and $u_0$ solves \eqref{eq u0}, this implies
\begin{align}
\dr_t u_0 = | \nabla u_0|_{\bar{g}_0}\label{dt u_0}
\end{align}
on $\Sigma_0$, where $\bar{g}_0$ is the induced metric on $\Sigma_0$. Finally, we mention another consequence of \eqref{estim inf nabla u0}: with a stationary phase argument we can show that
\begin{align*}
\left| \int_K \mathrm{T}\left(\frac{u_0}{\lambda}\right)\psi \right| \lesssim \lambda \l \psi \r_{W^{1,\infty}(K)},
\end{align*}
where $K$ is any compact subset of $\R^3$, $\mathrm{T}$ is any trigonometric functions and $\psi$ any test function defined on $K$. This in turn implies that the family $\left( \mathrm{T}\left(\frac{u_0}{\lambda}\right)\right)_{\lambda\in(0,\lambda_0]}$ converges weakly to 0 in $L^2(K)$ when $\la$ tends to 0.
\item \textbf{Assumptions on the density} $F_0$.  We assume that $F_0$ is initially compactly supported, i.e there exists $R>0$ such that $F_0\restriction{\Sigma_0}$ is supported in $\{ |x| \leq R \}$. Since $L_0$ is null and geodesic, \eqref{eq F0} implies that 
\begin{align*}
\supp{F_0}\subset J_0^+\left( \{ |x| \leq R \} \right),
\end{align*}
where $J_0^+$ denotes the causal future associated to $g_0$. We can assume for simplicity that 
\begin{align*}
J_0^+\left( \{ |x| \leq R \} \right)\subset \enstq{(t,x)\in\mathcal{M}}{|x|\leq C_{\mathrm{supp}}R},
\end{align*}
for some constant $C_{\mathrm{supp}}>0$. Finally we assume the following estimate:
\begin{align}
\l F_0 \r_{H^N} \leq \e.\label{estim F0}
\end{align}
\end{itemize}
In this article we don't prove the existence of the background solution $(g_0,F_0,u_0)$ on $\mathcal{M}$, it follows from \cite{ChoquetBruhat2006} adapted to the null dust case. We will denote by $C_0$ any numerical constant depending on the background estimates stated here, that is depending on $\delta$, $N$, $R$ or $\mathfrak{z}$.

\saut
In addition to the usual coordinates on $[0,1]\times \R^3$ we will use a null frame attached to the optical function $u_0$, which we define as follows. We denote by $\mathcal{H}_u$ the level sets of $u_0$, i.e
\begin{align*}
\mathcal{H}_u & = \enstq{(t,x)\in\mathcal{M}}{u_0(t,x)=u}.
\end{align*}
Since $u_0$ is a solution of the eikonal equation \eqref{eq u0}, each $\mathcal{H}_u$ is a null hypersurface generated by the geodesic vector field $L_0$. Thanks to \eqref{estim inf nabla u0} they induce a foliation of the spacetime. For $t\in[0,1]$ and $u$ in the image of $u_0$ we define the following 2-surfaces
\begin{align*}
P_{t,u} = \Sigma_t \cap \mathcal{H}_u.
\end{align*}
Thanks to \eqref{estim u0-z}, each $P_{t,u}$ have the topology of a plane in $\R^3$. We denote by $\mathring{g}_0$ the induced metric on $P_{t,u}$. We consider a null vector field $\Lb_0$ such that $g_0(\Lb_0,L_0)=-2$ and an orthonormal frame $(e_1,e_2)$ of $TP_{t,u}$ for $\mathring{g}_0$.  This defines the \textit{background null frame} $(L_0,\Lb_0,e_1,e_2)$, and such a choice is always possible, at least locally. We use the latin upper case letters as indices for the frame $(e_1,e_2)$ on $TP_{t,u}$.  We consider the subset 
\begin{align}
\mathcal{T}_0=\{ L_0, e_1,e_2 \} \label{def T0}
\end{align} 
of the background null frame. The use of such a frame originates in the seminal work of Christodoulou and Klainerman on the stability of Minkowski spacetime (see \cite{Christodoulou1993}) and a concise presentation can be found in \cite{Szeftel2018}.

\saut
Further geometric objects related to the background null frame will be needed for the proof of Theorem \ref{theo main}, linked to the commutator $[L_0,\Box_{g_0}]$. We refer to Lemma \ref{lem commute null} for the main estimate on this commutator and to Appendix \ref{section background null structure} for its proof.

\subsection{Initial data}\label{section initial data}

We reproduce here the main result of \cite{Touati2023a} where we construct high-frequency solutions $(\bar{g}_\lambda,K_\lambda)$ of the constraint equations on $\R^3$, that is
\begin{align}
R(\bar{g}_\lambda)+(\tr_{\bar{g}_\lambda}K_\lambda)^2-|K_\lambda|^2_{\bar{g}_\lambda} & = 0,   \label{hamiltonian constraint general chap 3}
\\ -\dive_{\bar{g}_\lambda} K_\lambda +\d\tr_{\bar{g}_\lambda}K_\lambda & = 0.\label{momentum constraint general chap 3}
\end{align}

\begin{thm}\label{theo initial data}
Let $(g_0,u_0,F_0)$ be the solution of the Einstein-null dust system described in Section \ref{section BG}, and let $\e>0$ be the smallness threshold. There exists $\e_0=\e_0(\delta,R)>0$ such that if $0<\e\leq \e_0$, there exists for all $\lambda\in (0,1]$ a solution $(\bar{g}_\lambda,K_\lambda)$ solution of the constraint equations \eqref{hamiltonian constraint general chap 3}-\eqref{momentum constraint general chap 3} on $\R^3$ of the form
\begin{align}
\bar{g}_\lambda & = \bar{g}_0 + \lambda \cos\left( \frac{u_0}{\lambda}\right) \bar{F}^{(1)} + \lambda^2 \left( \sin\left( \frac{u_0}{\lambda}\right) \bar{F}^{(2,1)} + \cos\left( \frac{2u_0}{\lambda}\right) \bar{F}^{(2,2)} \right) + \lambda^2 \bar{\h}_\lambda,\label{g bar theo constraint}
\\ K_\lambda & = K^{(0)}_\lambda + \lambda K^{(1)}_\lambda + \lambda^2 K^{(\geq 2)}_\lambda, \label{K theo constraint}
\end{align}
with 
\begin{align}
K^{(0)}_\lambda & = K_0 + \half \sin\left( \frac{u_0}{\lambda}\right) |\nabla u_0|_{\bar{g}_0} \bar{F}^{(1)},\label{K0 theo constraint}
\\ \left( K^{(1)}_\lambda\right)_{ij} & = -\half \cos \left( \frac{u_0}{\lambda}\right)\left(  -N_0\bar{F}^{(1)}_{ij} + (\dr_t + N_0)^\rho \Gamma(g_0)_{\rho(i}^k \bar{F}^{(1)}_{j)k} + \frac{1}{2|\nabla u_0|_{\bar{g}_0}} (\Box_{g_0} u_0) \bar{F}^{(1)}_{ij}  \right)\label{K1 theo constraint}
\\&\quad -\half   \left|\nabla u_0 \right|_{\bar{g}_0} \left( \cos\left( \frac{u_0}{\lambda}\right) \bar{F}^{(2,1)}_{ij} - 2\sin\left( \frac{2u_0}{\lambda}\right)\bar{F}^{(2,2)}_{ij}   \right) ,\nonumber
\end{align}
where $K_0$ is the second fundamental form of $\Sigma_0$ for $g_0$ and $N_0=-\frac{\bar{g}_0^{ij}\dr_i u_0 \dr_j}{| \nabla u_0 |_{\bar{g}_0}}$.
Moreover:
\begin{itemize}
\item[(i)] the tensors $\bar{F}^{(1)}$, $\bar{F}^{(2,1)}$ and $\bar{F}^{(2,2)}$ are supported in $\{|x|\leq R\}$ and there exists $C_{\mathrm{cons}}=C_{\mathrm{cons}}(\delta,R)>0$ such that
\begin{align}
\l \bar{F}^{(1)} \r_{H^N} + \l \bar{F}^{(2,1)}  \r_{H^{N-1}} + \l \bar{F}^{(2,2)} \r_{H^{N-1}} & \leq C_{\mathrm{cons}} \e,\label{estim F bar chap 3}
\end{align}
\item[(ii)] the tensor $\bar{F}^{(1)}$ is $\bar{g}_0$-traceless, tangential to $P_{0,u}$ and satisfies
\begin{align}
\left| \bar{F}^{(1)} \right|^2_{\bar{g}_0} = 8 F_0^2 \label{energie condition theo chap 3},
\end{align}
\item[(iii)] the tensors $\bar{\h}_\lambda$ and $K^{(\geq 2)}_\lambda$ belong to the spaces $H^5_\delta$ and $H^4_{\delta+1}$ respectively and satisfy
\begin{align}
\max_{ r\in\llbracket 0,4\rrbracket} \lambda^r \l \nabla^{r+1} \bar{\h}_\lambda \r_{L^2_{\delta+r+1}}&  \leq C_{\mathrm{cons}} \e ,\label{estim h bar chap 3}
\\ \max_{r\in\llbracket 0,4\rrbracket } \lambda^r \l \nabla^r K_\lambda^{(\geq 2)} \r_{L^2_{\delta+r+1}}&  \leq C_{\mathrm{cons}} \e .  \label{estim K geq 2 chap 3}
\end{align}
\end{itemize}
\end{thm}

We make two comments about the formulation of Theorem \ref{theo initial data}, compared to the corresponding statement in \cite{Touati2023a}.
\begin{itemize}
\item The main result of \cite{Touati2023a} is a purely elliptic result, it only requires a given solution of the maximal null dust constraint equations and not a spacetime solution of the Einstein-null dust system, as it may seem from the above formulation. Obviously, by projecting the background spacetime of Section \ref{section BG} on $\Sigma_0$, one recovers the assumptions of \cite{Touati2023a}.
\item Similarly, the expressions of $K^{(0)}_\la$ and $K^{(1)}_\la$ given by \eqref{K0 theo constraint} and \eqref{K1 theo constraint} are equal to the ones given in \cite{Touati2023a}, thanks to some of the assumptions made on the background spacetime in the previous section, namely the fact that $\dr_t$ is the unit normal to $\Sigma_0$ for $g_0$ and the wave condition.
\end{itemize}

\subsection{Local existence in generalised wave gauge} \label{section local existence}

The following theorem is our main result.

\begin{thm}\label{theo main}
Let $(g_0,u_0,F_0)$ be the solution of the Einstein-null dust system described in Section \ref{section BG}, and let $\e>0$ be the smallness threshold.  There exists $\lambda_0>0$ and $\e_0=\e_0 (\delta,R)>0$ such that if $0<\e\leq\e_0$,  there exists for all $\lambda\in (0,\lambda_0]$ a solution $g_\lambda$ of the form 
\begin{align}
g_\lambda = g_0 + \lambda g^{(1)}\left( \frac{u_0}{\lambda} \right) + \lambda^2 g^{(2)}\left( \frac{u_0}{\lambda} \right) + \lambda^2 \tilde{\mathfrak{h}}_\lambda, \label{ansatz theo}
\end{align}
to \eqref{EVE'} on $[0,1]\times \R^3$ in generalised wave gauge.  Moreover:
\begin{enumerate}
\item[(i)] the tensors $g^{(i)}$ for $i=1,2$ are supported in $J_0^+\left( \{ |x| \leq R \} \right)$ and are periodic and smooth functions of the argument $\frac{u_0}{\lambda}$,
\item[(ii)] there exists $C=C(\delta,R)>0$ such that the following estimate holds
\begin{align}
\max_{r\in\llbracket 0,10 \rrbracket} \lambda^r\l\nabla^r g^{(1)} \r_{L^2} + \max_{r\in\llbracket 0,6 \rrbracket} \lambda^r\l\nabla^r g^{(2)} \r_{L^2}  + \max_{r\in\llbracket 1,5 \rrbracket} \lambda^r\l\nabla^{r+1} \tilde{\mathfrak{h}}_\lambda\r_{L^2_{\delta+r}} & \leq C\e,\label{estim théorème}
\end{align}
\item[(iii)] the tensor $g^{(1)}_{\mu\nu}=\cos\left(\frac{u_0}{\lambda}\right) F^{(1)}_{\mu\nu}$ satisfies the following polarization condition, energy condition and transport equation,
\begin{align}
g_0^{\mu\nu} \left( \dr_{\mu}u_0  F^{(1)}_{\sigma\nu} -\half  \dr_\sigma u_0  F^{(1)}_{\mu\nu} \right) & = 0,\label{pola F1 theo}
\\  \frac{1}{8}\left|F^{(1)}\right|^2_{g_0} - \frac{1}{16}\left(\tr_{g_0}F^{(1)}\right)^2 & = F_0^2 ,\label{energie F1 theo}
\\ 2\dr^\alpha u_0 \D_\alpha F^{(1)}_{\mu\nu} + (\Box_{g_0}u_0 ) F^{(1)}_{\mu\nu} & = 0 . \label{eq F1 theo}
\end{align}
\end{enumerate}
\end{thm}

Before we present the strategy behind the proof of Theorem \ref{theo main} in Section \ref{section strategy}, let us make some general remarks.
\begin{itemize}
\item The notation $g^{(i)}\left(\frac{u_0}{\lambda}\right)$ is used in \eqref{ansatz theo} to emphasize the oscillating character of $g^{(i)}$, but these tensors also depends on the spacetime variables and this notation actually stands for the following functions on the spacetime
\begin{align*}
g^{(i)}_{\alpha\beta}\left( \frac{u_0}{\lambda} \right)(t,x) = g^{(i)}_{\alpha\beta}\left( t,x,\frac{u_0(t,x)}{\lambda} \right),
\end{align*}
where $(t,x)\in [0,1]\times \R^3$.
\item The exact expression of the ansatz for $g_\lambda$ is given in Section \ref{section high-frequency metric}. Compared to \eqref{ansatz theo}, we will then describe precisely the oscillating behaviour of each $g^{(i)}$ for $i=1,2$, i.e give a finite frequency decomposition. Note that in Theorem \ref{theo main} we already give the frequency decomposition of $g^{(1)}$, that is $g^{(1)}=\cos\left(\frac{u_0}{\lambda}\right)F^{(1)}$.
\item As explained in depth in the introduction, our main goal is to consider the high-frequency limit, that is $\lambda\to0$. From this perspective, the restriction $\lambda\in(0,\lambda_0]$ with $\lambda_0$ potentially very small is meaningless and only plays a role in the improvement of one of the bootstrap assumptions (see Proposition \ref{prop BA F n+1}).
\item The generalised wave condition defining our gauge has a complicated expression and we choose not to present it in Theorem \ref{theo main}. In particular, it depends on the exact ansatz for $g_\lambda$ and the coupling between the transport equations for the waves and the wave equation for the remainder in \eqref{ansatz theo}. See Section \ref{section GWC et PC} for its exact expression.
\item The tensor $g^{(2)}$ also satisfies polarization conditions similar to \eqref{pola F1 theo} but with RHS depending on $g^{(1)}$. See Section \ref{section GWC et PC} for their exact expressions.
\item The oscillating terms in \eqref{ansatz theo} depend on the phase $\frac{u_0}{\lambda}$, where $u_0$ is the solution of the background eikonal equation, i.e \eqref{eq u0}. We emphasize the fact that we don't solve the eikonal equation for the metric $g_\lambda$,  as in \cite{Huneau2018} but as opposed to \cite{Luk2020}. This would require the construction of a null foliation in 3+1 dimension and hence would complicate our proof. 
\end{itemize}

\section{Strategy of proof}\label{section strategy}

In this section, we sketch the proof of Theorem \ref{theo main}, the challenges of which we already introduced briefly in Section \ref{section exact GO}.

\subsection{The high-frequency ansatz}\label{section strategy a}

In any coordinates system, the Ricci tensor of a Lorentzian metric $g$ reads
\begin{align}
2R_{\alpha\beta}(g)=-\tBox_g g_{\alpha\beta}+H^\rho \dr_\rho g_{\alpha\beta} + g_{\rho (\alpha}\dr_{\beta)}H^\rho+P_{\alpha\beta}(g)(\dr g,\dr g) .\label{Ricci generalised wave coordinates}
\end{align}
We encounter three types of terms in \eqref{Ricci generalised wave coordinates}.
\begin{itemize}
\item $\tBox_g$ stands for the principal part of the wave operator. Because of this term, the Einstein vacuum equations \eqref{EVE'} are quasi-linear. Moreover, if we think of \eqref{EVE'} as a system of second order partial differential equations for the metric coefficients, the $-\tBox_g g_{\alpha\beta}$ term represents the diagonal part of the Ricci operator.
\item The terms involving $H^\rho=g^{\mu\nu}\Gamma^\rho_{\mu\nu}$ are the gauge terms of the Ricci operator and their derivatives in \eqref{Ricci generalised wave coordinates} involve second order derivatives of the metric coefficients. They represent the non-diagonal part of the Ricci operator and working in generalised wave gauge precisely means that we prescribe the value of $H^\rho$ in order to solve the loss of hyperbolicity it causes.
\item The quadratic non-linearity $P_{\alpha\beta}(g)(\dr g,\dr g)$ is given in coordinates by
\begin{align}
P_{\alpha\beta}(g)(\dr g,\dr g) & = g^{\mu\rho}g^{\nu\sigma}\pth{ \dr_{(\alpha}g_{\rho\sigma}\dr_\mu g_{\beta)\nu}    - \frac{1}{2}\dr_\alpha g_{\rho\sigma} \dr_\beta g_{\mu\nu}  - \dr_\rho g_{\alpha\nu}\dr_\sigma g_{\beta\mu} + \dr_\rho g_{\sigma\alpha} \dr_\mu g_{\nu\beta} }. \label{quadratic non-linearity}
\end{align}
It satisfies what we call the weak polarized null condition, see Section \ref{section weak polarized null condition}.
\end{itemize}
While being highly non-linear, the Ricci tensor and the three parts defining it all enjoy some special structures. The goal of this section is to show how we benefit from these various structures to construct $g_\la$, whose precise ansatz is given by
\begin{align}
g_\lambda= g_0 + \lambda g^{(1)} + \lambda^2 \left( g^{(2)} + \h_\lambda \right) + \lambda^3 g^{(3)}.\label{ansatz rough strategy}
\end{align}
The remainder $\h_\lambda$ is a symmetric 2-tensor and each $g^{(i)}$ are oscillating symmetric 2-tensor as defined in Section \ref{section high-frequency notations}. In the rest of this section and for the sake of clarity we drop the index $\lambda$ and only write $g$ and $\h$ instead of $g_\lambda$ and $\h_\lambda$.

\saut 
Without loss of generality, we can assume that $g^{(1)}$ has the form $g^{(1)}_{\alpha\beta}=\cos\left(\frac{u_0}{\lambda}\right) F^{(1)}_{\alpha\beta}$, as stated in Theorem \ref{theo main}. The Ricci tensor being non-linear, we must incorporate the creation of harmonics into $g$.  Even though terms like $P_{\alpha\beta}(g)(\dr g,\dr g) $ are of the form $g^{-1}g^{-1}\dr g \dr g$, thus corresponding to the interaction of four oscillating terms, the derivatives lose $\la$ powers so at leading order in $\la$ the terms $g^{-1}g^{-1}\dr g \dr g$ corresponds actually to an interaction of only two oscillating terms. Therefore, we expect $g^{(2)}$ to oscillate with frequencies $\frac{u_0}{\la}$ and $\frac{2u_0}{\la}$, i.e to contain the first and second harmonic. More precisely, by analogy with our toy model in \cite{Touati2023}, $g^{(2)}$ oscillates as $\sin\pth{\frac{u_0}{\la}}$ and $\cos\pth{\frac{2u_0}{\la}}$.

\saut
Since the Ricci tensor involves at most second order derivatives of the metric, it admits \textit{a priori} a formal expansion of the form
\begin{align}
R_{\mu\nu}(g_\lambda) & = \frac{1}{\lambda}R^{(-1)}_{\mu\nu} +  R^{(0)}_{\mu\nu} + \lambda R^{(1)}_{\mu\nu} + \cdots.\label{expansion Ricci mega rough}
\end{align}
As every geometric optics construction, it all comes down to plugging the high-frequency ansatz \eqref{ansatz rough strategy} in \eqref{EVE'}, obtaining the expression of the different orders in \eqref{expansion Ricci mega rough}, equating them to zero and deduce equations and conditions for each terms in \eqref{ansatz rough strategy}.

\saut
In order to understand the structures at stake, we introduce two key structural notions: admissible harmonics and non-tangential tensors.
\begin{mydef}\label{def admissible frequency}
Let $k=0,1$. The \textbf{admissible harmonics} at the $\lambda^k$ level in the Ricci tensor (i.e in $R^{(k)}_{\mu\nu}$) are $\frac{\ell u_0}{\la}$ for $\ell\in \llbracket 1,k+1\rrbracket$. A non-admissible harmonic at a certain level is said to be forbidden at this level.
\end{mydef}
\begin{mydef}\label{def non-tangential}
Let $S$ be a symmetric 2-tensor and recall \eqref{def T0}.
\begin{enumerate}
\item[(i)] We say that $S$ is \textbf{non-tangential} if $S_{XY}=0$ for all $X,Y\in\mathcal{T}_0$.  In coordinates this is equivalent to the existence of a 1-form $Q$ such that $S_{\alpha\beta}  = \dr_{(\alpha} u_0 Q_{\beta)}$.
\item[(iii)] We say that $S$ is \textbf{fully non-tangential} if the only non-zero null components of $S$ is $S_{\Lb_0\Lb_0}$. In coordinates this is equivalent to the existence of a scalar function $f$ such that $S_{\alpha\beta} = f \dr_\alpha u_0 \dr_\beta u_0$.
\end{enumerate}
\end{mydef}
We also define the polarization tensor of a symmetric 2-tensor, already introduced in \cite{ChoquetBruhat1969}.
\begin{mydef}\label{def pola}
Let $S$ be a symmetric 2-tensor, we define its polarization tensor by
\begin{align}
\Pol_\sigma(S) = g_0^{\mu\nu} \left( \dr_\mu u_0 S_{\sigma\nu} - \half \dr_\sigma u_0 S_{\mu\nu} \right).\label{def polarization tensor}
\end{align}
Its components in the background null frame (see Section \ref{section BG}) are
\begin{align}
\Pol_{L_0}(S) & = - S_{L_0L_0} ,\label{Pol L}
\\ \Pol_A (S) & = - S_{AL_0}, \label{Pol A}
\\ \Pol_{\Lb_0} (S) & = - \delta^{AB} S_{AB}.\label{Pol Lb}
\end{align}
\end{mydef}

\subsubsection{Quasi-linear forbidden harmonics}\label{section quasi-linearity}

We start by looking at the wave operator in the Ricci tensor, also called the wave part in the sequel, that is $-\tBox_g g_{\alpha\beta}$. Let us see how it acts on an oscillating scalar function, we have
\begin{align}
\tBox_g \left( \mathrm{T}\left( \frac{u_0}{\lambda} \right) f \right) & = \mathrm{T}\left( \frac{u_0}{\lambda} \right)\tBox_g f+ \frac{1}{\lambda}\mathrm{T}'\left( \frac{u_0}{\lambda} \right) \left( 2g^{\mu\nu}\dr_\mu u_0 \dr_\nu f + (\tBox_g u_0)f  \right)\label{box fct trigo bis 2}
\\&\quad + \frac{1}{\lambda^2}\mathrm{T}''\left( \frac{u_0}{\lambda} \right) g^{-1}(\d u_0, \d u_0) f,\nonumber
\end{align}
for any scalar function $f$ and $\mathrm{T}$ a trigonometric function. Thanks to \eqref{eq u0} and \eqref{Pol L}, we have schematically
\begin{align}
g^{-1}(\d u_0, \d u_0) = \lambda \Pol_{L_0}\pth{g^{(1)}} + \lambda^2 \pth{ \Pol_{L_0}\pth{g^{(2)}} + \h_{L_0L_0}} + \GO{\lambda^3},\label{terme eikonal strategy}
\end{align}
and
\begin{align*}
\tBox_g \pth{ \la g^{(1)} } & = L_0\dr_\theta g^{(1)} + \Pol_{L_0}\pth{g^{(1)}}\dr_\theta^2 g^{(1)}  + \la \pth{ \Pol_{L_0}\pth{g^{(2)}}\dr_\theta^2 g^{(1)} + \h_{L_0L_0}\dr_\theta^2 g^{(1)}} +  \GO{\la^2},
\\ \tBox_g\pth{ \la^2 g^{(2)} } & = \la \pth{ L_0 \dr_\theta g^{(2)} + \Pol_{L_0}\pth{g^{(1)}}\dr_\theta^2 g^{(2)} }  + \GO{\la^2},
\end{align*}
where we neglected lower order terms. Regrouping terms with respect to their $\la$ power, we obtain the following idealized equations by equating $\tBox_g g$ to 0:
\begin{align}
L_0 \dr_\theta g^{(1)} + \Pol_{L_0}\pth{g^{(1)}}\dr_\theta^2 g^{(1)} & = 0, \label{eq g1 strategy}
\\ L_0 \dr_\theta g^{(2)} + \Pol_{L_0}\pth{g^{(1)}}\dr_\theta^2 g^{(2)}  +  \Pol_{L_0}\pth{g^{(2)}}\dr_\theta^2 g^{(1)} + \h_{L_0L_0}\dr_\theta^2 g^{(1)} & = 0. \label{eq g2 strategy}
\end{align}
We recover here the main feature of geometric optics: the waves $g^{(i)}$ solve transport equations along the rays of the optical function $u_0$. These equations will absorb all the terms oscillating as $g^{(i)}$, i.e containing only admissible harmonics (recall Definition \ref{def admissible frequency}). The terms $\Pol_{L_0}\pth{g^{(i)}}\dr_\theta^2 g^{(k+2-i)}$ in \eqref{eq g1 strategy} and \eqref{eq g2 strategy} for $k=0,1$ contains necessarily forbidden harmonics that can't be absorbed by the transport equations. They are precisely the terms expressing the quasi-linear nature of $\tBox_g g$, i.e the ones coming from \eqref{terme eikonal strategy}. We refer to them as quasi-linear forbidden harmonics in the rest of this section. We will use the structures of the gauge part and the quadratic part of the Ricci tensor to remove them.

\saut
The remainder $\h$ in \eqref{ansatz rough strategy} being non-oscillating, its contribution to the Ricci tensor from the wave part is simply $-\la^2\tBox_g \h$. Therefore, the condition $R^{(\geq 2)}_{\mu\nu}=0$ will require $\h$ to solve a wave equation, absorbing any harmonics appearing at the $\lambda^2$ level. For instance, one of the $\la^2$ contribution of $\tBox_g\pth{ \la^2 g^{(2)} }$ is of the form $\la^2 \tilde{\dr}^2 g^{(2)}$, this oscillating term will appear in the RHS of the wave equation for $\h$
\begin{align}
\tBox_g \h = \tilde{\dr}^2 g^{(2)} + \text{other terms}.\label{eq h basic}
\end{align}
We choose to highlight the term $\tilde{\dr}^2 g^{(2)}$ since it is the source of the loss of one derivative in the coupling with \eqref{eq g2 strategy} (where $\h_{L_0L_0}$ appears), see Section \ref{section bad coupling strat} below.

\saut
The role played by $\h$, i.e absorbing every terms in $R^{(\geq 2)}_{\mu\nu}$, explains why $g^{(3)}$ is different than $g^{(1)}$ and $g^{(2)}$. Indeed the transport term $\la^2L_0\dr_\theta g^{(3)}$ (coming from $\tBox_g\pth{ \la^3 g^{(3)} } $) will simply appear in the RHS of \eqref{eq h basic}. Therefore, $g^{(3)}$ won't be defined as the solution of a transport equation, unlike $g^{(1)}$ and $g^{(2)}$. See Section \ref{section gauge part strategy} below.

\subsubsection{Semi-linear forbidden harmonics}\label{section weak polarized null condition}

We now turn our attention to the semi-linear terms, i.e the quadratic non-linearity $P_{\alpha\beta}(g)(\dr g,\dr g)$ in the Ricci tensor given by \eqref{quadratic non-linearity}. This term plays a key role in our construction, since it is the source of backreaction and the cause of new forbidden harmonics. To understand this, let us first look at its $\la^0$ contribution to the Ricci tensor. Schematically we have
\begin{align}
P_{\alpha\beta}(g)(\dr g,\dr g) & = P_{\alpha\beta}(g_0)(\dr g_0,\dr g_0) +  g_0^{-2}\dr g_0 \dr u_0 \dr_\theta g^{(1)} +  g_0^{-2}\pth{ \dr u_0 \dr_\theta g^{(1)} }^2  + \GO{\lambda}.\label{P0 strategy}
\end{align}
The first term in \eqref{P0 strategy} will simply be identified with the corresponding term in the background equation \eqref{eq g0}. The second term is linear in $\dr_\theta g^{(1)}$, and thus can be added to the RHS of \eqref{eq g1 strategy} since it corresponds to an admissible harmonic. The most crucial term in \eqref{P0 strategy} is the third one. Since $g^{(1)} = \cos\pth{\frac{u_0}{\la}}F^{(1)}$, it schematically expands as
\begin{align}
g_0^{-2}\pth{ \dr u_0 \dr_\theta g^{(1)} }^2 & = g_0^{-2}\pth{ \dr u_0 F^{(1)} }^2 + \cos\pth{\frac{2u_0}{\la}} g_0^{-2}\pth{ \dr u_0 F^{(1)} }^2.\label{P0 cross terms}
\end{align}
The first term \eqref{P0 cross terms} is non-oscillating and thus must be absorbed by the background equation \eqref{eq g0}. Looking at this equation, the only possibility is to identify it with $-2F_0^2\dr_\a u_0 \dr_\b u_0$, which explains why $P_{\alpha\beta}(g)(\dr g,\dr g)$ is the source of backreaction. The oscillating term in \eqref{P0 cross terms} is a forbidden harmonic and for this reason cannot be absorbed by the transport equation \eqref{eq g1 strategy}. We refer to this term as the semi-linear forbidden harmonic, and as for the quasi-linear forbidden harmonic $\Pol_{L_0}\pth{g^{(1)}}\dr_\theta^2g^{(1)}$ in \eqref{eq g1 strategy}, it needs to be remove by the gauge terms.

\saut
Similarly, the $\la^1$ contribution to the Ricci tensor of $P_{\alpha\beta}(g)(\dr g,\dr g)$ contains among others the term $g_0^{-2}\dr u_0 \dr_\theta g^{(1)} \dr u_0 \dr_\theta g^{(2)}$. It is a forbidden harmonic and cannot be absorbed by the transport equation \eqref{eq g2 strategy}. It is also a semi-linear forbidden harmonic and must be removed by gauge terms as the quasi-linear forbidden harmonics $\Pol_{L_0}\pth{g^{(1)}}\dr_\theta^2g^{(2)}$ and $\Pol_{L_0}\pth{g^{(2)}}\dr_\theta^2g^{(1)}$ in \eqref{eq g2 strategy}.

\saut
Let us understand the structure of the semi-linear forbidden harmonics 
\begin{align}
g_0^{-2}\pth{ \dr u_0 \dr_\theta g^{(1)} }^2 \quad \text{and} \quad g_0^{-2}\dr u_0 \dr_\theta g^{(1)} \dr u_0 \dr_\theta g^{(2)}.\label{semi-linear forbidden harmonics}
\end{align}
For this, we introduce the following quadratic form acting on symmetric 2 tensors
\begin{align*}
\mathcal{Q}_{\alpha\beta}(T,S) & = g_0^{\mu\rho}g_0^{\nu\sigma} \pth{ \dr_{(\alpha}u_0 T_{\rho\sigma}\dr_\mu u_0 S_{\beta)\nu}  - \frac{1}{2}\dr_\alpha u_0 T_{\rho\sigma} \dr_\beta u_0 S_{\mu\nu}  - \dr_\rho u_0 T_{\alpha\nu}\dr_\sigma u_0 S_{\beta\mu} }.
\end{align*}
This expression is obtained from \eqref{quadratic non-linearity} by replacing the contractions with $g$ by contractions with $g_0$ and replacing the $\dr g \dr g$ by $\dr u_0 T \dr u_0 S$ (the fourth term in \eqref{quadratic non-linearity} does not contribute because of \eqref{eq u0}). It is such that the problematic terms in \eqref{semi-linear forbidden harmonics} are precisely given by $\mathcal{Q}_{\a\b}\pth{ \dr_\theta g^{(1)} , \dr_\theta g^{(1)} }$ and $\mathcal{Q}_{\a\b}\pth{ \dr_\theta g^{(1)} , \dr_\theta g^{(2)} }+ \mathcal{Q}_{\a\b}\pth{ \dr_\theta g^{(2)} , \dr_\theta g^{(1)} }$ respectively. A cumbersome computation leads to
\begin{align}
\mathcal{Q}_{\a\b}(T,T) & =   \mathcal{E}(T,T) \dr_\alpha u_0 \dr_\beta u_0   +  g_0^{\nu\sigma} \dr_{(\alpha}u_0 \Pol_\sigma(T)  T_{\beta)\nu}  -   \Pol_{\a}(T)   \Pol_{\b}(T) , \label{2Q(T,T)}
\\ \mathcal{Q}_{\a\b}(T,S) + \mathcal{Q}_{\a\b}(S,T) & =  2 \mathcal{E}(T,S) \dr_\alpha u_0 \dr_\beta u_0   + g_0^{\nu\sigma} \dr_{(\alpha}u_0 \pth{ \Pol_\sigma(T)  S_{\beta)\nu} + \Pol_\sigma(S)  T_{\beta)\nu}   }\label{Q(T,S)+Q(S,T)}
\\&\quad  -   \Pol_{(\a}(T)   \Pol_{\b)}(S)  ,    \nonumber
\end{align}
where
\begin{align}
\mathcal{E}(T,S)= \frac{1}{4}  (\tr_{g_0}T) ( \tr_{g_0}S) -\frac{1}{2}  |T\cdot S|_{g_0}\label{def energie E}.
\end{align}
Since $L_0$ is a null vector field for $g_0$, the nature of the quadratic form $\mathcal{Q}_{\a\b}$ can be understood in the context of null conditions. The standard null condition has been introduced in \cite{Klainerman1986,Christodoulou1986} as a criterion for global existence in the small data regime for solutions to quadratic wave systems in 3+1 dimensions. In our context, it would translate to $\mathcal{Q}_{\a\b}(T,T)=0$. In \cite{ChoquetBruhat2000}, Choquet-Bruhat introduces the polarized null condition, which here reads
\begin{align*}
\Pol(T)=0 \Longrightarrow \mathcal{Q}_{\a\b}(T,T)=0.
\end{align*}
As it can be seen on \eqref{2Q(T,T)}, the quadratic form $\mathcal{Q}_{\a\b}$ does not satisfy neither the standard null condition nor the polarized null condition. However, it does satisfy what we call the weak polarized null condition:
\begin{align*}
\Pol(T)=0 \Longrightarrow \mathcal{Q}_{\a\b}(T,T) \; \text{is fully non-tangential},
\end{align*}
where we recall Definition \ref{def non-tangential}. For $\mathcal{Q}_{\a\b}(T,S) + \mathcal{Q}_{\a\b}(S,T)$, the weak polarized null condition reads
\begin{align*}
\Pol(T)=0 \Longrightarrow \mathcal{Q}_{\a\b}(T,S) + \mathcal{Q}_{\a\b}(S,T) \; \text{is non-tangential}.
\end{align*}
More precisely, if $\Pol(T)=0$ then
\begin{align}
\mathcal{Q}_{\a\b}(T,T) & =   \mathcal{E}(T,T) \dr_\alpha u_0 \dr_\beta u_0   , \label{2Q(T,T) bis}
\\ \mathcal{Q}_{\a\b}(T,S) + \mathcal{Q}_{\a\b}(S,T) & =  2 \mathcal{E}(T,S) \dr_\alpha u_0 \dr_\beta u_0   + g_0^{\nu\sigma} \dr_{(\alpha}u_0 \Pol_\sigma(S)  T_{\beta)\nu} .   \label{Q(T,S)+Q(S,T) bis}
\end{align}

\begin{remark}
The expression for $\mathcal{E}\left( T ,  T \right)$ does not seem to have a sign in the general case, but it does if $T$ satisfies $\Pol(T)=0$. Indeed, in this case, we can show that
\begin{align*}
\mathcal{E}\left( T ,  T \right) & = -\half \delta^{AB}\delta^{CD} T_{AC}T_{BD},
\end{align*} 
where the letters $A$, $B$, $C$ and $D$ refer to vectors of the orthonormal basis $(e_1,e_2)$ of $TP_{t,u}$. In this article, the energy $\mathcal{E}(T,T)$ is only used for tensors satisfying $\Pol(T)=0$, see in particular Section \ref{section algebraic properties F1}. 
\end{remark}

\begin{remark}
The weak polarized null condition is linked to the weak null condition identified by Lindblad and Rodnianski in their proof of the stability of Minkowski in wave coordinates. We postpone the discussion of this aspect to Section \ref{section remark null}.
\end{remark}

\subsubsection{The gauge term leads to polarization conditions}\label{section gauge part strategy}

In the last two sections, we showed how the quasi-linear aspects of the wave operator and the absence of the null condition for the quadratic non-linearity in the Ricci tensor cause forbidden harmonics to appear in $R^{(0)}_{\mu\nu}$ and $R^{(1)}_{\mu\nu}$. In this section, we show how the gauge terms in the Ricci tensor can deal with these forbidden harmonics. 

\saut
The gauge term $H^\rho$ rewrites
\begin{align*}
H^\rho = g^{\mu\nu}g^{\rho\sigma}\left(  \dr_\mu g_{\sigma\nu} - \half \dr_\sigma g_{\mu\nu}  \right).
\end{align*}
Its main contribution to the Ricci tensor of the oscillating metric $g$ is then given by
\begin{align}
g_{\rho (\alpha}\dr_{\beta)}H^\rho & = \frac{1}{\lambda} \dr_{(\alpha} u_0 \Pol_{\beta)}\left(\dr_\theta^2 g^{(1)} \right) +  \dr_{(\alpha} u_0 \Pol_{\beta)}\left(\dr_\theta^2 g^{(2)} \right)  + \lambda \dr_{(\alpha} u_0 \Pol_{\beta)}\left(\dr_\theta^2 g^{(3)} \right) 
\\&\quad + \GO{\la^2},\nonumber
\end{align}
where we neglected any lower order terms (recall also Definition \ref{def pola}). The fact that the gauge term manifests itself as polarization tensors in the Ricci tensor is the main reason why we introduced them in the first place, and why we expressed both the quasi-linear and semi-linear forbidden harmonics in terms of the polarization tensors of the different waves.

\saut
If we collect the expansions of the wave part and quadratic part of the Ricci tensor from the previous sections we obtain schematically an expression for the oscillating part of the Ricci tensor (we neglect the non-oscillating terms in $R^{(0)}_{\mu\nu}$ and the $\h_{L_0L_0}$ in $R^{(1)}_{\mu\nu}$):
\begin{align*}
R^{(-1)}_{\mu\nu} & = \dr_{(\mu} u_0 \Pol_{\nu)}\left(\dr_\theta^2 g^{(1)} \right) ,
\\ R^{(0)}_{\mu\nu} & = L_0\dr_\theta g^{(1)} + \Pol_{L_0}\pth{g^{(1)}} \dr^2_\theta g^{(1)}  + \text{admissible harmonic} 
\\&\quad + \dr_{(\mu} u_0 \Pol_{\nu)}\left(\dr_\theta^2 g^{(2)} \right) + \text{semi-linear forbidden harmonic} ,
\\ R^{(1)}_{\mu\nu} & = L_0 \dr_\theta g^{(2)} + \Pol_{L_0}\pth{g^{(1)}}\dr_\theta^2 g^{(2)}  +  \Pol_{L_0}\pth{g^{(2)}}\dr_\theta^2 g^{(1)} + \text{admissible harmonic}   
\\&\quad + \dr_{(\mu} u_0 \Pol_{\nu)}\left(\dr_\theta^2 g^{(3)} \right) + \text{semi-linear forbidden harmonic} .
\end{align*}
From the requirement $R^{(-1)}_{\mu\nu}=0$, we deduce that $\Pol\pth{g^{(1)}}=0$. This polarization condition was already identified in \cite{ChoquetBruhat1969}, and it implies that the physically significant part of $g^{(1)}$ satisfies the TT-gauge assumptions. For our construction, it has three major consequences. First, the quasi-linear forbidden harmonics in $R^{(0)}_{\mu\nu}$ vanishes. Second, the semi-linear forbidden harmonic in $R^{(0)}_{\mu\nu}$ are now fully non-tangential (recall \eqref{2Q(T,T) bis}). This implies that we can absorb them by imposing a polarization condition on $g^{(2)}$:
\begin{align*}
\Pol \pth{ g^{(2)} } = \mathrm{const.}\times \mathcal{E}(T,T) \d u_0.
\end{align*}
This condition implies in turn that the quasi-linear forbidden harmonic in $R^{(1)}_{\mu\nu}$ vanishes (this would not hold if $Q(T,T)$ were only non-tangential). Finally, $\Pol\pth{g^{(1)}}=0$ implies that the semi-linear forbidden harmonic in $R^{(1)}_{\mu\nu}$ are now non-tangential (recall \eqref{Q(T,S)+Q(S,T) bis}). We can thus absorb them by imposing a polarization condition on $g^{(3)}$. Here, we used twice the following simple fact: the equation $\dr_{(\mu}u_0 \Pol_{\nu)}(T)=U_{\mu\nu}$ is solvable if and only if $U_{\mu\nu}$ is non-tangential.

\saut
Assuming that these various polarization conditions hold, we now have
\begin{align*}
R^{(-1)}_{\mu\nu} & = 0,
\\ R^{(0)}_{\mu\nu} & = L_0\dr_\theta g^{(1)}  + \text{admissible harmonic} ,
\\ R^{(1)}_{\mu\nu} & = L_0 \dr_\theta g^{(2)}  + \text{admissible harmonic} .  
\end{align*}
We then obtain $R^{(0)}_{\mu\nu} =R^{(1)}_{\mu\nu} =0$ by defining $g^{(1)}$ and $g^{(2)}$ as the solutions of transport equations along the rays with appropriate RHS. In particular, since the admissible harmonic in $R^{(0)}_{\mu\nu}$ are linear terms in $g^{(1)}$ (see \eqref{P0 strategy}), the transport equation for $g^{(1)}$ is linear. As explained in depth in the introduction, this is known as transparency in the geometric optics literature. As shown here, transparency is made possible by the particular structure of the quadratic non-linearity and the presence of gauge terms. The link between transparency and the standard null condition has already been discussed in \cite{Lannes2013}. Here we see how it survives with a weaker form of the standard null condition with the help of gauge invariance.

\subsubsection{The gauge term leads to a generalised wave condition}\label{section wave gauge strategy}

The gauge term $H^\rho$ and its main contribution $g_{\rho (\alpha}\dr_{\beta)}H^\rho$ to the Ricci tensor are also responsible for a more standard loss of hyperbolicity for the Einstein vacuum equations. In our case, this loss of hyperbolicity affects the wave equation for the remainder $\h$. Very schematically, it reads
\begin{align}
\tBox_g \h = \dr \h \dr \h + \dr^2\h, \label{eq h ill-posed}
\end{align}
where we neglected all contributions from the oscillating parts of \eqref{ansatz rough strategy}, and singled out the second derivatives of $\h$ coming from $g_{\rho (\alpha}\dr_{\beta)}H^\rho$. To remove $\dr^2\h$ from \eqref{eq h ill-posed} and get a well-posed equation, we need to set to zero the corresponding parts of $H^\rho$, i.e the $\dr \h$ in $H^\rho$. Moreover, as the discussion of Section \ref{section bad coupling strat} will show, some terms in $g^{(2)}$ are actually coupled to $\h$ and we also need to set their contribution to $H^\rho$ to zero. The precise definition of the parts of $H^\rho$ we want to cancel is postponed to Section \ref{section gauge part}.

\subsubsection{Propagation of the polarization conditions}\label{section propagation strategy}

As explained in Section \ref{section gauge part strategy}, we need the waves $g^{(1)}$ and $g^{(2)}$ to satisfy both transport equations along the rays and polarization conditions. In \cite{ChoquetBruhat1969}, Choquet-Bruhat proves that the polarization condition for $g^{(1)}$, that is $\Pol\left( g^{(1)} \right)=0$, is propagated by the transport equation $g^{(1)}$ satisfies, which schematically reads $L_0g^{(1)}=g^{(1)}$. By commuting $L_0$ and $\Pol$, she deduced a transport equation for $\Pol\pth{g^{(1)}}$, allowing the propagation of $\Pol\left( g^{(1)} \right)=0$ from the data. Unfortunately, the transport equations and polarization conditions satisfied by $g^{(2)}$ (see Sections \ref{section bg system} and \ref{section GWC et PC}) are much more complicated than the one $g^{(1)}$ satisfies and the previous strategy seems unfeasible. Instead of deriving the equations for $\Pol\left( g^{(2)} \right)$ directly from the equations that $g^{(2)}$ satisfies, we treat the polarization conditions for $g^{(2)}$ as gauge conditions. This means that we first solve the hierarchy of equations as if they were satisfied, and then derive the desired transport equations from the contracted Bianchi identities, which states that the Einstein tensor of a Lorentzian metric is divergence free.  This shows how we recover the polarization conditions on the whole spacetime, since our choice of initial data will be such that they hold initially (see Corollary \ref{coro ID}).

\subsubsection{A remark on null conditions}\label{section remark null}

In this cultural section, we intend to motivate the designation "weak polarized null condition" and to make the link with the weak null condition of Lindblad and Rodnianski. This condition is introduced in \cite{Lindblad2003} for general non-linear wave systems with quadratic non-linearity. It is also proved there that the Einstein vacuum equations in wave coordinates do satisfy this condition. It plays a crucial role in the proof of the stability of Minkowski spacetime in wave coordinates, see \cite{Lindblad2010}.

\saut
Since we consider the same system of equations as Lindblad and Rodnianski, namely \eqref{EVE'} in a similar gauge, their weak null condition is satisfied here. Actually, the term $\mathcal{E}\left( T ,  T \right) \dr_\alpha u_0  \dr_\beta u_0$ in \eqref{2Q(T,T)} precisely corresponds to the problematic term $P(\dr_\mu h,\dr_\nu h)$ from \cite{Lindblad2010} (see in particular Proposition 3.1 there). However, the phenomena observed in the present article, that is a non-linear effect (backreaction) together with linear propagation (transparency), seem to single out the Einstein vacuum equations in (generalised or not) wave coordinates among systems satisfying Lindblad and Rodnianski's weak null condition. Indeed, consider the following system
\begin{equation}\label{toy system}
\left\{
\begin{aligned}
\Box \ffi & = (\dr_t \psi)^2,
\\ \Box \psi & = Q(\dr \ffi, \dr \ffi),
\end{aligned}
\right.
\end{equation}
where $\ffi$ and $\psi$ are scalar functions defined on $\R^{3+1}$, $\Box$ is the flat wave operator and $Q$ satisfies the standard null condition. Let us perform a geometric optics expansion for $\ffi$ and $\psi$, in the spirit of the present article. Consider functions of the form
\begin{align}
\ffi = \ffi_0 + \lambda \ffi^{(1)}\left( \frac{v}{\lambda} \right) \quad \text{and} \quad \psi = \psi_0 + \lambda \psi^{(1)}\left( \frac{v}{\lambda} \right),\label{ansatz ffi psi}
\end{align}
where $v$ solves the flat eikonal equation (say $v=t-x$ or $v=t-r$), and $\ffi^{(1)}$ and $\psi^{(1)}$ also depends on the variable $(t,x)$. The ansatz \eqref{ansatz ffi psi} defines approximate solutions of the system \eqref{toy system} up to terms of order $\lambda$ if and only if the following hierarchy of equations hold
\begin{align*}
\Box \ffi_0 & = (\dr_t \psi_0)^2 + (\dr_t v)^2 \pi \left( \dr_\theta \psi^{(1)} \right)^2,
\\ \Box \psi_0 &  = Q(\dr \ffi_0, \dr\ffi_0),
\\ \Ll_v \dr_\theta \ffi^{(1)} & = 2 \dr_t\psi_0 \dr_t v \dr_\theta \psi^{(1)} + (\dr_t v)^2 \left( \mathrm{Id}-\pi\right) \left( \dr_\theta \psi^{(1)} \right)^2,
\\ \Ll_v \dr_\theta \psi^{(1)} & = 2Q(\dr\ffi_0,\dr v) \dr_\theta \ffi^{(1)},
\end{align*}
where $\pi$ denotes the mean operator with respect to the oscillating variable and $\Ll_v=2\dr^\mu v\dr_\mu + \Box v$ is the transport operator along the rays. Since $(\ffi_0,\psi_0)$ does not satisfy \eqref{toy system}, we observe backreaction (which is of no surprise since \eqref{toy system} doesn't satisfy the standard null condition). We don't observe transparency since the transport system for $\pth{ \ffi^{(1)},\psi^{(1)}}$ is non-linear. However, it is well known that the system \eqref{toy system} does satisfy the weak null condition. This simple example explains in particular the need in our work for a new designation "weak polarized null condition", where the adjective "polarized" refers to the fact that transparency is obtained through polarization conditions for the waves.

\begin{remark}
Note that the specificity of the Einstein vacuum equations is even more evident in other gauges than wave gauges. For instance, in double null gauges or in Coulomb type gauges, \eqref{EVE'} satisfy equivalents of the null condition. This is crucially used in seminal works such as \cite{Christodoulou1993} or \cite{Klainerman2015}. 
\end{remark}

\subsection{Solving the hierarchy of equations}\label{section strategy b}

In this section, we explain how we prove well-posedness for the hierarchy of equations derived from putting the ansatz \eqref{ansatz rough strategy} into the Einstein vacuum equations \eqref{EVE'}. As explained above, the hierarchy of equations is made of transport equations for $g^{(1)}$ and $g^{(2)}$ and a wave equation for $\h$. As we will see, well-posedness would be trivial if not for the quasi-linear nature of the wave operator in the Ricci tensor.

\subsubsection{The background transport equations}

As explained above, the transport equation for $g^{(1)}$ is linear and reads schematically
\begin{align}
L_0  g^{(1)} & = g^{(1)}.\label{eq g1}
\end{align}
The transport equation for $g^{(2)}$ reads schematically
\begin{align}
L_0  g^{(2)} & = g^{(2)} + \tilde{\dr}^2 g^{(1)},\label{eq g2}
\end{align}
where in the RHS of \eqref{eq g2} we only consider the worst term from the regularity point of view. In particular, we neglected the term $\h_{L_0L_0}$ from \eqref{eq g2 strategy}. The system \eqref{eq g1}-\eqref{eq g2} is triangular and solving it causes no problem. In the sequel, the exact form of this system will be called the \textit{background system} because the regularity of its solution only depends on the background spacetime and does not require any bootstrap argument.

\begin{remark}\label{remark N geq 10}
The threshold for the background regularity, i.e $N\geq 10$ (see Section \ref{section BG}), is chosen so that all these background quantities and their derivatives can be ultimately bounded in $L^\infty$. Let us explain briefly why the value $10$ fills this requirement. If the background metric coefficients are in $H^{N+1}$ (we neglect the weights at spacelike infinity in this discussion), then \eqref{eq g1} and \eqref{eq g2} imply that $g^{(1)}\in H^N$ and $g^{(2)}\in H^{N-2}$. Second derivatives of $g^{(2)}$ will appear as a source term in the wave equation for the remainder $\h$. In order to prove local well-posedness for this equation we will differentiate it 4 times (see \eqref{estim broad h et F} below) and our worse term is thus $\tilde{\dr}^6 g^{(2)}$. We want to bound it in $L^\infty$ so the usual Sobolev embedding $H^2\xhookrightarrow{}L^\infty$ implies that we need $N-2\geq 8$.
\end{remark}

\subsubsection{The quasi-linear coupling and the loss of derivatives}\label{section bad coupling strat}

In this section, we discuss the major impact of the term $\h_{L_0L_0}$ in \eqref{eq g2 strategy}, which has been neglected so far since it causes no formal issues and thus does not enter the discussion on forbidden harmonics. Recall that the presence of this term as source term in the transport equation for $g^{(2)}$ is directly linked to the quasi-linear aspect of $\tBox_g g$. 

\saut
To make things clearer, we assume that $\h_{L_0L_0}$ is the only source term in the transport equation for $g^{(2)}$, which we can thus write as $g^{(2)}_{\alpha\beta} = \sin\left( \frac{u_0}{\lambda}\right) \F_{\alpha\beta} $. Similarly, we assume that $\Box_g \F$ is the only source term in the wave equation for $\h$ (this term was already singled out in \eqref{eq h basic}). The coupled system between $\F$ and $\h$ then reads
\begin{align}
L_0 \F & = \h,\label{eq F strat}
\\ \Box_g \h & = \sin\left(\frac{u_0}{\lambda}\right)\Box_g \F,\label{eq h strat}
\end{align}
where we simply write $\h$ instead of $\h_{L_0L_0}\dr^2_\theta g^{(1)}$. Performing energy estimates on \eqref{eq F strat} implies schematically that $\dr\F\sim\dr\h$ while the energy estimate for \eqref{eq h strat} gives $\dr\h\sim \dr^2\F$. Therefore, there is an \textit{a priori} loss of one derivative in the coupling between \eqref{eq F strat} and \eqref{eq h strat}. Regaining this derivative is the main analytical challenge of our work.

\subsubsection{Regaining the missing derivative}\label{section regaining the missing derivative}

The main idea is to show that $\Box_g \F$ is better than $\dr^2\F$, i.e better than any second order derivatives of $\F$. To show this, we would like to derive a transport operator for $\Box_g \F$ from \eqref{eq F strat}. However we can't hope for any structural property of $[L_0,\Box_g]$ since $L_0$ is a null vector for $g_0$ and not $g$. The idea is to first benefit from the expansion defining $g$ and write schematically 
\begin{align}
\Box_g \F & = \Box_{g_0}\F + \lambda \dr^2 \F. \label{decomposition cruciale}
\end{align} 
We improve the first term in \eqref{decomposition cruciale} by considering the commutator $[L_0,\Box_{g_0}]$ for which we can now hope for structural properties. Formally we have from \eqref{eq F strat}
\begin{align}
L_0 \left( \Box_{g_0}\F \right) & = [L_0,\Box_{g_0}]\F + \Box_{g_0}\h.\label{eq box F strat}
\end{align}
We need to know what are the second derivatives appearing in the commutator $[L_0,\Box_{g_0}]$.  This is the content of the following key lemma.
\begin{lem}\label{lem commute null}
Let $f$ be compactly supported, we have
\begin{align}
\l [L_0,\Box_{g_0}]f \r_{L^2} \leq C(C_0)\left( \l \dr L_0 f \r_{L^2} + \l \Box_{g_0}f \r_{L^2} + \l \dr f \r_{L^2} \right).\label{estim commute null}
\end{align}
Moreover, if $r\geq 1$ we have
\begin{align}
\l \nabla^r [L_0,\Box_{g_0}]f \r_{L^2} \leq C(C_0)\left( \l \nabla^r  \dr L_0 f \r_{L^2} + \l \nabla^r \Box_{g_0}f \r_{L^2} + \l \dr f \r_{H^r} + \l \dr^2 f \r_{H^{r-1}} \right).  \label{estim commute null nabla r}
\end{align}
\end{lem}
The proof of this lemma is postponed to Appendix \ref{section background null structure} and is based on the background null foliation induced by the optical function $u_0$.  In the case of equation \eqref{eq box F strat}, the term $\Box_{g_0}\F$ from \eqref{estim commute null} is treated with Gronwall's inequality and the term $ \dr L_0 \F$ rewrites as $\dr \h$ thanks to \eqref{eq F strat}. By using again the expansion for $g$, we can use equation \eqref{eq h strat} for the term $\Box_{g_0}\h$ in \eqref{eq box F strat}. Therefore, we are able to show that $\Box_{g_0}\F$ is at the level of $\dr\h$, which is consistent with \eqref{eq h strat}. This is done in detail in Section \ref{section estim box g0 F}.

\saut
It remains to improve the second term in \eqref{decomposition cruciale}. As our notation suggests it, we don't use any structure of the second order derivatives involved but rather benefit from the small factor $\lambda$ in front. The idea is to cancel one derivative with this factor, which rigorously involves a spatial Fourier projector $\Pi_\leq$, that is a Fourier multiplier supported in $\left\{ |\xi|\leq \frac{1}{\lambda} \right\}$. If we also define $\Pi_\geq=\mathrm{Id}-\Pi_\leq$, we have indeed Bernstein-like estimates in $L^2$:
\begin{align*}
\lambda\dr\Pi_\leq (f) \sim  f  \quad \text{and} \quad \frac{1}{\lambda}\Pi_\geq(f) \sim \dr f.
\end{align*}
The proof of these estimates and the exact definition of the projectors $\Pi_\leq$ and $\Pi_\geq$ are postponed to Appendix \ref{section proof lem commute bis}. In order to benefit from these projectors, we need to introduce them into the system \eqref{eq F strat}-\eqref{eq h strat}. Instead of this system, we thus solve
\begin{align}
L_0 \F & = \Pi_\leq ( \h) ,\label{eq F strat bis}
\\ \Box_g \h & = \sin\left(\frac{u_0}{\lambda}\right)\Box_g \F   + \frac{1}{\lambda} \Pi_\geq ( \h)  .\label{eq h strat bis}
\end{align}
Using the two Bernstein-like estimates,  equation \eqref{eq F strat bis} now implies that $\lambda\dr^2\F \sim \dr \h$ while the extra term in \eqref{eq h strat bis} is absorbed by the energy estimate for the wave operator. This modification of the system is only possible because \eqref{eq F strat}-\eqref{eq h strat} are obtained as part of the high-frequency expansion of the Ricci tensor. This allows us to artificially remove $\Pi_\geq ( \h)$ from the transport equation for $\F$ by simply rewriting the $\lambda\h$ term as 
\begin{align*}
 \lambda \Pi_\leq ( \h)  +  \lambda^2 \times \frac{1}{\lambda} \Pi_\geq ( \h) .
\end{align*}
In particular, this modification leaves \eqref{EVE'} invariant. This shows how one can use \eqref{decomposition cruciale} to gain a derivative, and obtain
\begin{align}
\l\dr\nabla^r \h \r_{L^2} + \l\dr\nabla^r \F \r_{L^2} \lesssim \frac{1}{\lambda^r},\label{estim broad h et F}
\end{align}
for $r\in \llbracket 0,4\rrbracket$ (we omitted the weight at spacelike infinity for $\h$). Note that since we modified the equation for $\F$, we need the following commutator estimate to deal with the first term in \eqref{decomposition cruciale}:
\begin{lem}\label{lem commute bis}
Let $u$ and $v$ be two tempered distribution. We have 
\begin{align}
\l \left[ u , \Pi_\leq  \right]\nabla v \r_{L^2} \lesssim  \left( \l \nabla  u \r_{L^\infty}  +   \lambda \l u \r_{H^{\frac{7}{2}}}  \right)\l v \r_{L^2}.\label{commute bis}
\end{align}
\end{lem}
This lemma will be used to estimate $\left[ \Pi_\leq , \Box_{g_0}\right]$. Its proof is postponed to Appendix \ref{section proof lem commute bis} and relies on the Littlewood-Paley decomposition, and we introduce in Section \ref{section littlewood paley} the material needed.

\begin{remark}
Strictly speaking, the introduction of the Fourier projectors $\Pi_\leq$ and $\Pi_\geq$, going from \eqref{eq F strat}-\eqref{eq h strat} to \eqref{eq F strat bis}-\eqref{eq h strat bis}, does not solve all of our issues. It actually replaces the loss of one derivative by the loss of one power of $\la$, since \eqref{eq F strat bis} now implies $\dr^2 \F \sim \frac{1}{\la}\dr h$. Plugging this into the energy estimate for \eqref{eq h strat bis}, this translates as a time of existence of order $\la$, which is far from satisfactory. A careful bootstrap argument is still required to regain this power of $\la$, see Section \ref{section solving reduced system}.
\end{remark}

\subsubsection{The order of precision of the ansatz}

Finally, let us comment on the order of precision of the high-frequency ansatz defining $g$. As \eqref{ansatz rough strategy} shows, the order of precision is 2, since the non-oscillating remainder $\h$ appears first at the order $\lambda^2$ in $g$.  This is similar to the construction of \cite{Huneau2018a} in $\mathbb{U}(1)$ symmetry, but different from the one of \cite{Touati2023}, where we are able to construct high-frequency ansatz of arbitrary precision for a semi-linear toy model.

\saut
As in \cite{Touati2023}, the value of this order of precision is linked to the estimates used to solve the equations at stake.  To make this clearer, let us assume that this value, denoted by $K\in\Z$, is not fixed yet. The transport equations only involve the transport operator $L_0$ which is part of the background and therefore their resolution is independent from $\lambda$. However, we use energy estimates associated to $g$ to solve \eqref{eq h strat bis}. As shown in \cite{Sogge1995} (see also Lemma \ref{lem EE} in this article) this energy estimates reads schematically
\begin{align}
\l \dr u \r_{L^2}(t)  \lesssim  \exp\left(\int_0^t \l\dr g \r_{L^\infty}\right)  \left( \l \dr u \r_{L^2}(0) + \int_0^t \l\tBox_{g_\lambda} u \r_{L^2}\right).\label{EE intro chap 2}
\end{align}
Therefore, in order to construct a solution to \eqref{eq h strat bis} with a time of existence independent of $\lambda$, we need $\l\dr g \r_{L^\infty}$ to be bounded uniformly in the high-frequency limit $\lambda\to 0$, since it appears as coefficient in \eqref{EE intro chap 2}. If we focus on the term $\lambda^K \h$ in $g$, the usual Sobolev embedding $H^2\xhookrightarrow{}L^\infty$ gives 
\begin{align*}
\l \lambda^K \dr \h \r_{L^\infty}\lesssim\lambda^K \l \dr \h \r_{H^2} \lesssim \lambda^{K-2},
\end{align*}
where we used the fact that one loses one power of $\lambda$ for each derivative of $\dr\h$ (see \eqref{estim broad h et F}). This shows why $K\geq 2$ is necessary. The present article then shows that $K=2$ is sufficient to construct local high-frequency solutions to \eqref{EVE'} on $[0,1]\times \R^3$. 

\begin{remark}
As explained in the introduction, constructing similar solutions to \eqref{EVE'} but of greater order of precision would be an interesting problem from the geometric optics point of view. Moreover, it could pave the way for the study of the long time dynamics of high-frequency gravitational waves. Indeed, the main result of \cite{Touati2023} shows on a toy model how the vector field method and the Klainerman-Sobolev inequality can be applied to geometric optics approximation but require $K\geq 4$.
\end{remark}

\subsection{Outline of the proof}

We give here an outline of the article's remainder which is concerned with the proof of Theorem \ref{theo main} and follows the strategy depicted in the previous sections. 
\begin{itemize}
\item In Section \ref{section high-frequency ansatz}, we give the final expression of the ansatz for $g$ and compute its Ricci tensor.
\item In Section \ref{section hierarchy}, we give the hierarchy of equations as well as the polarization and generalised wave gauge conditions. The hierarchy is composed of the background system and the reduced system. We also define the initial data and provide basic estimates.
\item In Section \ref{section solving bg system}, we solve the background system.
\item In Section \ref{section solving reduced system}, we solve the reduced system. 
\item In Section \ref{section solving EVE}, we show that the metric $g$ is actually a solution of \eqref{EVE'} by propagating the polarization and generalised wave gauge conditions.
\end{itemize}

\section{The high-frequency ansatz}\label{section high-frequency ansatz}

\subsection{High-frequency expansion of the metric}\label{section high-frequency metric}

The metric we consider is given by
\begin{align}
g_\lambda= g_0 + \lambda g^{(1)} + \lambda^2 \left( g^{(2)} + \h_\lambda \right) + \lambda^3  g^{(3)}\label{g} ,
\end{align}
where
\begin{align}
g^{(1)} & = F^{(1)}\cos\left( \frac{u_0}{\lambda}\right),\label{expression g1}
\\ g^{(2)} & =  \F_\lambda\sin\left( \frac{u_0}{\lambda}\right) + F^{(2,1)}\sin\left( \frac{u_0}{\lambda}\right) + F^{(2,2)}\cos\left( \frac{2u_0}{\lambda}\right).\label{expression g2}
\end{align}
The symmetric 2-tensors $F^{(1)}$ and $F^{(2,i)}$ for $i=1,2$ are called the background perturbations. They solve the background system, defined in Section \ref{section bg system}. The symmetric 2-tensors $\F_\lambda$ and $\h_\lambda$ solve the reduced system, defined in Section \ref{section reduced system}. We use a different font to emphasize their special role. The exact definition of $g^{(3)}$ will be given in Section \ref{section def g4}, let us just say for now that it is a polynomial in terms of the background perturbations and their derivatives and is independent from $\h_\lambda$ but does depends on $\F_\lambda$ without derivatives. The background perturbations tensors don't depend on $\lambda$, unlike the tensors $\F_\lambda$, $\h_\lambda$ and the metric $g_\lambda$ itself. However, for the sake of clarity we drop from now on this notation and only write $\F$, $\h$ and $g$.

\begin{remark}
The presence of $\F$ and $F^{(2,1)}$ in $g^{(2)}$ seems redundant and is not necessary. However it seems useful to distinguish between the part of the $\frac{u_0}{\la}$ oscillation in $g^{(2)}$ coupled to $\h$ (see Sections \ref{section bad coupling strat} for a informal discussion) and the part whose only purpose is to absorb terms coming from $g^{(1)}$ such as $\Box_{g_0}F^{(1)}$ and which can thus be treated as a background quantity.
\end{remark}

We can formally compute the high-frequency expansion of the inverse of $g$. In coordinates it is given by
\begin{align}
g^{\mu\nu} & = g_0^{\mu\nu} + \lambda(g^{\mu\nu})^{(1)} + \lambda^2(g^{\mu\nu})^{(2)} + \lambda^3 H^{\mu\nu},\label{inverse g}
\end{align}
where 
\begin{align}
(g^{\mu\nu})^{(1)} & = - ( F^{(1)})^{\mu\nu}\cos\left( \frac{u_0}{\lambda}\right),\label{inverse g1}
\\(g^{\mu\nu})^{(2)} & = g_0^{\mu\rho} (F^{(1)})^{\nu\sigma} F^{(1)}_{\rho\sigma} \cos^2\left( \frac{u_0}{\lambda}\right)  - \F^{\mu\nu}\sin\left( \frac{u_0}{\lambda}\right)  - \h^{\mu\nu}  \label{inverse g2}
\\&\quad - ( F^{(2,1)})^{\mu\nu}\sin\left( \frac{u_0}{\lambda}\right) - ( F^{(2,2)})^{\mu\nu}\cos\left( \frac{2u_0}{\lambda}\right). \nonumber 
\end{align}
On the RHS of these expressions, all the inverses are with respect of the background metric $g_0$. For example we have $( F^{(1)})^{\mu\nu} = g_0^{\mu\rho}g_0^{\nu\sigma}F^{(1)}_{\rho\sigma}$.  The quantity $H^{\mu\nu}$ is defined so that \eqref{inverse g} holds with $(g^{\mu\nu})^{(i)}$ for $i=1,2$ being defined by \eqref{inverse g1}-\eqref{inverse g2}, and one can verify by hand that $H^{\mu\nu}=\GO{1}$.

\begin{remark}
We denote by $g^{-k}$ any product of $k$ coefficients of the inverse metric $g^{-1}$. This also applies to the background inverse metric $g_0^{-1}$.
\end{remark}

We conclude this section by computing the Christoffel symbols of the metric $g$. They are given by
\begin{align}
\Gamma^\rho_{\mu\nu} & = \Gamma(g_0)^\rho_{\mu\nu} + (\tilde{\Gamma}^{(0)})^\rho_{\mu\nu} + \GO{\lambda},\label{christoffel}
\end{align}
where 
\begin{align}
(\tilde{\Gamma}^{(0)})^\rho_{\mu\nu} = -\half \sin\left( \frac{u_0}{\lambda}\right) g_0^{\rho\sigma} \left( \dr_\mu u_0 F^{(1)}_{\nu\sigma}  +  \dr_\nu u_0 F^{(1)}_{\mu\sigma}   - \dr_\sigma u_0 F^{(1)}_{\mu\nu} \right).\label{gammatilde0}
\end{align}

\subsection{High-frequency expansion of the Ricci tensor}\label{section Ricci}

In this section, we compute the Ricci tensor of the metric $g$ given by \eqref{g}. In order to lighten the computations, we assume that $F^{(1)}$ satisfies
\begin{align}
F^{(1)}_{L_0 \a}=0 \quad \text{and} \quad \tr_{g_0}F^{(1)}=0. \label{assumptions on F1}
\end{align}
These are very reasonable assumptions, since they correspond to the polarization condition $\Pol\pth{F^{(1)}}=0$ already identified in \cite{ChoquetBruhat1969} and discussed in Section \ref{section gauge part strategy}, together with the stronger requirement $F^{(1)}_{L_0\Lb_0}=0$. The fact that $F^{(1)}$ satisfies these conditions will be proved in Section \ref{section propagation ordre 1}.

\saut
The expansion of the Ricci tensor is based on \eqref{Ricci generalised wave coordinates}, i.e its decomposition in generalised wave coordinates, and we first expand the wave part $\tBox_g g_{\alpha\beta}$, then the quadratic non-linearity $P_{\alpha\beta}$ and finally the gauge part $H^\rho $, in Propositions \ref{prop expansion wave part}, \ref{prop expansion quad part} and \ref{prop expansion gauge part} respectively. Finally, Proposition \ref{prop expansion ricci tensor} gives the expansion of the Ricci tensor of the oscillating metric $g$. The computations leading to these propositions are postponed to Appendix \ref{appendix calcul}.

\subsubsection{The wave part}\label{section wave part}

We start by expanding the wave part $\tBox_gg_{\a\b}$ of the Ricci tensor of $g$.

\begin{prop}\label{prop expansion wave part}
The wave part of the Ricci tensor of $g$ admits the expansion
\begin{align}
\tBox_g g_{\alpha\beta} & = W^{(0)}_{\alpha\beta} + \lambda W^{(1)}_{\alpha\beta} + \lambda^2 W^{(\geq 2)}_{\alpha\beta},   \label{expansion wave part}
\end{align}
where 
\begin{align}
W^{(0)}_{\alpha\beta} & = \tBox_{g_0}(g_0)_{\alpha\beta} - \sin\left( \frac{u_0}{\lambda} \right) \pth{ -2L_0 F^{(1)}_{\alpha\beta} + (\tBox_{g_0} u_0)F^{(1)}_{\alpha\beta}  },\label{W0}
\\ W^{(1)}_{\alpha\beta} & = \cos\left(\frac{u_0}{\lambda}\right) \pth{ -2 L_0F^{(2,1)}_{\alpha\beta} + (\tBox_{g_0}u_0)F^{(2,1)}_{\alpha\beta} +  \tilde{W}^{(1,1)}_{\alpha\beta}   -2 L_0\F_{\alpha\beta} + (\tBox_{g_0}u_0)\F_{\alpha\beta}   + \h_{L_0L_0}F^{(1)}_{\alpha\beta} }  \label{W1}
\\&\quad -2  \sin\left(\frac{2u_0}{\lambda}\right) \pth{  -2 L_0F^{(2,2)}_{\alpha\beta} + (\tBox_{g_0}u_0)F^{(2,2)}_{\alpha\beta} +  \tilde{W}^{(1,2)}_{\alpha\beta}  } \nonumber
\\&\quad +\frac{1}{2} \sin\left(\frac{2u_0}{\lambda}\right)\F_{L_0L_0}F^{(1)}_{\alpha\beta} + \cos\left(\frac{u_0}{\lambda}\right)\cos\left(\frac{2u_0}{\lambda}\right) F^{(2,2)}_{L_0L_0} F^{(1)}_{\alpha\beta} \nonumber ,
\\ W^{(\geq 2)}_{\a\b} & = \tBox_g \h_{\alpha\beta} +  \sin\left( \frac{u_0}{\lambda} \right)\tBox_g \F_{\alpha\beta}  + \lambda \tBox_g g^{(3)}_{\alpha\beta} +  \Tilde{W}^{(\geq 2)}_{\alpha\beta} \label{W2} ,
\end{align}
with lower order terms given schematically by
\begin{align}
\tilde{W}^{(1,1)} & = g_0^{-2} F^{(1)} \dr^2 g_0 +  g_0^{-2} \dr^2 F^{(1)} + g_0^{-2} (F^{(1)})^2 \dr^2 u_0, \label{Wtilde 11}
\\ \tilde{W}^{(1,2)} & = g_0^{-2} F^{(1)} \dr^2 g_0 +  g_0^{-2} \dr^2 F^{(1)} + g_0^{-2} (F^{(1)})^2 \dr^2 u_0 + g_0^{-1}(\dr u_0)^2 F^{(2,1)} F^{(1)}, \label{Wtilde 12}
\\ \tilde{W}^{(\geq 2)} & = (g^{-1})^{(\geq 3)} (\dr u_0)^2 F^{(1)} \label{Wtilde geq 2}
\\&\quad + (g^{-1})^{(\geq 2)}\left( \dr^2 g_0 + \dr^2 u_0 F^{(1)} + \dr u_0 \pth{ \dr F^{(1)} +  \dr F^{(2,1)} +  \dr \F } + (\dr u_0)^2 \pth{ F^{(2,i)} + \F} \right) \nonumber
\\&\quad + (g^{-1})^{(\geq 1)} \left(  \dr^2 F^{(1)} + \dr^2 u_0 \pth{ F^{(2,i)} +\F }  \right)  + g^{-1} \dr^2 F^{(2,i)}.\nonumber
\end{align}
\end{prop}

\begin{remark}
Let us briefly comment on these various expressions, recalling the formal discussion of Section \ref{section quasi-linearity}.
\begin{itemize}
\item The absence of the quasi-linear forbidden harmonic $\frac{2u_0}{\la}$ in $W^{(0)}_{\a\b}$ coming from $F^{(1)}_{L_0L_0}F^{(1)}_{\a\b}$ is due to \eqref{assumptions on F1}.
\item The quasi-linear forbidden harmonic $\frac{3u_0}{\la}$ in $W^{(1)}_{\a\b}$ is due to $F^{(2)}_{L_0L_0}$, and will eventually vanish thanks to polarization conditions.
\end{itemize}
\end{remark}

\subsubsection{The quadratic non-linearity}

In this section we expand $P_{\alpha\beta}(g)(\dr g, \dr g)$ (see \eqref{quadratic non-linearity} for an exact expression). 

\begin{prop}\label{prop expansion quad part}
The quadratic non-linearity in the Ricci tensor of $g$ admits the expansion
\begin{equation}
P_{\alpha\beta}(g)(\dr g,\dr g) = P^{(0)}_{\alpha\beta} + \lambda P^{(1)}_{\alpha\beta} + \lambda^2 P^{(\geq 2)}_{\alpha\beta}, \label{expansion P}
\end{equation}
where
\begin{align}
P^{(0)}_{\alpha\beta} & = P_{\alpha\beta}(g_0)(\dr g_0, \dr g_0) - \frac{1}{4}\left| F^{(1)} \right|^2_{g_0} \dr_\alpha u_0 \dr_\beta u_0 \label{P0}
\\&\quad -\sin\pth{ \frac{u_0}{\la}} \pth{ -2 L_0^\rho \Gamma(g_0)^\nu_{(\alpha \rho} F^{(1)}_{\beta)\nu} + \dr_{(\alpha}u_0\hat{P}^{(0,1)}_{\beta)} }  +\frac{1}{4} \cos\pth{\frac{2u_0}{\la}}  \left| F^{(1)} \right|^2_{g_0} \dr_{\alpha} u_0  \dr_{\beta} u_0   \nonumber,
\\  P^{(1)}_{\alpha\beta} & = \cos\pth{\frac{u_0}{\la}}\pth{ -2L_0^\rho \Gamma(g_0)^\nu_{(\alpha\rho}\pth{  \F_{\nu\beta)} +  F^{(2,1)}_{\nu\beta)}} + P^{(1,1)}_{\a\b} + \dr_{(\alpha}u_0 \hat{P}^{(1,1)}_{\beta)} } \label{P1}
\\&\quad + \sin\pth{\frac{2u_0}{\la}}\pth{ 4L_0^\rho \Gamma(g_0)^\nu_{(\alpha\rho}  F^{(2,2)}_{\nu\beta)}  +P^{(1,2)}_{\alpha\beta} +  \dr_{(\alpha}u_0  \hat{P}^{(1,2)}_{\beta)} } + \cos\left( \frac{3u_0}{\lambda}\right)\dr_{(\alpha} u_0  \hat{P}_{\beta)}^{(1,3)}  \nonumber,
\end{align}
with lower order terms given schematically by
\begin{align}
\hat{P}^{(0,1)} & = g_0^{-2}F^{(1)} \dr g_0 \label{Phat0},
\\ P^{(1,k)} & = \left( g_0^{-3}F^{(1)}\dr g_0  +  g_0^{-2}  \dr F^{(1)}  \right)\left( \dr g_0  +  \dr u_0 F^{(1)}\right)  \label{P1k},
\\ \hat{P}^{(1,i)} & = g_0^{-2} \left( \dr g_0   +  \dr u_0 F^{(1)} \right) \left( \F + F^{(2,k)}\right) +\dr u_0 g_0^{-3}(F^{(1)})^3,\label{Phat 1i}
\\ P^{(\geq 2)} & = (g^{-2}\dr g \dr g)^{(\geq 2)},
\end{align}
for $k=1,2$ and $i=1,2,3$.
\end{prop}

\begin{remark}
Let us briefly comment on these various expressions, recalling the formal discussion of Section \ref{section weak polarized null condition}.
\begin{itemize}
\item Thanks to $\Pol\pth{F^{(1)}}=0$ (which follows from \eqref{assumptions on F1}), the semi-linear forbidden harmonic $\frac{2u_0}{\la}$ in $P^{(0)}_{\a\b}$ is fully non-tangential. See \eqref{2Q(T,T) bis}, where $\mathcal{E}\pth{F^{(1)},F^{(1)}}$ is simplified thanks to $\tr_{g_0}F^{(1)}=0$ from \eqref{assumptions on F1}.
\item Again thanks to $\Pol\pth{F^{(1)}}=0$, the semi-linear forbidden harmonic $\frac{3u_0}{\la}$ in $P^{(1)}_{\a\b}$ is non-tangential, see \eqref{Q(T,S)+Q(S,T) bis}.
\end{itemize}
\end{remark}

\subsubsection{The gauge part}\label{section gauge part}

In this section we expand the gauge term $H^\rho$. As explained in Sections \ref{section gauge part strategy} and \ref{section wave gauge strategy}, it is responsible for polarization conditions for each wave in \eqref{g} and for a generalised wave gauge condition linked to the equation for the remainder $\h$. The generalised wave gauge is equivalent to setting a group of problematic terms to 0, and these terms are grouped in $\Upsilon^\rho$ below.

\begin{prop}\label{prop expansion gauge part}
The gauge term $H^\rho$ admits the expansion
\begin{align}
H^\rho & =  \lambda  (H^{(1)})^\rho  + \lambda^2 \left( (H^{(2)})^\rho + \Upsilon^\rho \right)  + \lambda^3  (H^{(\geq 3)})^\rho,  \label{Hrho formel}.
\end{align}
where
\begin{align}
(H^{(1)})^\rho & =  \cos\left( \frac{u_0}{\lambda} \right) \left( g_0^{\rho\sigma} \Pol_\sigma \left( F^{(2,1)} \right) + g_0^{\rho\sigma} \Pol_\sigma\left(\F\right) + (\Tilde{H}^{(1,1)})^\rho \right) \label{H1} 
\\&\quad + \sin\left( \frac{2u_0}{\lambda} \right) \left( -2 g_0^{\rho\sigma} \Pol_\sigma \left( F^{(2,2)} \right)    -\frac{1}{4}g_0^{\rho\sigma} \dr_\sigma u_0 \left| F^{(1)} \right|^2_{g_0} \right)\nonumber,
\\ (H^{(2)})^\rho & = g_0^{\rho\sigma} \Pol_\sigma \left( \dr_\theta g^{(3)} \right) + (\tilde{H}^{(2)})^\rho\label{H2},
\\  \Upsilon^\rho & = g^{\mu\nu}g^{\rho\sigma} \left( \dr_\mu \h_{\sigma\nu} - \half \dr_\sigma \h_{\mu\nu} + \sin\left(\frac{u_0}{\lambda}\right) \left( \dr_\mu \F_{\sigma\nu} - \half \dr_\sigma \F_{\mu\nu}  \right) +  \lambda  \left( \tilde{\dr}_\mu g^{(3)}_{\sigma\nu} - \half \tilde{\dr}_\sigma g^{(3)}_{\mu\nu}  \right)  \right)\label{def Upsilon h}
\\&\quad - g_0^{\rho\sigma} \h^{\mu\nu} \left( \dr_\mu (g_0)_{\sigma\nu} - \half \dr_\sigma (g_0)_{\mu\nu} - \sin\left(\frac{u_0}{\lambda}\right)  \left(   \dr_\mu u_0 F^{(1)}_{\sigma\nu} - \half \dr_\sigma u_0 F^{(1)}_{\mu\nu} \right) \right),\nonumber
\end{align}
with lower order terms given schematically by
\begin{align}
\Tilde{H}^{(1,1)} & = g_0^{-2} \left( \dr F^{(1)} + g_0^{-1} F^{(1)} \dr g_0 \right),\label{Htilde 11}
\\ \tilde{H}^{(2)} & = g_0^{-3}\left((1+g_0^{-1}) (F^{(1)})^2 + F^{(2,i)} + \F \right) (\dr g_0 + \dr u_0 F^{(1)} ) \label{Htilde 2}
\\&\quad +  g_0^{-3}F^{(1)}(\dr F^{(1)} + \dr u_0 F^{(2,i)} + \dr u_0 \F) + g_0^{-2} \dr F^{(2,i)}\nonumber ,
\\ H^{(\geq 3)} & = (g^{-2})^{(\geq 3)} \left( \dr g_0 + \dr u_0 F^{(1)} \right) \label{H geq 3}
\\&\quad + (g^{-2})^{(\geq 2)} \left( \dr F^{(1)} + \dr u_0 F^{(2,i)} + \dr u_0 \F \right) + (g^{-2})^{(\geq 1)} \dr F^{(2,i)}\nonumber
\\&\quad + g^{-1}(g^{-1})^{(\geq 1)} \dr u_0 \dr_\theta g^{(3)} . \nonumber
\end{align}
\end{prop}

\begin{remark}
Concretely, $\Upsilon^\rho$ contains all derivatives of $\h$ and $\F$ present in $H^\rho$. Derivatives of $\F$ are also problematic since the analysis of the equations will show that $\dr\F$ is at the level of $\dr \h$. This also explains the presence of the $\tilde{\dr} g^{(3)}$ terms since $g^{(3)}$ will involve $ \F$. The second line in \eqref{def Upsilon h} is also linked to $g^{(3)}$ and will be fully explained in Remark \ref{remark chelou}.
\end{remark}

For clarity, we denote by $\mathring{H}$ the part of $H$ containing only allowed terms from the wave equation perspective, that is
\begin{align}
\mathring{H}^\rho = H^\rho - \lambda^2 \Upsilon^\rho.\label{def H rond}
\end{align}

\subsubsection{Expression of the Ricci tensor}

We conclude this section by putting together Propositions \ref{prop expansion wave part}, \ref{prop expansion quad part} and \ref{prop expansion gauge part}. Before that, we define the main transport operator
\begin{align}
\Ll_0 \vcentcolon = -2 \D_{L_0} + \Box_{g_0}u_0, \label{def Ll_0}
\end{align}
where we recall that $\D$ is the covariant derivative associated to $g_0$. It acts on tensor fields of all type, including scalar functions.

\begin{prop}\label{prop expansion ricci tensor}
Under the assumptions \eqref{assumptions on F1}, the Ricci tensor of $g$ admits the expansion
\begin{align*}
R_{\a\b} = R^{(0)}_{\a\b} + \la R^{(1)}_{\a\b} + \la^2 R^{(\geq 2)}_{\a\b}.
\end{align*}
The term $R^{(0)}_{\a\b}$ is given by
\begin{align}
2R^{(0)}_{\a\b} & = - \tBox_{g_0}(g_0)_{\alpha\beta}  +P_{\alpha\beta}(g_0)(\dr g_0, \dr g_0) - \frac{1}{4}  \left| F^{(1)} \right|^2_{g_0}  \dr_\alpha u_0 \dr_\beta u_0  \label{expression Ricci0}
\\&\quad + \sin\pth{\frac{u_0}{\la}} \pth{  \Ll_0 F^{(1)}_{\alpha\beta}  - \dr_{(\alpha}u_0\pth{  \Pol_{\b)} \left( F^{(2,1)} \right) +  \Pol_{\b)}\left(\F\right) + \Tilde{H}^{(1,1)}_{\b)} +\hat{P}^{(0,1)}_{\beta)} } } \nonumber
\\&\quad - \cos\pth{\frac{2u_0}{\la}} \dr_{(\a}u_0 \pth{   4 \Pol_{\b)} \left( F^{(2,2)} \right)    + \frac{3}{8}  \dr_{\b)} u_0 \left| F^{(1)} \right|^2_{g_0}    } \nonumber .
\end{align}
The term $R^{(1)}_{\a\b}$ is given by
\begin{align}
2R^{(1)}_{\a\b} & =      \cos\pth{\frac{u_0}{\la}}\bigg(  -\Ll_0F^{(2,1)}_{\alpha\beta}  -  \tilde{W}^{(1,1)}_{\alpha\beta} +   P^{(1,1)}_{\a\b}  - \Ll_0\F_{\alpha\beta}  - \h_{L_0L_0}F^{(1)}_{\alpha\beta}           \label{expression Ricci1}
\\&\quad \hspace{4cm} +   \D_{(\a} \left(  \Pol \left( F^{(2,1)} \right) +  \Pol\left(\F\right) + \Tilde{H}^{(1,1)} \right)_{\b)} \bigg)\nonumber
\\&\quad + \sin\left(\frac{2u_0}{\lambda}\right) \Bigg(  -2 \Ll_0F^{(2,2)}_{\alpha\beta}  +  2\tilde{W}^{(1,2)}_{\alpha\beta} + P^{(1,2)}_{\alpha\beta}  \nonumber
\\&\quad \hspace{3cm}+\D_{(\a} \left( -2  \Pol \left( F^{(2,2)} \right)    -\frac{1}{4} \d u_0 \left| F^{(1)} \right|^2_{g_0} \right)_{\b)} \Bigg)\nonumber
\\&\quad + \cos\pth{\frac{u_0}{\la}} \dr_{(\alpha}u_0 \hat{P}^{(1,1)}_{\beta)}  + \sin\pth{\frac{2u_0}{\la}}  \dr_{(\alpha}u_0  \hat{P}^{(1,2)}_{\beta)} + \cos\left( \frac{3u_0}{\lambda}\right)\dr_{(\alpha} u_0  \hat{P}_{\beta)}^{(1,3)} \nonumber
\\&\quad -\frac{1}{2} \sin\left(\frac{2u_0}{\lambda}\right)\F_{L_0L_0}F^{(1)}_{\alpha\beta} - \cos\left(\frac{u_0}{\lambda}\right)\cos\left(\frac{2u_0}{\lambda}\right) F^{(2,2)}_{L_0L_0} F^{(1)}_{\alpha\beta}\nonumber
\\&\quad  - \half \sin\pth{\frac{2u_0}{\la}}  \left( g_0^{\rho\sigma} \Pol_\sigma \left( F^{(2,1)} \right) + g_0^{\rho\sigma} \Pol_\sigma\left(\F\right) + (\Tilde{H}^{(1,1)})^\rho \right) \dr_\rho u_0 F^{(1)}_{\a\b} \nonumber
\\&\quad -2 \sin\pth{\frac{u_0}{\la}}  \sin\left( \frac{2u_0}{\lambda} \right)   \Pol_{L_0} \pth{ F^{(2,2)} }   F^{(1)}_{\a\b} \nonumber
\\&\quad +   \dr_{(\a} u_0  \pth{  \Pol_{\b)} \left( \dr_\theta^2 g^{(3)} \right) + (g_0)_{\rho \b)}  \dr_\theta(\tilde{H}^{(2)})^\rho + (g_0)_{\rho \b)}  \dr_\theta(\Upsilon^\rho)^{(0)} } \nonumber
\\&\quad - \half \sin\pth{ \frac{2u_0}{\la}}  F^{(1)}_{\rho(\a} \dr_{\b)} u_0  \left( g_0^{\rho\sigma} \Pol_\sigma \left( F^{(2,1)} \right) + g_0^{\rho\sigma} \Pol_\sigma\left(\F\right) + (\Tilde{H}^{(1,1)})^\rho \right) \nonumber
\\&\quad -4 \cos\pth{\frac{u_0}{\la}} \cos\pth{ \frac{2u_0}{\la}}  F^{(1)}_{\rho(\a} \dr_{\b)} u_0   g_0^{\rho\sigma} \Pol_\sigma \left( F^{(2,2)} \right)  . \nonumber
\end{align}
The term $R^{(\geq 2)}_{\a\b}$ is given by
\begin{align}
2R^{(\geq 2)}_{\a\b} & = - \tBox_g \h_{\alpha\beta} -  \sin\left( \frac{u_0}{\lambda} \right)\tBox_g \F_{\alpha\beta}  - \lambda \tBox_g g^{(3)}_{\alpha\beta} -  \Tilde{W}^{(\geq 2)}_{\alpha\beta} \label{expression Ricci2}
\\&\quad  + P^{(\geq 2)}_{\a\b} + \pth{ H^\rho \dr_\rho g_{\alpha\beta} + g_{\rho (\alpha}\dr_{\beta)}H^\rho }^{(\geq 2)} . \nonumber
\end{align}
\end{prop}

\begin{remark}
As discussed in Section \ref{section gauge part strategy}, the high-frequency expansion of the Ricci tensor \textit{a priori} contains a term of order $\la^{-1}$. This term is zero thanks to $\Pol\pth{F^{(1)}}=0$ which follows again from \eqref{assumptions on F1}.
\end{remark}

\subsection{Towards the hierarchy}\label{section towards the hierarchy}

In this section, we give a first version of the various partial differential equations and algebraic conditions we impose on each tensors in \eqref{g} so that $g$ solves \eqref{EVE'}. We thus deduce them from the requirements that 
\begin{align*}
R^{(0)}_{\a\b}  = 0, \quad R^{(1)}_{\a\b}  = 0 \quad\text{and}\quad R^{(\geq 2)}_{\a\b} = 0.
\end{align*}
The final version will be properly given in Section \ref{section hierarchy}, where the modification of the coupling between $\h$ and $\F$ will be introduced (see Section \ref{section regaining the missing derivative} for where this modification is first discussed).

\subsubsection{Ensuring $R^{(0)}_{\a\b}=0$}\label{section ensuring R0=0}

Canceling the non-oscillating term and the two oscillating terms in \eqref{expression Ricci0} is enough to ensure $R^{(0)}_{\a\b}=0$. Set to zero, the non-oscillating term is a wave equation for the background metric $g_0$ with $F^{(1)}$ as source term
\begin{align}
- \tBox_{g_0}(g_0)_{\alpha\beta}  +P_{\alpha\beta}(g_0)(\dr g_0, \dr g_0) - \frac{1}{4}  \left| F^{(1)} \right|^2_{g_0}  \dr_\alpha u_0 \dr_\beta u_0 = 0. \label{eq g0 bis}
\end{align}
Given the equation \eqref{eq g0} satisfied by $g_0$, we see that \eqref{eq g0 bis} is equivalent to the following quadratic condition for $F^{(1)}$
\begin{align}
  \left| F^{(1)} \right|^2_{g_0} & = 8F_0^2. \label{condition quadratique F1}
\end{align}
Recall also that all the computations of Section \ref{section high-frequency ansatz} (and in particular Proposition \ref{prop expansion ricci tensor}) are performed under the assumptions \eqref{assumptions on F1} on $F^{(1)}$. Canceling the $\sin\pth{\frac{u_0}{\la}}$ term in \eqref{expression Ricci0} is done by imposing the following transport equation for $F^{(1)}$
\begin{align}
\Ll_0 F^{(1)} & = 0,\label{eq F1 first}
\end{align}
and the following polarization conditions on $F^{(2,1)}$ and $\F$
\begin{align}
\Pol \pth{F^{(2,1)}} & = - \tilde{H}^{(1,1)} - \hat{P}^{(0,1)},
\\ \Pol \pth{\F} & = 0.\label{condition Pol F=0}
\end{align}
Canceling the $\cos\pth{\frac{2u_0}{\la}}$ term in \eqref{expression Ricci0} is done by imposing the following polarization conditions on $F^{(2,2)}$
\begin{align*}
\Pol \pth{F^{(2,2)}} & = - \frac{3}{32} \left| F^{(1)} \right|^2_{g_0} \d u_0.
\end{align*}
For convenience, we define the following tensors
\begin{align}
V^{(2,1)} & = \Pol \pth{F^{(2,1)}} + \tilde{H}^{(1,1)} + \hat{P}^{(0,1)},\label{def V 2 1}
\\ V^{(2,2)} & = \Pol \pth{F^{(2,2)}} + \frac{3}{32} \left| F^{(1)} \right|^2_{g_0} \d u_0, \label{def V 2 2}
\end{align}
so that the polarization conditions on $F^{(2,1)}$ and $F^{(2,2)}$ simply rewrite
\begin{align}
V^{(2,1)} = 0 \quad \text{and} \quad V^{(2,2)}=0.\label{condition V=0}
\end{align}

\subsubsection{Ensuring $R^{(1)}_{\a\b}=0$}\label{section ensuring R1=0}

Before deriving equations and conditions from the requirement $R^{(1)}_{\a\b}=0$, let us rewrite $R^{(1)}_{\a\b}$ with the help of the tensors $V^{(2,1)}$ and $V^{(2,2)}$. From \eqref{expression Ricci1}, \eqref{def V 2 1} and \eqref{def V 2 2} we see that
\begin{align}
2R^{(1)}_{\a\b} & =      \cos\pth{\frac{u_0}{\la}}\pth{  -\Ll_0F^{(2,1)}_{\alpha\beta}  -  \tilde{W}^{(1,1)}_{\alpha\beta} +   P^{(1,1)}_{\a\b} -   \D_{(\a} \hat{P}^{(0,1)}_{\b)}  - \Ll_0\F_{\alpha\beta}  - \h_{L_0L_0}F^{(1)}_{\alpha\beta}     }\label{expression Ricci1 bis}
\\&\quad + \sin\left(\frac{2u_0}{\lambda}\right) \Bigg(  -2 \Ll_0F^{(2,2)}_{\alpha\beta}  +  2\tilde{W}^{(1,2)}_{\alpha\beta} + P^{(1,2)}_{\alpha\beta}  + \half  g_0^{\rho\sigma} F^{(1)}_{\rho(\a} \dr_{\b)} u_0 \hat{P}^{(0,1)}_\si \nonumber
\\&\quad \hspace{5cm} - \half  \hat{P}^{(0,1)}_{L_0}  F^{(1)}_{\a\b}  -\frac{1}{16} \D_{(\a} \left(  \d u_0 \left| F^{(1)} \right|^2_{g_0} \right)_{\b)}    \Bigg)\nonumber
\\&\quad +   \dr_{(\a} u_0  \bigg(  \Pol_{\b)} \left( \dr_\theta^2 g^{(3)} \right) + (g_0)_{\rho \b)}  \dr_\theta(\tilde{H}^{(2)})^\rho + (g_0)_{\rho \b)}  \dr_\theta(\Upsilon^\rho)^{(0)} \nonumber
\\&\quad \hspace{3cm} +  \cos\pth{\frac{u_0}{\la}}  \hat{P}^{(1,1)}_{\beta)}  + \sin\pth{\frac{2u_0}{\la}}    \hat{P}^{(1,2)}_{\beta)} + \cos\left( \frac{3u_0}{\lambda}\right)  \hat{P}_{\beta)}^{(1,3)} \bigg)\nonumber
\\&\quad + \cos\pth{\frac{u_0}{\la}}  \D_{(\a} \left( V^{(2,1)} +  \Pol\left(\F\right)  \right)_{\b)} -2\sin\pth{\frac{u_0}{\la}}\D_{(\a}  V^{(2,2)}_{\b)}\nonumber
\\&\quad  + \half \sin\pth{\frac{2u_0}{\la}} \pth{ \left(  V^{(2,1)}_{L_0}  +  2\Pol_{L_0}\left(\F\right)  \right) F^{(1)}_{\a\b} - g_0^{\rho\sigma} \pth{  V^{(2,1)}_\sigma  + \Pol_\sigma\left(\F\right)  }   F^{(1)}_{\rho(\a} \dr_{\b)} u_0    }\nonumber
\\&\quad -4 \cos\pth{\frac{u_0}{\la}} \cos\pth{ \frac{2u_0}{\la}}g_0^{\rho\sigma}  F^{(1)}_{\rho(\a} \dr_{\b)} u_0  V^{(2,2)}_\si \nonumber
\\&\quad + \pth{ \cos\left(\frac{u_0}{\lambda}\right)\cos\left(\frac{2u_0}{\lambda}\right) -2 \sin\pth{\frac{u_0}{\la}}  \sin\left( \frac{2u_0}{\lambda} \right) }V^{(2,2)}_{L_0} F^{(1)}_{\alpha\beta}\nonumber,
\end{align}
where we also used
\begin{align*}
\F_{L_0L_0} = - \Pol_{L_0}(\F) ,
\end{align*}
and
\begin{align*}
F^{(2,2)}_{L_0L_0} = - \Pol_{L_0}\pth{F^{(2,2)}} = - V^{(2,2)}_{L_0},
\end{align*}
which follows from \eqref{def V 2 2}. We see that the terms in \eqref{expression Ricci1 bis} can be put in three categories that are canceled by different mechanisms.
\begin{itemize}
\item The first two lines can be set to zero by imposing transport equations for $F^{(2,1)}$, $F^{(2,2)}$ and $\F$:
\begin{align}
 \Ll_0F^{(2,1)}_{\alpha\beta} & =  -  \tilde{W}^{(1,1)}_{\alpha\beta} +   P^{(1,1)}_{\a\b} -   \D_{(\a} \hat{P}^{(0,1)}_{\b)} , \label{eq F21 first}
 \\   \Ll_0F^{(2,2)}_{\alpha\beta}  & =   \tilde{W}^{(1,2)}_{\alpha\beta} +\half P^{(1,2)}_{\alpha\beta}  + \frac{1}{4}g_0^{\rho\sigma} F^{(1)}_{\rho(\a} \dr_{\b)} u_0 \hat{P}^{(0,1)}_\si - \frac{1}{4} \hat{P}^{(0,1)}_{L_0}  F^{(1)}_{\a\b}  -\frac{1}{32} \D_{(\a} \left(  \d u_0 \left| F^{(1)} \right|^2_{g_0} \right)_{\b)} ,\label{eq F22 first}
 \\  \Ll_0\F_{\alpha\beta} & = - \h_{L_0L_0}F^{(1)}_{\alpha\beta}  . \label{eq F first}
\end{align}
\item The third line is composed of non-tangential terms and can be set to zero by imposing the following polarization condition on $g^{(3)}$:
\begin{align}\label{first def g3}
\Pol_{\b} \left( \dr_\theta^2 g^{(3)} \right) =- (g_0)_{\rho \b}  \dr_\theta(\tilde{H}^{(2)})^\rho  -  \cos\pth{\frac{u_0}{\la}}  \hat{P}^{(1,1)}_{\beta}  - \sin\pth{\frac{2u_0}{\la}}    \hat{P}^{(1,2)}_{\beta} - \cos\left( \frac{3u_0}{\lambda}\right)  \hat{P}_{\beta}^{(1,3)} .
\end{align}
Note that we deliberately leave $(g_0)_{\rho \b)}  \dr_\theta(\Upsilon^\rho)^{(0)}$ aside. Moreover, note that the RHS of \eqref{first def g3} does not depend on $\h$, thanks to the addition of the last line in $\Upsilon^\rho$ (see Remark \ref{remark chelou} for why we want $g^{(3)}$ to be independent from $\h$).
\item The last four lines are proportional to either $V^{(2,1)}$, $V^{(2,2)}$ or $\Pol(\F)$, which will eventually vanish according to the previous discussion on ensuring $R^{(0)}_{\a\b}=0$, see \eqref{condition Pol F=0} and \eqref{condition V=0}.
\end{itemize}

\subsubsection{Ensuring $R^{(\geq 2)}_{\a\b}=0$}\label{section ensuring R2=0}

From \eqref{expression Ricci2}, we see that $R^{(\geq 2)}_{\a\b}=0$ is ensured if we impose the following wave equation for the remainder 
\begin{align}
  \tBox_g \h_{\alpha\beta} &= -  \sin\left( \frac{u_0}{\lambda} \right)\tBox_g \F_{\alpha\beta}  - \lambda \tBox_g g^{(3)}_{\alpha\beta} -  \Tilde{W}^{(\geq 2)}_{\alpha\beta}  + P^{(\geq 2)}_{\a\b} + \pth{ \mathring{H}^\rho \dr_\rho g_{\alpha\beta} + g_{\rho (\alpha}\dr_{\beta)}\mathring{H}^\rho }^{(\geq 2)} ,\label{eq h first}
\end{align}
and the following generalised wave gauge condition
\begin{align*}
\Upsilon^\rho = 0,
\end{align*}
where we recall that $H^\rho = \mathring{H}^\rho + \la^2 \Upsilon^\rho$. 

\begin{remark}
As explained in depth in Section \ref{section regaining the missing derivative}, \eqref{eq F first} and \eqref{eq h first} are not the final version of the equations that $\F$ and $\h$ will solve, see Section \ref{section reduced system} below.
\end{remark}

\section{The hierarchy of equations}\label{section hierarchy}

In this section, we give the final version of the equations that the different terms in \eqref{g} solve. We divide them into two groups of equations, the background system presented in Section \ref{section bg system} and the reduced system presented in Section \ref{section reduced system}. Note that $g^{(3)}$ is defined in Section \ref{section def g4} but doesn't solve a differential equation.  

\subsection{The background system}\label{section bg system} 

The \textit{background system} is the system satisfied by $F^{(1)}$, $F^{(2,1)}$ and $F^{(2,2)}$. According to \eqref{eq F1 first}, \eqref{eq F21 first} and \eqref{eq F22 first} it writes
\begin{align}
\Ll_0F^{(1)}_{\alpha\beta} & =0 \label{eq F1},
\\  \Ll_0F^{(2,1)}_{\alpha\beta} & =  -  \tilde{W}^{(1,1)}_{\alpha\beta} +   P^{(1,1)}_{\a\b} -   \D_{(\a} \hat{P}^{(0,1)}_{\b)} ,\label{eq F21}
\\ \Ll_0F^{(2,2)}_{\alpha\beta}  & =   \tilde{W}^{(1,2)}_{\alpha\beta} +\half P^{(1,2)}_{\alpha\beta}  + \frac{1}{4}g_0^{\rho\sigma} F^{(1)}_{\rho(\a} \dr_{\b)} u_0 \hat{P}^{(0,1)}_\si - \frac{1}{4} \hat{P}^{(0,1)}_{L_0}  F^{(1)}_{\a\b}  -\frac{1}{32} \D_{(\a} \left(  \d u_0 \left| F^{(1)} \right|^2_{g_0} \right)_{\b)} .\label{eq F22}
\end{align}
The background system admits a triangular structure. It will be solved in Section \ref{section solving bg system} and all the estimates and support properties of the background perturbations are given in Theorem \ref{theo bg system}. Since \eqref{eq F1}, \eqref{eq F21} and \eqref{eq F22} are first order transport equation, the initial data for the background system consist in 
\begin{align*}
F^{(1)}_{\alpha\beta} \restriction{\Sigma_0} \quad \text{and} \quad F^{(2,i)}_{\alpha\beta} \restriction{\Sigma_0}.
\end{align*}
They are given in Section \ref{section ID bg system}.

\subsection{The reduced system}\label{section reduced system} 

We present now the system solved by the tensors $\F$ and $\h$, called the \textit{reduced system}. We already derived equations for $\F$ and $\h$ in Sections \ref{section ensuring R1=0} and \ref{section ensuring R2=0}, namely \eqref{eq F first} and \eqref{eq h first}. However, as explained in Section \ref{section regaining the missing derivative}, the coupled system \eqref{eq F first}-\eqref{eq h first} is ill-posed and we need to modify it. We introduce the Fourier cut-offs $\Pi_\leq$ and $\Pi_\geq$.
\begin{mydef}\label{def P}
Let $\chi_1 : \R^3\longrightarrow[0,1]$ be a smooth function supported in \[\enstq{ \xi\in\R^3}{|\xi|\leq 2},\] and such that $\chi_1(\xi)=1$ if $|\xi|\leq 1$. For $\lambda>0$ we define $\chi_\lambda(\xi)=\chi_1(\lambda\xi)$, thus $\chi_\lambda$ is supported in $\enstq{ \xi\in\R^3}{|\xi|\leq \frac{2}{\lambda}}$ and $\chi_\lambda(\xi)=1$ if $|\xi|\leq \frac{1}{\lambda}$. We define $\Pi_\leq$ by
\begin{align*}
\Pi_\leq (f) = \mathcal{F}^{-1} \left( \chi_\lambda \mathcal{F}(f)\right),
\end{align*}
for $f\in L^2$. We also define
\begin{align*}
\Pi_\geq & = \mathrm{Id}-\Pi_\leq.
\end{align*}
\end{mydef}
The reduced system, i.e the equations for $\F$ and $\h$, is then 
\begin{align}
\Ll_0\F_{\alpha\beta} & = - \Pi_\leq \left( \h_{L_0L_0} \right) F^{(1)}_{\alpha\beta} \label{eq F},
\\ \tBox_g \h_{\alpha\beta} & = -  \sin\left( \frac{u_0}{\lambda} \right)\tBox_g \F_{\alpha\beta}  - \lambda \tBox_g g^{(3)}_{\alpha\beta} - \frac{1}{\lambda}\cos\left(\frac{u_0}{\lambda} \right)\Pi_\geq \left(\h_{L_0L_0}  \right)F^{(1)}_{\alpha\beta} + \mathcal{R}(g),\label{eq h}
\end{align}
where we set
\begin{align}
\mathcal{R}(g) & = - \tilde{W}^{(\geq 2)}_{\alpha\beta} + \left( \mathring{H}^\rho \dr_\rho g_{\alpha\beta} + g_{\rho(\alpha} \dr_{\beta)} \mathring{H}^\rho \right)^{(\geq 2)} + P^{(\geq 2)}_{\alpha\beta}.\label{def R(g)}
\end{align}
Since \eqref{eq F} is a first order transport equation and \eqref{eq h} is a second order wave equation the initial data for the reduced system consist in
\begin{align*}
\F_{\alpha\beta}\restriction{\Sigma_0}, \quad \h_{\alpha\beta}\restriction{\Sigma_0} \quad \text{and} \quad \dr_t \h_{\alpha\beta}\restriction{\Sigma_0}.
\end{align*}
They are given in Section \ref{section ID reduced system}.

\subsection{Generalised wave gauge and polarization conditions}\label{section GWC et PC}

As it is usual when solving the Einstein equations, we impose gauge conditions. Since we solve a wave equation for $\h$, we need a generalised wave gauge. Moreover, we also impose polarization conditions on $F^{(1)}$, $F^{(2,i)}$ and $\F$. These various algebraic conditions appeared naturally in the discussions of Sections \ref{section ensuring R0=0}, \ref{section ensuring R1=0} and \ref{section ensuring R2=0} and we recall them for convenience.

\begin{mydef}\label{def generalised wave gauge}
The generalised wave gauge considered in Theorem \ref{theo main} is defined by
\begin{align}
\Upsilon^\rho = 0,\label{Upsilon =0}
\end{align}
where $\Upsilon^\rho$ is defined in \eqref{def Upsilon h}.
\end{mydef}

The purpose of this generalised wave gauge is to remove from the RHS of the wave equation for $\h$ any second order derivatives of $\h$, which would induce a loss of derivatives in the standard energy estimate. For later use, we define $\tilde{\Upsilon}^\rho$ as the terms in $\Upsilon^\rho$ independent from $\dr\h$, i.e
\begin{align}
\tilde{\Upsilon}^\rho & = g^{\mu\nu}g^{\rho\sigma} \left(  \sin\left(\frac{u_0}{\lambda}\right) \left( \dr_\mu \F_{\sigma\nu} - \half \dr_\sigma \F_{\mu\nu}  \right) +  \lambda  \left( \tilde{\dr}_\mu g^{(3)}_{\sigma\nu} - \half \tilde{\dr}_\sigma g^{(3)}_{\mu\nu}  \right)  \right)\label{def Upsilon tilde}
\\&\quad - g_0^{\rho\sigma} \h^{\mu\nu} \left( \dr_\mu (g_0)_{\sigma\nu} - \half \dr_\sigma (g_0)_{\mu\nu} - \sin\left(\frac{u_0}{\lambda}\right)  \left(   \dr_\mu u_0 F^{(1)}_{\sigma\nu} - \half \dr_\sigma u_0 F^{(1)}_{\mu\nu} \right) \right) . \nonumber
\end{align}

The polarization conditions prescribe the polarization tensors of $F^{(1)}$, $F^{(2,i)}$ and $\F$.

\begin{mydef}\label{def polarization conditions}
The polarization conditions are 
\begin{align}
V^{(2,j)}_\beta=0 , \label{Vij=0}
\end{align}
where $V^{(2,j)}$ are defined in \eqref{def V 2 1} and \eqref{def V 2 2}, together with and
\begin{align}
\Pol_\b\pth{ F^{(1)}} & = 0,\label{pola F 1 == 0}
\\ F^{(1)}_{L_0\Lb_0} & = 0,\label{F 1 Lb L == 0}
\\ \Pol_\b(\F) & = 0 . \label{pol F =0}
\end{align}
\end{mydef}

\begin{remark}
Note that \eqref{assumptions on F1} is equivalent to \eqref{pola F 1 == 0}-\eqref{F 1 Lb L == 0}.
\end{remark}

\subsection{Definition of $g^{(3)}$}\label{section def g4}

Unlike the other terms in \eqref{g}, the tensor $g^{(3)}$ doesn't solve any differential equation but only a polarization condition, already given in \eqref{first def g3}, which we recall for convenience:
\begin{align}\label{def g3}
\Pol_{\b} \left( \dr_\theta^2 g^{(3)} \right) =- (g_0)_{\rho \b}  \dr_\theta(\tilde{H}^{(2)})^\rho  -  \cos\pth{\frac{u_0}{\la}}  \hat{P}^{(1,1)}_{\beta}  - \sin\pth{\frac{2u_0}{\la}}    \hat{P}^{(1,2)}_{\beta} - \cos\left( \frac{3u_0}{\lambda}\right)  \hat{P}_{\beta}^{(1,3)} .
\end{align}
Since the RHS of \eqref{def g3} is clearly purely oscillating (note that any $\dr_\theta$-derivative is purely oscillating), it is possible to integrate twice \eqref{def g3} with respect to the oscillating variable and obtain $g^{(3)}$ if and only if we know how to solve an equation of the form
\begin{equation}
\Pol_\beta(T)=\Omega_\beta, \label{eq Pol(T)=Omega}
\end{equation}
with $\Omega_\beta$ a 1-form and for $T$ a symmetric 2-tensor. According to \eqref{Pol L}, \eqref{Pol A} and \eqref{Pol Lb} we just need to set 
\begin{align*}
T_{L_0L_0} & =-\Omega_{L_0},
\\ T_{L_0 A} & = -\Omega_A,
\\ T_{11} & = -\Omega_{\Lb_0},
\end{align*}
with all the other coefficients of $T$ being zero.  This implies that we can solve \eqref{eq Pol(T)=Omega} for $T$ satisfying $|T|\lesssim |\Omega|$ with an implicit constant depending only on the background spacetime.

\saut
Thanks to \eqref{Htilde 2} and \eqref{Phat 1i}, the RHS of \eqref{def g3} depends on the background quantities $g_0$, $F^{(1)}$, $F^{(2,i)}$ and their first derivatives and on the non-background field $\F$ . However, it does not depend on $\h$.  We summarize this discussion in the next lemma.

\begin{lem}\label{lem g3}
Given $F^{(1)}$, $F^{(2,i)}$ and $\F$, there exists a symmetric 2-tensor $g^{(3)}$ solution of \eqref{def g3} of the form
\begin{align*}
g^{(3)} & = \sum_{\mathrm{T}\in\mathcal{A}}\mathrm{T}\left(\frac{u_0}{\lambda}\right) \left( g^{(3,\mathrm{T})}(\F) + g^{(3,\mathrm{T},\mathrm{BG})} \right),
\end{align*}
where $\mathcal{A}$ is a finite subset of trigonometric functions (see \eqref{ensemble fcts trigo}), $g^{(3,\mathrm{T})}(\F)$ satisfies schematically
\begin{align}
 g^{(3,\mathrm{T})}(\F)  = g_0 g_0^{-3}\left( \dr g_0 + \dr u_0 F^{(1)} \right) \F \label{g3T(F)},
\end{align}
and where $g^{(3,\mathrm{T},\mathrm{BG})} $ is a polynomial in the background quantities $g_0$, $F^{(1)}$, $F^{(2,i)}$ and their first derivatives. In the rest of this article, we will use the notation
\begin{align*}
g^{(3)}(\F) & = \sum_{\mathrm{T}\in\mathcal{A}}\mathrm{T}\left(\frac{u_0}{\lambda}\right)  g^{(3,\mathrm{T})}(\F) ,
\\ g^{(3,\mathrm{BG})} & = \sum_{\mathrm{T}\in\mathcal{A}}\mathrm{T}\left(\frac{u_0}{\lambda}\right)  g^{(3,\mathrm{T},\mathrm{BG})} .
\end{align*}
\end{lem}

Note that since $g^{(3)}$ only absorbs oscillating terms, it has the same support as $F^{(1)}$, $F^{(2,i)}$ or $\F$, which will be proved to be compact in this article.

\subsection{Initial data for the spacetime metric}\label{section ID spacetime metric}

In this section, we define the initial data on $\Sigma_0$ for the different tensors appearing in the high-frequency ansatz \eqref{g}. They are based on the solution of the constraint equations presented in Theorem \ref{theo initial data}, which is the main result of \cite{Touati2023a}.

\subsubsection{Initial data for the background system}\label{section ID bg system}

We start with the tensors $F^{(1)}$, $F^{(2,1)}$ and $F^{(2,2)}$.  For the spatial components, we simply set
\begin{align}
F^{(1)}_{ij} \restriction{\Sigma_0} = \bar{F}^{(1)}_{ij}, \quad F^{(2,1)}_{ij} \restriction{\Sigma_0} = \bar{F}^{(2,1)}_{ij}, \quad F^{(2,2)}_{ij}  \restriction{\Sigma_0} = \bar{F}^{(2,2)}_{ij}. \label{ID F ij}
\end{align}
For the time components of $F^{(1)}$, we set 
\begin{align}
F^{(1)}_{0\alpha}\restriction{\Sigma_0}=0\label{ID F1 0},
\end{align}
which is consistent with \eqref{assumptions on F1}. For the time components of $F^{(2,i)}$, we need to ensure that $F^{(2,i)}$ satisfy on $\Sigma_0$ the polarization conditions \eqref{Vij=0}. For clarity, we define the 1-forms $\Omega^{(i)}$ for $i=1,2$ so that 
\begin{align}
 V^{(2,i)}_\beta = \Pol_\beta\left( F^{(2,i)} \right) - \Omega^{(i)}_\beta  .\label{def inter omega}
\end{align}
Note from \eqref{def V 2 1}, \eqref{def V 2 2}, \eqref{Htilde 11} and \eqref{Phat0} that $\Omega^{(i)}$ involve only $g_0$ and its derivatives or $F^{(1)}$ and its first derivatives $\dr F^{(1)}$. For $\nabla F^{(1)}$, the initial data on $\Sigma_0$ are obtained by differentiating directly $F^{(1)}_{\alpha\beta}$ defined in \eqref{ID F ij} and \eqref{ID F1 0} and for $\dr_t F^{(1)}_{\alpha\beta}$ they are obtained by rewriting the equation we want $F^{(1)}$ to satisfy, that is \eqref{eq F1} which rewrites as
\begin{align}
\dr_t F^{(1)}_{\alpha\beta} & = - N_0 F^{(1)}_{\alpha\beta}  +   (\dr_t + N_0)^\rho \Gamma(g_0)^\mu_{\rho (\alpha} F^{(1)}_{\mu\beta)} + \frac{1}{2|\nabla u_0|_{\bar{g}_0}}(\Box_{g_0}u_0)  F^{(1)}_{\alpha\beta},  \label{ID dt F1}
\end{align}
where we used the fact that on $\Sigma_0$ we have $L_0= |\nabla u_0|_{\bar{g}_0} \left( \dr_t + N_0 \right)$. At this stage, the RHS of \eqref{ID dt F1} is fully defined on $\Sigma_0$, which implies that this is also the case for $\Omega^{(i)}$. It remains to define the initial data for $F^{(2,i)}_{0\alpha}$ for $i=1,2$ by
\begin{align}
F^{(2,i)}_{0A} \restriction{\Sigma_0} & = - \bar{F}^{(2,i)}_{AN_0} - \frac{1}{|\nabla u_0|_{\bar{g}_0}} \Omega^{(i)}_A \restriction{\Sigma_0},\label{ID F2 0A }
\\   F^{(2,i)}_{0N_0} \restriction{\Sigma_0}   & = -\half \bar{F}^{(2,i)}_{N_0N_0}  -  \frac{1}{2|\nabla u_0|_{\bar{g}_0}^2} \Omega^{(i)}_{L_0} \restriction{\Sigma_0},\label{ID F2 0N}
\\ F^{(2,i)}_{00} \restriction{\Sigma_0} & = 0.\label{ID F2 00}
\end{align}
Note that these definitions only ensure 
\begin{align}
V_{\mathcal{T}_0}^{(2,i)}\restriction{\Sigma_0}=0, \label{initial pola tangential}
\end{align}
thanks to \eqref{Pol L} and \eqref{Pol A}. The fact that $V^{(2,i)}_{\Lb_0}$ also vanishes on $\Sigma_0$ is the content of Lemma \ref{lem V=0 initialement}.

\subsubsection{Initial data for the reduced system}\label{section ID reduced system}

We define here the initial data for the solution of the reduced system. For the tensor $\F$ we simply set 
\begin{align}
\F_{\alpha\beta} \restriction{\Sigma_0} & = 0 . \label{ID F=0}
\end{align}
Before we look at $\h$, note that with \eqref{ID F=0} we can apply Lemma \ref{lem g3} and fully compute $g^{(3)}\restriction{\Sigma_0}$. Indeed it is independent from $\h$ and depends only on $F^{(1)}$, $F^{(2,i)}$ and their derivatives which can be computed on $\Sigma_0$ by looking at the transport equations we want them to satisfy, that is \eqref{eq F1}, \eqref{eq F21} and \eqref{eq F22}.

\saut
Since $\h$ solves a second order equation we need to prescribe $\h$ and $\dr_t \h$ on $\Sigma_0$. We start with the spatial components of $\h$, defined so that the induced metric on $\Sigma_0$ of the spacetime metric $g$ defined by \eqref{g} is precisely the Riemannian metric $\bar{g}$ obtained in Theorem \ref{theo initial data}. This imposes
\begin{align}
\h_{ij} \restriction{\Sigma_0}  =  \bar{\h}_{ij} - \lambda g^{(3)}_{ij} . \label{ID h ij}
\end{align}
where $\bar{\h}$ is given in \eqref{g bar theo constraint}. 

\begin{remark}\label{remark chelou}
The definition \eqref{ID h ij} is the reason why we need $g^{(3)}$ to be independent of $\h$. This requirement led to the presence of the term
\begin{align}
- g_0^{\rho\sigma} \h^{\mu\nu} \left( \dr_\mu (g_0)_{\sigma\nu} - \half \dr_\sigma (g_0)_{\mu\nu} - \sin\left(\frac{u_0}{\lambda}\right)  \left(   \dr_\mu u_0 F^{(1)}_{\sigma\nu} - \half \dr_\sigma u_0 F^{(1)}_{\mu\nu} \right) \right),\label{terme chelou}
\end{align}
in $\Upsilon^\rho$ (see \eqref{def Upsilon h}). Otherwise, this term would have been present in the RHS of \eqref{def g3}, leading to $g^{(3)}$ being a function of $\h$. Defining $\h_{ij} \restriction{\Sigma_0} $ would have thus required to invert a system of equations of the form 
\begin{align*}
\h + \la g^{(3)}(\h) = \bar{\h}.
\end{align*}
If $\la$ is small enough, one can easily solve this system of equations but we prefered to avoid this and to include \eqref{terme chelou} in $\Upsilon^\rho$.
\end{remark}

For the time components of $\h$, we choose them so that they compensate the one of $g^{(3)}$, that is
\begin{align}
\h_{0\alpha} \restriction{\Sigma_0} & = - \lambda g^{(3)}_{0\alpha}. \label{ID h 0} 
\end{align}
Note that we can now obtain the initial data for the derivatives of $\F$. Indeed, \eqref{ID F=0} obviously gives us that $\nabla \F_{\alpha\beta}\restriction{\Sigma_0}=0$ and since we want $\F$ to solve \eqref{eq F}, the time derivative of $\F$ is initially given by
\begin{align}
 \dr_t \F_{\alpha\beta}  & = \frac{1}{2|\nabla u_0|_{\bar{g}_0} }  \Pi_\leq \left( \h_{L_0L_0} \right) F^{(1)}_{\alpha\beta} \label{ID dt F}
\end{align}
where we also used \eqref{ID F=0} to simplify \eqref{eq F}.

\saut
We now to define the initial data for $\dr_t \h_{\alpha\beta}$. The spatial components, i.e $\dr_t \h_{ij}$, are chosen so that the second fundamental form of $\Sigma_0$ as an hypersurface in $(\mathcal{M},g)$ is given by $K$ obtained in Theorem \ref{theo initial data}. The second fundamental form is defined by $-\half \mathcal{L}_{T}g$ where $T$ is the unit normal to $\Sigma_0$ for $g$ and $\mathcal{L}$ is the Lie derivative.  \textit{A priori}, the unit normal $T$ depends on $\lambda$. This shows that in order to obtain the dependence of $\dr_t \h_{ij}$ on $\lambda$, we need to know more about $T$.

\begin{lem}\label{lem unit normal}
The future unit normal $T$ to $\Sigma_0$ for the metric $g$ given by \eqref{g} satisfies
\begin{align}
T = \dr_t + Z,\label{def unit normal}
\end{align}
with the vector field $Z=Z^t \dr_t + \bar{Z}^i\dr_i$ given by
\begin{align}
Z^t & = \left( 1 + \lambda^4 \bar{g}^{ij} g^{(2)}_{0i} g^{(2)}_{0j} \right)^{-\half} -1,\label{def Zt}
\\ \bar{Z}^i & = -\lambda^2 \bar{g}^{ij}g^{(2)}_{0j} \left( 1 + \lambda^4 \bar{g}^{ij} g^{(2)}_{0i} g^{(2)}_{0j} \right)^{-\half}.\label{def Zi}
\end{align}
\end{lem}

\begin{proof}
Since the unit normal to a spacelike hypersurface for a given Lorentzian metric is unique, we just need to prove that we can construct the vector field $Z$ such that $T$ defined by \eqref{def unit normal} satisfies
\begin{align*}
g(T,T)=-1 \quad \text{and} \quad g(T,\dr_i)=0.
\end{align*}
Note that thanks to \eqref{ID F1 0}, \eqref{ID F2 00}, \eqref{ID F=0} and \eqref{ID h 0} and the fact that $\dr_t$ is the unit normal to $\Sigma_0$ for $g_0$ we have $g_{00}=-1$ and $g_{0i}= \lambda^2 g^{(2)}_{0i}$. This gives
\begin{align*}
g(T,\dr_i) & = \lambda^2  g^{(2)}_{0i}\left( 1 + Z^t \right)  +   \bar{g}_{ij}\bar{Z}^j,
\end{align*}
so we impose 
\begin{align}
\bar{Z}^i & = -\lambda^2 \bar{g}^{ij}g^{(2)}_{0j}\left( 1 + Z^t \right)\label{def Zi en fonction Z1}.
\end{align}
Moreover, we have
\begin{align*}
g(T,T) & = -1 + 2 g_{0\alpha}Z^\alpha + g_{\alpha\beta}Z^\alpha Z^\beta,
\end{align*}
so we impose $2 g_{0\alpha}Z^\alpha + g_{\alpha\beta}Z^\alpha Z^\beta=0$. After expanding the quadratic term and inserting \eqref{def Zi en fonction Z1} this condition rewrites as a quadratic equation for $Z^t$
\begin{align*}
 \left( 1 + \lambda^4 \bar{g}^{ij} g^{(2)}_{0i}  g^{(2)}_{0j} \right) \left( Z^t \right)^2 + 2 \left( 1 + \lambda^4 \bar{g}^{ij} g^{(2)}_{0i}  g^{(2)}_{0j} \right)Z^t +  \lambda^4 \bar{g}^{ij} g^{(2)}_{0i}  g^{(2)}_{0j} =0.
\end{align*}
The two roots of this equation are
\begin{align*}
\pm \left( 1 + \lambda^4 \bar{g}^{ij} g^{(2)}_{0i}  g^{(2)}_{0j}  \right)^{-\half} - 1.
\end{align*}
Since we want $T$ to be future directed, we choose the root where $\pm=+$, i.e \eqref{def Zt}. Inserting this into \eqref{def Zi en fonction Z1} gives \eqref{def Zi}.
\end{proof}

We now compute the expression of the second fundamental form of $\Sigma_0$ as an hypersurface in $(\mathcal{M},g)$ with $g$ being given by \eqref{g}. 

\begin{lem}\label{lem second fundamental form evol}
If $g$ is given by \eqref{g} and if $\dr_t F^{(1)}_{\alpha\beta}$ is given by \eqref{ID dt F1} then on $\Sigma_0$ the following holds
\begin{align*}
-\half \mathcal{L}_T g_{ij} & = K^{(0)}_\lambda + \lambda K^{(1)}_\lambda + \lambda^2 \left( -\half \mathcal{L}_T \h + \tilde{K}_{evol}^{(\geq 2)} \right),
\end{align*}
where $K^{(0)}_\lambda$ and $K^{(1)}_\lambda$ are given by \eqref{K0 theo constraint} and \eqref{K1 theo constraint}. Moreover, we have schematically
\begin{align}
\lambda^2 \tilde{K}_{evol}^{(\geq 2)} & = Z^\alpha \dr g_0 + \lambda Z^\alpha \dr g^{(1)} +  \lambda^2 \left(\dr \F + \dr F^{(2,i)} + Z^\alpha \dr g^{(2)} \right)+ \lambda^3 T g^{(3)}.\label{K evol geq 2}
\end{align}
\end{lem}

\begin{proof}
Recalling that $g$ is defined by \eqref{g} we have
\begin{align*}
-\half \mathcal{L}_T g & = -\half \mathcal{L}_T g_0 -\frac{\lambda}{2} \mathcal{L}_T g^{(1)} -\frac{\lambda^2}{2} \mathcal{L}_T g^{(2)} -\frac{\lambda^2}{2} \mathcal{L}_T \h -\frac{\lambda^3}{2} \mathcal{L}_T g^{(3)}. 
\end{align*}
We compute or schematically estimate each term in this expression, using the expansion of $T$ given in Lemma \ref{lem unit normal}.  The latter gives
\begin{align*}
-\half \mathcal{L}_T g_0 = K_0 +  \GO{Z^\alpha \dr g_0},
\end{align*}
where $Z^\alpha$ denotes either $Z^t$ or $\bar{Z}^i$ and where we used the fact that $\dr_t$ is the unit normal to $\Sigma_0$ for $g_0$. Similarly, using $\dr_t u_0 = |\nabla u_0|_{\bar{g}_0}$ on $\Sigma_0$ we have
\begin{align*}
-\frac{\lambda}{2} \mathcal{L}_T g^{(1)}_{ij} & = \frac{1}{2} \sin\left(\frac{u_0}{\lambda}\right)  |\nabla u_0|_{\bar{g}_0}  \bar{F}^{(1)}_{ij} -\frac{\lambda}{2}\cos\left(\frac{u_0}{\lambda}\right) \dr_t F^{(1)}_{ij} + \GO{\lambda Z^\alpha \dr g^{(1)}}
\\& = \half \sin\left(\frac{u_0}{\lambda}\right)  |\nabla u_0|_{\bar{g}_0}  \bar{F}^{(1)}_{ij}
\\&\quad  -\frac{\lambda}{2}\cos\left(\frac{u_0}{\lambda}\right)\left( - N_0 \bar{F}^{(1)}_{ij}  +   (\dr_t + N_0)^\rho \Gamma(g_0)^k_{\rho (i} \bar{F}^{(1)}_{k j)} + \frac{1}{2|\nabla u_0|_{\bar{g}_0}}(\Box_{g_0}u_0)  \bar{F}^{(1)}_{ij}  \right)
\\&\quad + \GO{\lambda Z^\alpha \dr g^{(1)}},
\end{align*}
where we used \eqref{ID dt F1} and \eqref{ID F1 0} to compute $\dr_t F^{(1)}_{ij}$ and \eqref{ID F ij} to compute $F^{(1)}_{ij}$. Using in addition \eqref{ID F=0} we now obtain
\begin{align*}
-\frac{\lambda^2}{2} \mathcal{L}_T g^{(2)}_{ij} & = - \frac{\lambda}{2}\cos\left(\frac{u_0}{\lambda}\right)  |\nabla u_0|_{\bar{g}_0} \bar{F}^{(2,1)}_{ij} + \lambda\sin\left(\frac{2u_0}{\lambda}\right)  |\nabla u_0|_{\bar{g}_0} \bar{F}^{(2,2)}_{ij}
\\&\quad + \GO{ \lambda^2 \left(\dr \F + \dr F^{(2,i)} + Z^\alpha \dr g^{(2)} \right)  }.
\end{align*}
This concludes the proof of the lemma, by recalling the expression of $K^{(i)}_\lambda$ for $i=1,2$ given in \eqref{K0 theo constraint} and \eqref{K1 theo constraint}.
\end{proof}

\saut
As explained above,  we want the initial data for $\dr_t \h_{ij}$ to ensure that the second fundamental form of $\Sigma_0$ for $g$, i.e $-\half \mathcal{L}_T g_{ij} $, matches the tensor $K_\lambda$ given by Theorem \ref{theo initial data}. According to the previous lemma, this is equivalent to 
\begin{align}
-\half \mathcal{L}_T\h + \tilde{K}^{(\geq 2)}_{evol} = K^{(\geq 2)}_\lambda,\label{ID dt h ij inter}
\end{align}
where $K_\lambda^{(\geq 2)}$ is given in Theorem \ref{theo initial data}. Using the expansion of $T$, this defines the initial data for $\dr_t \h_{ij}$ on $\Sigma_0$ by
\begin{align}
 \dr_t \h_{ij}  \restriction{\Sigma_0}    = \frac{1}{1 + Z^t} \left(  -  \bar{Z}^k\dr_k \h_{ij} + \h\left( [Z,\dr_{(i}],\dr_{j)}\right)  - 2(K_\lambda^{(\geq 2)})_{ij}  +  2( \tilde{K}_{evol}^{(\geq 2)})_{ij} \right).\label{ID dt h ij}
\end{align}

\saut
It remains to define the initial value for $\dr_t \h_{0\alpha}$. They are chosen so that the generalised wave condition \eqref{Upsilon =0} holds on $\Sigma_0$.  This condition rewrites as
\begin{align}
 g^{\mu\nu}  \left( \dr_\mu \h_{\sigma\nu} - \half \dr_\sigma \h_{\mu\nu} \right) = - g_{\sigma\rho}  \tilde{\Upsilon}^\rho.\label{wave gauge inter}
\end{align}
Note that the initial data for $\tilde{\Upsilon}^\rho$ are fully defined. Indeed the last line in \eqref{def Upsilon tilde} depends on background quantities and on $g\restriction{\Sigma_0}$. For the first line, we use \eqref{ID F=0} for the tangential derivatives or \eqref{ID dt F} for the initial value of $\dr_t \F$. A similar argument using Lemma \ref{lem g3} allows us to deal with $\dr g^{(3)}$ in $\tilde{\Upsilon}^\rho$.

\saut
Thanks to the decomposition of $g^{\mu\nu}$ on $\Sigma_0$, \eqref{wave gauge inter} rewrites as
\begin{align*}
- T^\mu T^\nu  \left( \dr_\mu \h_{\sigma\nu} - \half \dr_\sigma \h_{\mu\nu} \right)  = - g_{\sigma\rho}  \tilde{\Upsilon}^\rho -   \bar{g}^{k\ell}  \left( \dr_k \h_{\sigma\ell} - \half \dr_\sigma \h_{k\ell} \right),
\end{align*}
where the RHS is fully defined thanks to \eqref{ID h ij}, \eqref{ID h 0} and \eqref{ID dt h ij}. We project this expression onto the initial frame $(T,\dr_i)$ and rearrange terms to obtain the initial value of $T \h_{TT}$ and $T \h_{iT}$:
\begin{align}
 T \h_{TT} \restriction{\Sigma_0}  & = 2 \h_{\mu T}TT^\mu   +  2  g_{T\rho}  \tilde{\Upsilon}^\rho + 2T^\sigma \bar{g}^{k\ell}  \left( \dr_k \h_{\sigma\ell} - \half \dr_\sigma \h_{k\ell} \right),\label{ID T h TT}
\\ T \h_{iT} \restriction{\Sigma_0}   & =  \h_{i\nu}TT^\nu + \half T^\mu T^\nu \dr_i \h_{\mu\nu} + g_{i\rho}  \tilde{\Upsilon}^\rho +   \bar{g}^{k\ell}  \left( \dr_k \h_{i\ell} - \half \dr_i \h_{k\ell} \right). \label{ID T h iT}
\end{align}
Using the decomposition of $T$ given by Lemma \ref{lem unit normal} we can deduce from \eqref{ID T h TT} and \eqref{ID T h iT} the initial value of $\dr_t \h_{0\alpha}$ (similarly to how we obtain \eqref{ID dt h ij} from \eqref{ID dt h ij inter}).

\subsubsection{Initial estimates and properties}

In the following corollary, we summarize the formal properties satisfied by the initial data defined in the previous sections and give the estimates they satisfy. It shows in particular that the initial data defined above from the solution of the constraint equations obtained in \cite{Touati2023a} are compatible with the generalised wave gauge condition and the polarization conditions (with the exception of $V^{(2,i)}_{\Lb_0}=0$).

\begin{coro}\label{coro ID}
The initial data for the background and reduced systems defined in Sections \ref{section ID bg system} and \ref{section ID reduced system} are such that
\begin{itemize}
\item[(i)] the first and second fundamental forms of the hypersurface $\Sigma_0$ are given by $(\bar{g}_\la,K_\la)$ from Theorem \ref{theo initial data} and therefore solve the constraint equations,
\item[(ii)] the tensors $F^{(1)}$, $F^{(2,1)}$ and $F^{(2,2)}$ are initially supported in $\{|x|\leq R\}$ and there exists a constant $C_{\mathrm{in}}>0$ such that
\begin{align}
\l F^{(1)} \r_{H^N(\Sigma_0)}  +  \l F^{(2,1)} \r_{H^{N-1}(\Sigma_0)}  +  \l F^{(2,2)} \r_{H^{N-1}(\Sigma_0)}  & \leq C_{\mathrm{in}} \e,\label{estim F 1 21 22 initial}
\\ \max_{k\in\llbracket 0,4 \rrbracket} \lambda^{r} \l \dr\nabla^r \h \r_{L^2_{\delta+ r+1}(\Sigma_0)} & \leq C_{\mathrm{in}} \e,\label{estim h initial}
\end{align}
\item[(iii)] the following initial polarization conditions hold
\begin{align}
\Pol \left( F^{(1)} \right) \restriction{\Sigma_0} & = 0,\label{ID pola F1}
\\ V^{(2,i)}_{\mathcal{T}_0}  \restriction{\Sigma_0}& = 0,\label{ID V2i=0}
\\ \Pol \left( \F \right) \restriction{\Sigma_0} & = 0,\label{ID pola F}
\end{align}
\item[(iv)] the initial backreaction condition holds
\begin{align}
 \mathcal{E}\left( F^{(1)} , F^{(1)} \right) \restriction{\Sigma_0} & = -4 F_0^2,\label{ID energie F1}
\end{align}
where $\mathcal{E}$ is defined in \eqref{def energie E},
\item[(v)] the initial generalised wave gauge condition holds
\begin{align}
\Upsilon^\rho \restriction{\Sigma_0} & = 0.\label{ID Upsilon =0}
\end{align}
\end{itemize}
\end{coro}

\begin{proof}
The first item of the corollary holds thanks to \eqref{ID F ij}, \eqref{ID F=0}, \eqref{ID h ij} and \eqref{ID dt h ij inter}. We now turn to the second item the corollary. The support property of $F^{(1)}$, $F^{(2,1)}$ and $F^{(2,2)}$ follows directly from the support property of $\bar{F}^{(1)}$,  $\bar{F}^{(2,1)}$ and $\bar{F}^{(2,2)}$ stated in Theorem \ref{theo initial data} and \eqref{ID F ij}, \eqref{ID F1 0}, \eqref{ID F2 0A }, \eqref{ID F2 0N} and \eqref{ID F2 00}. Similarly, we deduce \eqref{estim F 1 21 22 initial} from \eqref{estim F bar chap 3}.  

Proving \eqref{estim h initial} amounts to bound the Sobolev norm of the initial data for $\h$ and $\dr_t \h$. First, from \eqref{ID h ij} and \eqref{ID h 0} we can schematically write $\h = \bar{\h} + \lambda g^{(3)}$. On $\Sigma_0$, $g^{(3)}$ is compactly supported and depends only on background quantities (see the discussion following \eqref{ID F=0}). The estimate
\begin{align}
\l \nabla^{r+1} \h \r_{L^2_{\delta+r+1}(\Sigma_0)} \lesssim \frac{\e}{\lambda^r},\label{estim h initial inter}
\end{align}
for $r\in\llbracket 0,4\rrbracket$ thus follows from \eqref{estim h bar chap 3}. We now look at $\dr_t \h_{ij}$, defined in \eqref{ID dt h ij}. Since $Z=\GO{\lambda^2}$ and only contains background quantities (see \eqref{def Zi} and \eqref{def Zt}) the worse terms in \eqref{ID dt h ij} are the last two (the first two can be estimated with \eqref{estim h initial inter}). For them we use \eqref{estim K geq 2 chap 3} and \eqref{K evol geq 2} (note that $\tilde{K}^{(\geq 2)}_{evol}$ involves the non-background quantities $\dr_t \F$ for which we need to use \eqref{ID dt F}). This proves that
\begin{align*}
\l \dr_t \nabla^r \h_{ij} \r_{L^2_{\delta+r+1}(\Sigma_0)} \lesssim \frac{\e}{\lambda^r},
\end{align*}
for $r\in\llbracket 0,4\rrbracket$.  We proceed similarly for $\dr_t \h_{0\alpha}$ using \eqref{ID T h TT} and \eqref{ID T h iT}. This concludes the proof of \eqref{estim h initial}.

\saut
The initial polarization conditions \eqref{ID V2i=0} and \eqref{ID pola F} are consequences of \eqref{ID F2 0A }-\eqref{ID F2 0N}-\eqref{ID F2 00} and \eqref{ID F=0}, while \eqref{ID pola F1} holds since $F^{(1)}$ is initially $P_{0,u}$-tangent and $g_0$-traceless (see the corresponding properties of $\bar{F}^{(1)}$ in Theorem \ref{theo initial data}). Samewise, \eqref{ID energie F1} follows from \eqref{energie condition theo chap 3}. The initial generalised wave gauge condition \eqref{ID Upsilon =0} follows from \eqref{ID T h TT} and \eqref{ID T h iT}.
\end{proof}

\subsection{Reformulation of Theorem \ref{theo main}}

The conclusion of Sections \ref{section high-frequency ansatz} and \ref{section hierarchy} is the following: proving Theorem \ref{theo main} is equivalent to constructing solutions to both the background and reduced system (\eqref{eq F1}-\eqref{eq F22} and \eqref{eq F}-\eqref{eq h} respectively) on $[0,1]\times \R^3$ with initial data given in Section \ref{section ID spacetime metric}, and which satisfies moreover the generalised wave gauge condition and polarization conditions (see Definition \ref{def generalised wave gauge} and \ref{def polarization conditions} respectively). The rest of this article is devoted to this task, which we divide in three parts:
\begin{itemize}
\item in Section \ref{section solving bg system}, we solve the background system (see Theorem \ref{theo bg system}),
\item in Section \ref{section solving reduced system}, we solve the reduced system (see Theorem \ref{theo reduced system}),
\item in Section \ref{section solving EVE}, we show that the generalised wave gauge condition and polarization conditions are satisfied (see Propositions \ref{prop V2i} and \ref{prop Upsilon}).
\end{itemize}

\subsection{Basic transport and energy estimates}

Before proceeding with this program, we collect here the basic transport estimate for the operator $L_0$ and the energy estimate for the operator $\tBox_g$.

\begin{lem}\label{lem L0 h =f}
Let $h$ a scalar function on $[0,T]\times \R^3$ for some $T\leq 1$, compactly supported. The following holds for all $t\in[0,T]$
\begin{align}
\l h \r_{L^2}^2(t) \leq  C(C_0) \left( \l h \r_{L^2}^2(0) + \int_0^t \left( \l h \r_{L^2}^2(s) + \l L_0 h \r_{L^2}^2(s) \right) \d s \right)
\label{estim transport L2}.
\end{align}
\end{lem}

\begin{proof}
We start by decomposing $L_0$ in the $(\dr_t,\dr_i)$ basis:
\begin{align*}
L_0 = L_0^t \dr_t + L_0^i \dr_i.
\end{align*}
Recalling \eqref{def L0} we have
\begin{align*}
L_0^t & = - g_0^{t t} \dr_t u_0 - g_0^{i t} \dr_i u_0,
\\ L_0^i & = - g_0^{t i} \dr_t u_0- g_0^{j i} \dr_j u_0.
\end{align*}
Since $u_0$ satisfies the eikonal equation \eqref{eq u0} we have
\begin{align*}
g_0^{t t} ( \dr_t u_0)^2  + 2 g_0^{t i} \dr_t u_0 \dr_i u_0 + | \nabla u_0 |^2_{\bar{g}_0}= 0.
\end{align*}
Taking $\e$ small enough in \eqref{estim g0} and \eqref{estim u0-z} then ensures that 
\begin{align}
\left| \frac{1}{L_0^t} \right|  + \left| L_0^\mu \right| +  \left| \dr L_0^\mu \right| \leq C(C_0) ,\label{estim L0}
\end{align}
on $[0,1]\times \R^3$. We now prove the energy estimate. After multiplication by $h$, the equation $L_0 h =f$ can be rewritten as
\begin{align}
\half  \dr_t h^2& =\frac{ f h}{L_0^t}  - \half \frac{ L_0^i}{L_0^t} \dr_i h^2.\label{L0t dt h}
\end{align}
For $s\in[0,T]$ we integrate \eqref{L0t dt h} on $\Sigma_s$ with respect to the usual Lebesgue measure and integrate by part in the last integral to obtain
\begin{align}
\half \int_{\Sigma_s} \dr_t h^2\d x& = \int_{\Sigma_s} \frac{ f h}{L_0^t} \d x + \half  \int_{\Sigma_s}\dr_i \left( \frac{ L_0^i}{L_0^t} \right)  h^2 \d x.\label{first integral}
\end{align}
In the integral of the LHS, we recognize the time derivative of an energy for $h$. Therefore \eqref{first integral} leads to
\begin{align*}
 \frac{\d}{\d t}  \l h \r_{L^2}^2(s)    & = 2 \int_{\Sigma_s} \frac{ f h}{L_0^t} \d x +   \int_{\Sigma_s}\dr_i \left( \frac{ L_0^i}{L_0^t} \right)  h^2 \d x.
\end{align*}
Using now the Cauchy-Schwarz inequality for the first integral in the RHS and \eqref{estim L0} we obtain
\begin{align*}
 \frac{\d}{\d t}  \l h \r_{L^2}^2(s)  \leq C(C_0) \left(  \l h \r_{L^2}^2(s) + \l f \r_{L^2}^2(s) \right).
\end{align*}
Integrating this inequality on $[0,t]$ for $t\leq T$ concludes the proof.
\end{proof}

For clarity we define
\begin{align*}
E_\sigma(h) & = \l  \dr_t h \r_{L^2_\sigma}^2 + \l  \nabla h \r_{L^2_\sigma}^2
\end{align*}
for $h$ a scalar function and $\sigma\in\R$.

\begin{lem}\label{lem EE}
Let $h$ be a scalar function on $[0,T]\times\R^3$, for all $t\in[0,T]$ we have
\begin{align*}
E_\sigma(h) (t) \lesssim E_\sigma(h) (0) +  \left( 1 + \l \dr g \r_{L^\infty} \right)  \int_0^t E_\sigma(h)(s)\d s + \int_0^t \l \tBox_gh \r_{L^2_\sigma}^2(s)\d s .
\end{align*}
\end{lem}

\begin{proof}
Let $w(x)=\left( 1+ |x|^2\right)^\sigma$. We set $f=\tBox_gh$ and multiply this identity by $w\dr_t h$ and integrate over $\Sigma_t$ for $t\in[0,T]$ with the usual Lebesgue measure:
\begin{align*}
\int_{\Sigma_t} wg^{tt}\dr_t h\dr^2_t h \d x + 2  \int_{\Sigma_t}wg^{ti}\dr_t h\dr_i\dr_t h \d x + \int_{\Sigma_t}wg^{ij}\dr_t h \dr_i \dr_j h \d x = \int_{\Sigma_t}w \dr_t hf \d x .
\end{align*}
We rewrite the first integral and integrate by parts in the second and third integrals to obtain:
\begin{align}
 -\half& \frac{\d E}{\d t}   -  \int_{\Sigma_t} w \left( \dr_t g^{tt}\right) (\dr_t h)^2 \d x - \int_{\Sigma_t} \dr_i \left( w g^{0i} \right) (\dr_t h)^2 \d x\nonumber
\\&\qquad - \int_{\Sigma_t}\dr_i\left(wg^{ij} \right)\dr_t h \dr_j h  \d x + \half \int_{\Sigma_t}w \left( \dr_t g^{ij} \right) \dr_i h  \dr_j h \d x  = \int_{\Sigma_t}w \dr_t hf \d x ,\label{energie égalité}
\end{align}
where
\begin{align*}
E(t) = \int_{\Sigma_t} w \left( -g^{tt}(\dr_t h)^2  + g^{ij}\dr_i h \dr_j h \right).
\end{align*}
Taking $\e$ small enough in \eqref{estim g0} we obtain the existence of a universal constant $C'>0$ such that 
\begin{align}
\frac{1}{C'} E_\sigma(h) \leq E \leq C' E_\sigma(h). \label{comparaison E et Esigma}
\end{align}
Using $|\nabla w|\lesssim w$, we obtain $ | \dr ( wg) | \lesssim \left( 1 + \l \dr g \r_{L^\infty} \right)  w  $. Going back to \eqref{energie égalité}, this implies
\begin{align*}
\frac{\d E}{\d t}(t) & \lesssim \left( 1 + \l \dr g \r_{L^\infty} \right) E_\sigma(h)(t) + \l f \r_{L^2_\sigma}^2(t),
\end{align*}
where we also used the Cauchy-Schwarz inequality. Integrating this inequality from 0 to $t\in[0,T]$ we obtain
\begin{align*}
E(t) \leq E(0) +  \left( 1 + \l \dr g \r_{L^\infty} \right)  \int_0^t E_\sigma(h)(s)\d s + \int_0^t \l f \r_{L^2_\sigma}^2(s)\d s.
\end{align*}
Using \eqref{comparaison E et Esigma} concludes the proof.
\end{proof}

\section{Solving the background system}\label{section solving bg system}

In this section we solve the background system \eqref{eq F1}-\eqref{eq F22} defined in Section \ref{section bg system}. The following theorem summarizes the estimates and the support properties of the background perturbations.

\begin{thm}\label{theo bg system}
Given initial data as in Corollary \ref{coro ID} and if $\e$ is small enough, there exists a unique solution 
\begin{align}
\left( F^{(1)},F^{(2,1)},F^{(2,2)}\right),\label{solution bg system n uplet}
\end{align}
of the background system \eqref{eq F1}-\eqref{eq F22} on $[0,1]\times \R^3$. Moreover:
\begin{itemize}
\item[(i)] the tensors in \eqref{solution bg system n uplet} are supported in $\{ |x|\leq C_{\mathrm{supp}}R\}$,  
\item[(ii)] there exists $C_1=C_1(C_0)>0$ such that the following estimate holds
\begin{align*}
\l  F^{(1)} \r_{H^N}   +  \l  F^{(2,1)} \r_{H^{N-2}}  +  \l  F^{(2,2)} \r_{H^{N-2}}  & \leq C_1\e   ,
\end{align*}
\item[(iii)] $F^{(1)}$ satisfies \eqref{pola F1 theo}, \eqref{energie F1 theo} and \eqref{eq F1 theo}.
\end{itemize}
\end{thm}

We want the background perturbations to solve the equations \eqref{eq F1}-\eqref{eq F21}-\eqref{eq F22}. The equations solved by their coefficients are of the form
\begin{align}
\left( L_0 + \eta\right) h = f,  \label{eq modèle transport}
\end{align}
for $f$ and $\eta$ given scalar functions on $[0,1]\times \R^3$ (recall \eqref{def Ll_0}). Moreover, thanks to Corollary \ref{coro ID} the initial data on $\Sigma_0$ for the background perturbations are compactly supported.  For the model equation \eqref{eq modèle transport} this allows us to consider $f$ supported in $[0,1]\times K$ for $K$ a compact of $\R^3$ and $h\restriction{\Sigma_0}$ compactly supported. Since $L_0$ is the spacetime gradient of the background phase $u_0$ and thanks to \eqref{eq u0} and \eqref{estim u0-z} we can define characteristics for the equation \eqref{eq modèle transport} on $[0,1]\times \R^3$. The existence of a solution $h$ to this equation can thus be proved using the method of characteristics, we don't give the details. 

\saut
A solution $h$ of \eqref{eq modèle transport} can be proved to be as regular as $f$ using its expression given by the characteristics method. However, since the proof of Theorem \ref{theo main} is based on Sobolev spaces, we have derived an estimate for the transport operator $L_0$ in Lemma \ref{lem L0 h =f}. As a byproduct, this standard estimate proves the uniqueness of the solution to \eqref{eq modèle transport} and therefore of the solutions to \eqref{eq F1}-\eqref{eq F21}-\eqref{eq F22}. These equations have a triangular structure:
\begin{itemize}
\item \eqref{eq F1} only depends on the background quantities through the coefficients of the operator $\Ll_0$,
\item the RHS of \eqref{eq F21} depends in addition on $\dr^{\leq 2}F^{(1)}$,
\item the RHS of \eqref{eq F22} depends on the background quantities, $\dr^{\leq 2}F^{(1)}$ and $F^{(2,1)}$. 
\end{itemize}
Moreover, thanks to the support properties stated in Corollary \ref{coro ID}, we obtain that the background perturbations are supported in $\{|x|\leq C_{\mathrm{supp}}R\}$. The estimates stated in Theorem \ref{theo bg system} follow directly from Lemma \ref{lem L0 h =f}, the estimates of the RHS of \eqref{eq F21}-\eqref{eq F22} we prove in the following lemma and the initial regularity stated in Corollary \ref{coro ID}, see \eqref{estim F 1 21 22 initial}.

\begin{lem}
We have
\begin{align}
\l \text{RHS of \eqref{eq F21}} \r_{H^{N-2}} & \lesssim C(C_0) \l F^{(1)} \r_{H^N}  ,\label{estim eq F21}
\\ \l \text{RHS of \eqref{eq F22}} \r_{H^{N-2}} & \lesssim C(C_0) \left(  \l F^{(1)} \r_{H^N}  + \l F^{(2,1)} \r_{H^{N-2}}  \right)   .    \label{estim eq F22}
\end{align}
\end{lem}

\begin{proof}
The estimates \eqref{estim eq F21} and \eqref{estim eq F22} follow from \eqref{Wtilde 11}, \eqref{Wtilde 12}, \eqref{P1k} and \eqref{Phat0} and the background regularity stated in Section \ref{section BG}.
\end{proof}

\section{Solving the reduced system}\label{section solving reduced system}

In this section we solve the reduced system \eqref{eq F}-\eqref{eq h}. This is done by a bootstrap argument, i.e by showing that we can improve \textit{a priori} estimates using the equations. 

\subsection{Bootstrap assumptions}

As explained in Section \ref{section bad coupling strat}, the reduced system loses one derivative. The addition of the Fourier cut-off $\Pi_\leq$ in the equation for $\F$ implies that two derivatives of $\F$ are actually at the level of one derivative of $\h$, which is consistent with what we would obtain from an energy estimate on the wave equation for $\h$. Indeed, by commuting \eqref{eq F} with $\nabla^2$ (the case of time derivatives is similar) we obtain broadly from \eqref{P low}
\begin{align*}
\l  \nabla ^2 \F  \r_{L^2} \lesssim \int_0^t \l \nabla^2 \Pi_\leq \left( \h \right) \r_{L^2} \lesssim \frac{1}{\lambda} \int_0^t \l \nabla \h \r_{L^2}.
\end{align*}
Putting this in the energy estimate for $\h$ would give $\l \nabla \h \r_{L^2}\lesssim \frac{1}{\lambda}T$ with $T$ the time of existence. Compensating the largeness of $\frac{1}{\lambda}$ by the smallness of $T$ would prove well-posedness for the system \eqref{eq F}-\eqref{eq h} on $[0,T_\lambda]\times \R^3$ for $T_\lambda\to 0$ when $\lambda$ tends to 0.  This is far from what we want since it doesn't allow us to consider the high-frequency limit $\lambda\to 0$ of $g_\lambda$. For this we need a time of existence independent of $\lambda$. In other words, the addition of $\Pi_\leq$ removes the loss of derivative but does not remove the "loss of $\lambda$". This discussion explains our choice of bootstrap assumptions on $\F$ and $\h$:
\begin{itemize}
\item $\F$ is supported in $J_0^+(\{ |x| \leq R \})$ and satisfies
\begin{align}
\l \F \r_{L^2} + \max_{r\in\llbracket 1,6\rrbracket}\lambda^{r-1} \l \nabla^r \F \r_{L^2} + \max_{r\in\llbracket 0,5\rrbracket}\lambda^{r} \l \dr_t \nabla^r \F \r_{L^2} + \max_{r\in\llbracket 0,4\rrbracket}\lambda^{r+1} \l \dr_t^2 \nabla^r \F \r_{L^2} & \leq A_1 \e, \label{BA F k}\tag{\textbf{BA1}}
\\ \max_{r\in\llbracket 0,4\rrbracket} \lambda^r \l \nabla^r \tBox_{ g } \F \r_{L^2} & \leq A_2 \e, \label{BA box F k}\tag{\textbf{BA2}}
\end{align}
\item $\h$ satisfies
\begin{align}
\max_{r\in\llbracket 0,4\rrbracket}\lambda^r\left( \l \dr_t \nabla^r \h \r_{L^2_{\delta+1+r}} + \l \nabla \nabla^r \h \r_{L^2_{\delta+1+r}}  \right) + \max_{r\in\llbracket 0,3\rrbracket} \lambda^{r+1} \l \dr_t^2 \nabla^r \h \r_{L^2_{\delta+2+r}} & \leq A_3 \e   \label{BA h k}.\tag{\textbf{BA3}}
\end{align}
\end{itemize}
These bootstrap assumptions solve the "loss of $\lambda$" since \eqref{BA box F k} implies that $\tBox_g\F$ is bounded by $\lambda^0$ in $L^2$ while \eqref{BA F k} implies boundedness by $\frac{1}{\lambda}$. Recovering \eqref{BA box F k} is the main challenge we face. Note that the support property of $\F$ just follows from the transport equation it satisfies \eqref{eq F}, the fact that its RHS is supported in $J_0^+(\{ |x| \leq R \})$ and that $\F\restriction{\Sigma_0}=0$.

\saut 
As explained above, a local in time solution to \eqref{eq F}-\eqref{eq h} exists \textit{a priori} and by taking the constants $A_i$ large enough (compared to $C_{\mathrm{in}}$ of Corollary \ref{coro ID}) we can assume that the set of times $t$ such that \eqref{BA F k}-\eqref{BA box F k}-\eqref{BA h k} are satisfied on $[0,t]$ is non-empty. We define $T$ to be the supremum of these times. We have $0<T\leq 1$ but \textit{a priori} $T$ depends on $\lambda$. In the next sections, we prove that actually \eqref{BA F k}-\eqref{BA box F k}-\eqref{BA h k} hold with better constants on $[0,T]$. This will contradict the definition of $T$ and prove that a solution to \eqref{eq F}-\eqref{eq h} exists and satisfy \eqref{BA F k}-\eqref{BA box F k}-\eqref{BA h k} on $[0,1]$, that is the interval of existence of $g_0$, $F^{(1)}$, $F^{(2,1)}$ and $F^{(2,2)}$.

\saut
In the sequel, we will always precise the dependence on $\lambda$ of the estimates. For $X$ a numerical quantity, the notation $C(X)$ will denote any function of $X$. The symbol $\lesssim$ will denote $\leq C$ for $C$ depending only on $\delta$, $R$ and other irrelevant numerical constants.  We already assume that $A_1 \ll A_2 \ll A_3$ and that $A_i \gg C(C_0) $ for all the $C(C_0)$ that will appear in the proofs.

\begin{remark}
Since $\lambda$ can be chosen as small as we want, the bootstrap assumptions \eqref{BA F k} and \eqref{BA h k} imply that $\l X \r_{H^s} \lesssim \l \nabla^s X \r_{L^2}$ for $X\in\{ \F,\h\}$ (for $\h$ this holds with the appropriate weights at spacelike infinity), i.e that the main term in each $H^s$ norm is the norm of the top derivatives. Therefore in the sequel, when using (weighted) Sobolev embeddings we will only write down the norm of the top derivatives.
\end{remark}

\subsection{First consequences}

We give here the first consequences of the bootstrap assumptions \eqref{BA F k}-\eqref{BA box F k}-\eqref{BA h k}.  Note first that the inequality
\begin{align*}
\l \h \r_{L^2_{\delta+1}}(t) & \lesssim \l \h \r_{L^2_{\delta+1}}(0) + \int_0^t \l \dr_t \h \r_{L^2_{\delta+1}}\d s
\end{align*}
together with \eqref{BA h k} implies $\l \h \r_{L^2_{\delta+1}} \lesssim C_{\mathrm{in}}\e + T A_3\e$. Since $T\leq 1$ and $A_3$ is assumed to be larger than $C_{\mathrm{in}}$ we obtain
\begin{align}
\l \h \r_{L^2_{\delta+1}} \lesssim A_3 \e . \label{estim h L2}
\end{align}

\saut
We now deduce from the bootstrap assumptions some estimates for the metric $g$.  We introduce the following split of the metric $g$:
\begin{align}
g = g_{\text{BG}} + \tilde{g},\label{splitting g n}
\end{align}
where 
\begin{align}
g_{\text{BG}} & = g_0 + \lambda g^{(1)} + \lambda^2 \left(  F^{(2,1)} \sin\left( \frac{u_0}{\lambda}\right) + F^{(2,2)} \cos\left(\frac{2u_0}{\lambda}\right)   \right) + \lambda^3 g^{(3,\mathrm{BG})},\label{def gBG}
\\ \tilde{g} & = \lambda^2 \left(  \F\sin\left( \frac{u_0}{\lambda}\right)  + \h + \lambda g^{(3)}\left( \F \right) \right).\label{def gtilde}
\end{align}

\begin{remark}
Note that the term $g^{(3)}\left( \F \right)$ in $\tilde{g}$ will be omitted in the sequel (with the exception of Lemma \ref{lem B}) since it is basically a linear term in $\F$ (see Lemma \ref{lem g3}) with an extra $\lambda$, and thus behaves at least better than $\F\sin\left( \frac{u_0}{\lambda}\right)$.
\end{remark}

On the one hand, the metric $g_{\text{BG}}$ contains the background metric and all the background perturbations. Thanks to Theorem \ref{theo bg system} it is very regular but because of $g^{(1)}$ its behaviour with respect to $\lambda$ is bad. On the other hand,  the metric $\tilde{g}$ is less regular than $g_{\text{BG}}$ but satisfies better estimates in terms of $\lambda$. The following lemma gathers the required properties of $g$, $g_\BG$ and $\tilde{g}$.

\begin{lem}\label{lem g BG gtilde}
The following estimates hold
\begin{align}
\max_{r\in\llbracket 0,3\rrbracket}\lambda^{r-1} \l \nabla^r (g- g_0) \r_{L^\infty} & \lesssim A_3\e , \label{d g-g0 Linfini}
\\ \l \nabla^4 (g_{\BG}- g_0) \r_{L^\infty} & \lesssim C(C_0)\frac{\e}{\lambda^3}. \label{d4 gBG-g0 Linfini}
\end{align}
\end{lem}

\begin{proof}
Let $r\leq 3$. We have $g-g_0= g_{\mathrm{BG}}-g_0 + \tilde{g}$. For $g_{\mathrm{BG}}-g_0$ we simply use \eqref{def gBG}, the estimates of Theorem \ref{theo bg system} and Lemma \ref{lem g3} to obtain
\begin{align*}
\l \nabla^r(g_{\mathrm{BG}}-g_0) \r_{L^\infty} \lesssim C(C_0)\frac{\e}{\lambda^{r-1}}.
\end{align*}
The proof of \eqref{d4 gBG-g0 Linfini} is similar. For $\tilde{g}$ we use the weighted Sobolev embedding $H^2_{\delta+r}\xhookrightarrow{}L^\infty$ of Proposition \ref{prop WSS chap 2}:
\begin{align*}
\l \nabla^r \tilde{g} \r_{L^\infty} & \lesssim \lambda^2  \l \nabla^r \h \r_{L^\infty}  + \sum_{r_1+r_2=r}\frac{\l \nabla^{r_1}\F \r_{L^\infty}}{\lambda^{r_2-2}}
\\& \lesssim \lambda^2  \l \nabla^{r+2} \h \r_{L^2_{\delta+r+2}}  + \sum_{r_1+r_2=r}\frac{\l \nabla^{r_1+2}\F \r_{L^2}}{\lambda^{r_2-2}}
\\&\lesssim \lambda^2 \frac{A_3\e}{\lambda^{r+1}}  + \sum_{r_1+r_2=r}\frac{1}{\lambda^{r_2-2}} \times \frac{A_1\e}{\lambda^{r_1+1}}
\\&\lesssim \frac{A_3\e}{\lambda^{r-1}},
\end{align*}
where we used \eqref{BA h k},  \eqref{BA F k} and $r\leq 3$.
\end{proof}

Also note that the estimates in Lemma \ref{lem g BG gtilde} also hold for the inverse of $g$, which we can also decompose as in \eqref{splitting g n}. In what follows we won't make any difference between a coefficient of $g$ or a coefficient of its inverse.

\saut
We also deduce from the bootstrap assumptions an estimate for $\mathcal{R}(g)$, i.e the term in \eqref{eq h} containing only non-problematic terms.
\begin{lem}
The following estimate holds 
\begin{align}
\max_{r\in\llbracket 0,4\rrbracket} \lambda^r  \l \nabla^r \mathcal{R}(g)\r_{L^2_{\delta+1+r}} \lesssim  C(A_i)\e^2.\label{d R g L2}
\end{align}
\end{lem}

\begin{proof}
In terms of expansions in powers of $\lambda$, the definition of $\mathcal{R}(g)$ (see \eqref{def R(g)}) ensures the correct behaviour, that is $\mathcal{R}(g)=\GO{\lambda^0}$ but since it contains oscillating terms one loses a power of $\lambda$ for each derivatives. This explains why
\begin{align}
 \l \nabla^r \mathcal{R}(g)\r_{L^2_{\delta+1+r}} \lesssim \frac{1}{ \lambda^r}.
\end{align}
In terms of its dependency on the metric coefficients, we don't need to be very precise about the exact expression of $\mathcal{R}(g)$, its only important property is that it does not contain second order derivatives of either $\F$, $g^{(3)}$ or $\h$. Indeed:
\begin{itemize}
\item the quadratic non-linearity term $P^{(\geq 2)}$ is of the form $g^{-2}\dr g \dr g$,
\item the $\dr^2\F$, $\tilde{\dr}^2g^{(3)}$ and $\dr^2\h$ coming from the wave part of the Ricci tensor are already in \eqref{eq h} and are thus absent from $\tilde{W}^{(\geq 2)}$,
\item the $\dr\F$, $\tilde{\dr} g^{(3)}$ and $\dr\h$ coming from the $H^\rho$ term in the Ricci tensor are already put into $\Upsilon^\rho$ and therefore absent from $\mathring{H}^\rho$, which implies that the $\dr^2\F$, $\tilde{\dr}^2g^{(3)}$ and $\dr^2\h$ coming from the gauge part are absent from $\mathcal{R}(g)$.
\end{itemize}
Therefore, the terms in $\mathcal{R}(g)$ can be put into two categories: first the purely background terms, i.e depends only on $g_0$, $F^{(1)}$ or $F^{(2,i)}$ and their first and second derivatives, and second the non-background terms, depending quadratically on zeroth or first order derivatives of $\F$ or $\h$ with background coefficients.

\saut
We can bound all of terms from the first category in $L^\infty$ thanks to the background regularity, Theorem \ref{theo bg system} and the fact that $N\geq 10$ (see Remark \ref{remark N geq 10}). Since these terms are all background, the constant appearing in the estimate is of the form $C(C_0)$.

\saut
The terms from the second category can all be estimated using the bootstrap assumptions \eqref{BA F k} and \eqref{BA h k} and by bounding the background coefficients in $L^\infty$. We only give the details for the worse term in terms of $\lambda$ behaviour and spatial support, i.e $\lambda^2g_0^{-2}\left( \dr\h\right)^2$:
\begin{align*}
\l \lambda^2 \nabla^r\left( g_0^{-2}\left( \dr\h\right)^2 \right) \r_{L^2_{\delta+r+1}} & \lesssim C(C_0) \lambda^2  \sum_{r_1+r_2=r} \l \dr \nabla^{r_1}  \h \dr \nabla^{r_2} \h  \r_{L^2_{\delta+r+1}}  .
\end{align*}
If $r_1\leq 2$, we put $\dr\nabla^{r_1}\h$ in $L^\infty$, i.e we use the product law of weighted Sobolev spaces from Proposition \ref{prop WSS chap 2}:
\begin{align*}
\l \dr \nabla^{r_1}  \h \dr \nabla^{r_2} \h  \r_{L^2_{\delta+r+1}}  & \lesssim \l \dr\nabla^{r_1}\h \r_{H^2_{\delta+r_1+1}} \l  \dr\nabla^{r_2}\h  \r_{L^2_{\delta+r_2+1}} 
\\& \lesssim  \frac{A_3\e}{\lambda^{r_1+2} } \times \frac{A_3\e}{\lambda^{r_2} }.
\end{align*}
If $r_1\geq 3$, then $r_2\leq 2$ (since $r\leq 4$) and we proceed similarly. Finally we obtain
\begin{align*}
\l \lambda^2 \nabla^r\left( g_0^{-2}\left( \dr\h\right)^2 \right) \r_{L^2_{\delta+r+1}} & \lesssim \frac{C(A_i)\e^2}{\lambda^r}.
\end{align*}

\end{proof}

\subsection{Estimates for $\F$}

We start by studying $\F$, which solves the transport equation \eqref{eq F}.  For clarity, we simply write $\h$ instead of $\h_{L_0L_0}$ appearing in this equation.

\begin{prop}\label{prop BA F n+1}
The following estimates hold
\begin{align}
\l \F \r_{L^2} & \lesssim  A_3\e^2 , \label{BA F n+1 a}
\\  \max_{r\in\llbracket 1,6\rrbracket}\lambda^{r-1} \l \nabla^r \F \r_{L^2} & \lesssim \left(  \lambda A_1 \e +  A_3\e^2 \right) ,\label{BA F n+1 b}
\\  \max_{r\in\llbracket 0,5\rrbracket}\lambda^{r} \l \dr_t \nabla^r \F \r_{L^2} + \max_{r\in\llbracket 0,4\rrbracket}\lambda^{r+1} \l \dr_t^2 \nabla^r \F \r_{L^2} & \lesssim \left(  \lambda A_1 \e +  A_3\e^2 \right) .\label{BA F n+1 c}
\end{align}
\end{prop}

\begin{proof}
We want to apply Lemma \ref{lem L0 h =f}, so we need to estimate the $L^2$ norm of the RHS of \eqref{eq F}, which we rewrite
\begin{align}
L_0\F_{\alpha\beta}  & = \half  \Pi_\leq \left( \h \right) F^{(1)}_{\alpha\beta} + \half (\Box_{g_0}u_0)\F_{\alpha\beta} + L_0^\mu \Gamma(g_0)^\nu_{\mu(\alpha}\F_{\nu\beta)}.\label{eq F true}
\end{align}
The RHS of \eqref{eq F true} is supported in $\{ |x|\leq C_{\mathrm{supp}}R\}$ thanks to the support property of $\F$ and $F^{(1)}$.  We estimate all the background quantities (see Theorem \ref{theo bg system}) in $L^\infty$ and obtain
\begin{align*}
\l \text{RHS of \eqref{eq F true}}\r_{L^2}& \lesssim C(C_0) \left( \e \l \Pi_\leq \left( \h\right) \r_{L^2} + \l \F \r_{L^2} \right)
\\& \lesssim C(C_0) A_3\e^2   + C(C_0) \l \F \r_{L^2},
\end{align*}
where we forgot about the projector $\Pi_\leq$ and used \eqref{estim h L2}. Therefore,  $\F\restriction{\Sigma_0}=0$ and Lemma \ref{lem L0 h =f} imply that for all $t\in [0,T]$ we have
\begin{align*}
\l \F \r^2_{L^2} (t) \lesssim C(C_0) A^2_3\e^4   + C(C_0) \int_0^t \l \F \r^2_{L^2}\d s,
\end{align*}
which gives \eqref{BA F n+1 a} after applying Gronwall's inequality.

\saut
We now turn to the proof of \eqref{BA F n+1 b}, for which we need to commute \eqref{eq F true} with spatial derivatives. For $r\in\llbracket 1,6 \rrbracket$ we schematically obtain
\begin{align}
L_0 \left( \nabla^r \F \right) & = \left[ L_0 , \nabla^r\right] \F + \sum_{r_1+r_2=r} \left( \nabla^{r_1}\Pi_\leq \left( \h\right) \nabla^{r_2}F^{(1)} + \nabla^{r_1}\left( \Box_{g_0}u_0 + L_0\Gamma(g_0) \right) \nabla^{r_2}\F\right).\label{eq nabla r F n+1}
\end{align}
From Lemma \ref{lem L0 h =f}, \eqref{eq nabla r F n+1} and $\F\restriction{\Sigma_0}=0$ we have for $t\in[0,T]$
\begin{align}
\l \nabla^r \F \r^2_{L^2} (t) \lesssim C(C_0)\int_0^t \left(  \l \left[ L_0 , \nabla^r\right] \F  \r^2_{L^2} +  \sum_{r'\leq r} \left(\e^2 \l\nabla^{r'} \Pi_\leq \left( \h\right)  \r^2_{L^2} +  \l \nabla^{r'}\F \r^2_{L^2} \right) \right)\d s , \label{nabla r F estimate}
\end{align} 
where we again estimate all the background quantities in $L^\infty$. 

\saut
We estimate the commutator $\left[ L_0 , \nabla^r\right] \F$ using the following formula, true for any linear operators:
\begin{align}
[A,B^n] = \sum_{k=0}^{n-1}B^{n-1-k}[A,B]B^k.\label{commutateur abstrait}
\end{align}
With the usual Leibniz rule this gives
\begin{align}
\left| \left[ L_0 , \nabla^r\right] \F \right| & \lesssim \sum_{k=0}^{r-1}\left| \nabla^{r-1-k}\left( (\nabla L_0^\alpha)\dr_\alpha \nabla^k\F \right) \right| \lesssim C(C_0)\sum_{k=0}^{r-1}\left|\dr \nabla^k \F \right|\label{commut L0 nabla r}
\end{align}
where we used the background regularity.  

\saut
We first treat the case $r=1$ in \eqref{nabla r F estimate}.  In order to treat the case of a time derivative in \eqref{commut L0 nabla r}, we rewrite the equation satisfied by $\F$
\begin{align}
L_0^t \dr_t \F = L_0^i \nabla \F + \Pi_\leq (\h) F^{(1)} + (\Box_{g_0}u_0 + L_0 \Gamma(g_0)) \F\label{dt F}
\end{align}
which together with \eqref{estim h L2} gives $\l \dr \F \r_{L^2} \lesssim C(C_0) \l \nabla \F \r_{L^2} + C(C_0)A_3\e^2 $. Therefore, using \eqref{BA F n+1 a} and \eqref{estim h L2} we obtain from \eqref{nabla r F estimate}
\begin{align*}
\l \nabla \F \r^2_{L^2} (t) & \lesssim C(C_0) A_3^2\e^4  + C(C_0)\int_0^t \left(  \e^2 \l\nabla \Pi_\leq \left( \h\right)  \r^2_{L^2} +  \l \nabla\F \r^2_{L^2}  \right)\d s
\\&\lesssim C(C_0) A_3^2\e^4  + C(C_0)\int_0^t   \l \nabla\F \r^2_{L^2} \d s,
\end{align*} 
where we also forgot about $\Pi_\leq$ and used \eqref{BA h k}. Using Gronwall's lemma, this proves \eqref{BA F n+1 b} in the case $r=1$. If now $r\geq 2$ in \eqref{nabla r F estimate}, we have from \eqref{BA F k}
\begin{align*}
\l \left[ L_0 , \nabla^r\right] \F  \r^2_{L^2} & \lesssim C(C_0)\sum_{k=0}^{r-2}\l \dr \nabla^k \F \r^2_{L^2} + C(C_0) \l \dr \nabla^{r-1} \F \r^2_{L^2}
\\&\quad \lesssim C(C_0) \frac{A_1^2\e^2}{\lambda^{2(r-2)}}+ C(C_0) \l \dr \nabla^{r-1} \F \r^2_{L^2}.
\end{align*}
For the second term in this last expression we rewrite the equation for $\nabla^{r-1}$:
\begin{align}
L_0^t\dr_t  \nabla^{r-1} \F & = \left[ L_0 , \nabla^{r-1}\right] \F + L_0^i   \nabla^{r} \F \label{dt nabla r F expression}
\\&\quad + \sum_{r_1+r_2=r-1} \left( \nabla^{r_1}\Pi_\leq \left( \h\right) \nabla^{r_2}F^{(1)} + \nabla^{r_1}\left( \Box_{g_0}u_0 + L_0\Gamma(g_0) \right) \nabla^{r_2}\F \right).\nonumber
\end{align}
which thanks to \eqref{commut L0 nabla r}, \eqref{BA F k}, \eqref{BA h k} and \eqref{P low} gives
\begin{align}
\l \dr \nabla^{r-1} \F \r^2_{L^2} & \lesssim C(C_0) \l \nabla^r \F \r_{L^2}^2 + C(C_0)\frac{A_3^2\e^4+ A_1^2\e^2}{\lambda^{2(r-2)} }.\label{dt nabla r F}
\end{align}
Therefore, we have estimated the commutator in \eqref{nabla r F estimate}:
\begin{align*}
\l \left[ L_0 , \nabla^r\right] \F  \r^2_{L^2} \lesssim C(C_0) \l \nabla^r \F \r_{L^2}^2 + C(C_0)\frac{\lambda^2 A_3^2\e^4+ \lambda^2 A_1^2\e^2}{\lambda^{2(r-1)} }.
\end{align*}
For the first sum in the time integral in \eqref{nabla r F estimate}, we distinguish between $r'=r$ and $r'\leq r-1$ and use \eqref{P low} and \eqref{BA h k} (and $r-1\leq 5$): 
\begin{align*}
\sum_{r'\leq r} \e^2\l\nabla^{r'} \Pi_\leq \left( \h_{L_0L_0}\right)  \r^2_{L^2} &  \lesssim \frac{\e^2}{\lambda^2}\l\nabla^{r-1} \h \r^2_{L^2} +  \sum_{r'\leq r-1} \e^2\l\nabla^{r'}  \h \r^2_{L^2} \lesssim \frac{A_3^2\e^4}{\lambda^{2(r-1)}}.
\end{align*}
For the second sum in the time integral in \eqref{nabla r F estimate}, we simply use \eqref{BA F k} if $r'\leq r-1$. Putting everything together gives
\begin{align*}
\l \nabla^r \F \r^2_{L^2} (t) \lesssim C(C_0)\left(  \frac{\lambda^2A_1^2\e^2}{\lambda^{2(r-1)}} +  \frac{A_3^2\e^4}{\lambda^{2(r-1)}} \right) + C(C_0)\int_0^t  \l \nabla^{r}\F \r^2_{L^2}  \d s ,
\end{align*} 
which gives \eqref{BA F n+1 b} after applying Gronwall's lemma. Now that \eqref{BA F n+1 b} is proved, we notice that \eqref{dt nabla r F} implies the first part of \eqref{BA F n+1 c}. For the second part, we apply $\dr_t$ to \eqref{dt nabla r F expression} and estimate the result with \eqref{BA F n+1 b}.
\end{proof}

\subsection{Estimates for $\tBox_g \F$}

In this section, we derive the estimates for $\tBox_{g}\F$.  This is the most crucial and intricate part of the proof, where we deal with the loss of derivatives exhibited in Section \ref{section bad coupling strat}. As explained in Section \ref{section regaining the missing derivative}, we start by using the high-frequency character of $g$ and decompose $\tBox_{g}\F_{\alpha\beta}$ as in
\begin{align}
\tBox_{g}\F_{\alpha\beta} & = \tBox_{g_0}\F_{\alpha\beta} + \left( g^{\mu\nu} - g_0^{\mu\nu} \right)\dr_\mu\dr_\nu \F_{\alpha\beta}\label{decomposition box F n+1}.
\end{align}
We are going to estimate separately the two terms in \eqref{decomposition box F n+1} in order to improve the bootstrap assumption \eqref{BA box F k}. 
\begin{itemize}
\item The term $\tBox_{g_0}\F_{\alpha\beta} $ is estimated in Section \ref{section estim box g0 F} thanks to the transport equation it satisfies. The heart of the proof is the use of various commutator estimates in Lemmas \ref{lem Ar} and \ref{lem A}. We also benefit from $g-g_0=\GO{\la}$ in Lemma \ref{lem C}. The final result of this section is the content of Proposition \ref{prop box g0 F n+1}.
\item The term $\left( g^{\mu\nu} - g_0^{\mu\nu} \right)\dr_\mu\dr_\nu \F_{\alpha\beta}$ is estimated in Section \ref{section (g-g0)d2F}. The heart of the proof is the $\la$ gain from $g-g_0=\GO{\la}$, used in Proposition \ref{prop g1d2F n+1}. Section \ref{section (g-g0)d2F} is concluded by Proposition \ref{prop BA box F n+1} with the final estimate on $\tBox_{g}\F_{\alpha\beta}$.
\end{itemize}

\subsubsection{Estimates for $\Box_{g_0}\F$}\label{section estim box g0 F}

In this section, we prove by a finite strong induction argument on the value $r\in \llbracket 0,4 \rrbracket$ that
\begin{align}
\l \nabla^r \Box_{g_0} \F \r_{L^2} \lesssim \frac{C(A_i)\e^2 + A_1\e}{\lambda^r}.\label{hypothèse de récurrence}
\end{align}
The proof of the induction step or the base case are similar. Therefore, we will prove \eqref{hypothèse de récurrence} for a fixed value of $r$ and assume that it holds for all value strictly less than $r$ in $\llbracket 0,4\rrbracket$. This induction assumption will only be used in the proof of Lemma \ref{lem Br}.

\saut
In order to estimate $\nabla^r\Box_{g_0}\F$, we start by deriving a transport equation from \eqref{eq F}:
\begin{align}
L_0 (\nabla^r\Box_{g_0}\F) = \nabla^r [ L_0 , \Box_{g_0} ] \F + [L_0 , \nabla^r]\Box_{g_0}\F + \nabla^r \Box_{g_0}L_0\F. \label{eq box F}
\end{align}
We give names to the three terms in \eqref{eq box F}:
\begin{align*}
\mathbf{A}_r & \vcentcolon  = \nabla^r [ L_0 , \Box_{g_0} ] \F    ,
\\ \mathbf{B}_r & \vcentcolon  =  [L_0 , \nabla^r]\Box_{g_0}\F  ,
\\ \mathbf{C}_r  & \vcentcolon  = \nabla^r \Box_{g_0}L_0\F   .
\end{align*}
These notations and the induction argument presented above will allow us to use the triangular structure of the equations
\begin{align*}
L_0 ( \Box_{g_0}\F) & = \mathbf{A}_0 + \mathbf{B}_0 + \mathbf{C}_0,
\\ L_0 (\nabla \Box_{g_0}\F) & = \mathbf{A}_1 + \mathbf{B}_1 + \mathbf{C}_1,
\\ & \;\, \vdots
\\ L_0 (\nabla^r\Box_{g_0}\F) & = \mathbf{A}_r + \mathbf{B}_r + \mathbf{C}_r.
\end{align*}
Indeed, during the proof of Lemma \ref{lem Br}, we will estimate $\mathbf{B}_r$ in terms of the $\mathbf{A}_k$, $\mathbf{B}_k$ and $\mathbf{C}_k$ for $0\leq k\leq r-1$ (see \eqref{estim Br inter 2}).

\saut
We start by estimating $\mathbf{A}_r$.
\begin{lem}\label{lem Ar}
We have
\begin{align*}
\l \mathbf{A}_r \r_{L^2} \lesssim \frac{A_1\e + A_3\e^2}{\lambda^r} +  \l \nabla^r \Box_{g_0}\F \r_{L^2} .
\end{align*}
\end{lem}

\begin{proof}
We start with the case $r\geq 1$. We apply \eqref{estim commute null nabla r} to $\F$:
\begin{align*}
\l \nabla^r [L_0,\Box_{g_0}]\F \r_{L^2} \lesssim  \l \nabla^r  \dr L_0 \F \r_{L^2} + \l \nabla^r \Box_{g_0}\F \r_{L^2} + \l \dr \F \r_{H^r} + \l \dr^2 \F \r_{H^{r-1}}   . 
\end{align*}
Thanks to \eqref{BA F k} and $r\leq 4$ we have
\begin{align*}
\l \dr \F \r_{H^r} + \l \dr^2 \F \r_{H^{r-1}} \lesssim \frac{A_1\e}{\lambda^r}.
\end{align*}
Since the equation \eqref{eq F} schematically reads $L_0\F= \F + \h F^{(1)}$ we also have
\begin{align*}
\l \nabla^r  \dr L_0 \F \r_{L^2} & \lesssim \sum_{k=0}^r \left( \l \dr \nabla^k \F \r_{L^2} + \e \l \dr \nabla^k \h \r_{L^2} \right) 
\\&\lesssim \frac{A_1\e + A_3\e^2}{\lambda^r},
\end{align*}
where we used \eqref{BA F k}, \eqref{BA h k} and $r\leq 4$. This concludes the proof of the lemma in the case $r\geq 1$. The case $r=0$ is treated similarly with \eqref{estim commute null} instead of \eqref{estim commute null nabla r}.
\end{proof}

We now estimate $\mathbf{C}_r$. Using \eqref{eq F} we see that
\begin{align}
| \mathbf{C}_r | 	& \lesssim \left| \nabla^r\tBox_{g_0} \Pi_\leq \left( \h \right) F^{(1)}  \right| + \left| \nabla^r\Box_{g_0}\F \right| + | \mathcal{Q}_r|, \label{estim Cr}
\end{align}
where $\mathcal{Q}_r$ contains the terms where at most $r+1$ derivatives (out of the $r+2$ involved in $\nabla^r\tBox_{g_0}$) hit either $\F$ or $\h$.  In this case, the last derivative must therefore hit some background quantity. This implies that $\mathcal{Q}_r$ contains only lower order terms and using the background regularity we easily obtain $\l \mathcal{Q}_0 \r_{L^2} \lesssim A_1\e + A_3\e^2$ and if $1\leq r \leq 4$
\begin{align}
\l \mathcal{Q}_r \r_{L^2} & \lesssim  \sum_{r'\leq r-1}\left( \l \dr^2 \nabla^{r'}\F \r_{L^2}  + \e \l \dr^2 \nabla^{r'}\h \r_{L^2}  \right)\nonumber
\\&\lesssim   \frac{A_1\e + A_3\e^2}{\lambda^r} ,  \label{estim Q}
\end{align}
where the $\e$ in front of $\l \dr^2 \nabla^{r'}\h \r_{L^2}$ comes from $F^{(1)}$ and where we used \eqref{BA F k}, \eqref{BA h k} and $r-1\leq 3$.

\saut
We now turn to the estimates of the first term in \eqref{estim Cr}, i.e $\nabla^r\tBox_{g_0} \Pi_\leq \left( \h \right) $. Forgetting about $\Pi_\leq$ and $g_0$, this term is of the form $\Box \h$, for which we would like to use \eqref{eq h}. In order to do this, we need first to commute $\tBox_{g_0}$ and $\Pi_\leq$ and second to replace $\tBox_{g_0}$ by $\tBox_{g}$. This give the decomposition
\begin{align}
\nabla^r\tBox_{g_0} \Pi_\leq \left( \h \right) & = \nabla^r \left[ \tBox_{g_0},\Pi_\leq \right] \h +\nabla^r \Pi_\leq \left( \tBox_{g}\h\right)  - \nabla^r\Pi_\leq \left(\tBox_{g-g_0} \h \right) .\label{decomposition box P h}
\end{align}
The lemmas \ref{lem A}, \ref{lem B} and \ref{lem C} below estimate the three different terms in the decomposition \eqref{decomposition box P h}.  Since $\nabla^r\tBox_{g_0} \Pi_\leq \left( \h \right)$ is multiplied by $F^{(1)}$ in \eqref{estim Cr}, giving an extra $\e$, we don't need to be very precise in terms of the bootstrap constants $A_i$. For the same reason, the support property of $F^{(1)}$ implies that we can forget about the weight at spacelike infinity for $\h$ or $g$ (except in Lemma \ref{lem B} below).

\begin{lem}\label{lem A}
We have
\begin{align*}
\l \nabla^r \left[ \tBox_{g_0},\Pi_\leq \right] \h \r_{L^2} \lesssim  \frac{A_3\e}{\lambda^r}.
\end{align*}
\end{lem}

\begin{proof}
Thanks to $\left[ \nabla^r,\Pi_\leq \right]=0$ and to the Leibniz rule we obtain
\begin{align*}
\l \nabla^r \left[ \tBox_{g_0},\Pi_\leq \right] \h \r_{L^2} & \lesssim \sum_{r_1+r_2=r}  \l \left[  \nabla^{r_1} g_0 \dr^2 \nabla^{r_2},\Pi_\leq  \right] \h   \r_{L^2}.
\end{align*}
We first study the case where $r_2\leq 3$.  We use $\left[ \dr^2\nabla^{r_2},\Pi_\leq \right]=0$ and \eqref{commute BCD}
\begin{align*}
\l \left[  \nabla^{r_1} g_0 \dr^2 \nabla^{r_2},\Pi_\leq  \right] \h   \r_{L^2} & = \l \left[  \nabla^{r_1} g_0 ,\Pi_\leq  \right] \dr^2 \nabla^{r_2}\h   \r_{L^2} 
\\&\lesssim \lambda \l \nabla^{r_1+1} g_0 \r_{L^\infty} \l \dr^2 \nabla^{r_2} \h \r_{L^2}
\\&\lesssim C(C_0) \lambda  \times \frac{A_3 \e}{\lambda^{r_2+1}}
\\&\lesssim C(C_0) \frac{A_3\e}{\lambda^r},
\end{align*}
where we used the background regularity and \eqref{BA h k} (with $r_2\leq 3$). If now $r_2 =4$,  then $r_1=0$ and $r=4$. Since we can't estimate $\dr^2\nabla^{4} \h$, we need to gain one spatial derivative, and for this purpose we use \eqref{commute bis} instead of \eqref{commute BCD} to obtain
\begin{align*}
 \l \left[   g_0 \dr^2 \nabla^{4},\Pi_\leq  \right] \h  \r_{L^2} & =  \l \left[   g_0 ,\Pi_\leq  \right] \dr^2 \nabla^{4} \h   \r_{L^2}
 \\&\lesssim C(C_0) \l \dr^2 \nabla^3 \h \r_{L^2}
 \\&\lesssim C(C_0) \frac{A_3\e}{\lambda^4},
\end{align*}
where we also used the background regularity and \eqref{BA h k}.
\end{proof}

In the next lemma, we estimate the term $\nabla^r  \Pi_\leq \left( \tBox_{g}\h\right)$ in \eqref{decomposition box P h} and therefore we need to estimate the $L^2$ norm of
\begin{align}
\nabla^r  \Pi_\leq \left( \text{RHS of \eqref{eq h}}\right).\label{nabla pi RHS eq h}
\end{align}
In the forthcoming Section \ref{section estimate h} we will improve \eqref{BA h k} with the energy estimate for the wave operator. Therefore, we will apply $\nabla^r$ to \eqref{eq h} and will need an estimate for the \textit{weighted} $L^2$ norm of \eqref{nabla pi RHS eq h} \textit{without} the projector $\Pi_\leq$ (in addition to a commutator estimate, see Proposition \ref{prop BA h n+1}). This explains why in the next lemma we project ourself and first prove \eqref{lem B estim 1} and then get \eqref{lem B estim 2} as a consequence.

\begin{lem}\label{lem B}
We have
\begin{align}
\l\nabla^r   \tBox_{g}\h \r_{L^2_{\delta+r+1}} \lesssim  \frac{C(A_i)\e^2 + A_2\e}{\lambda^r}, \label{lem B estim 1}
\end{align}
which in particular implies
\begin{align}
\l\nabla^r  \Pi_\leq \left( \tBox_{g}\h\right) \r_{L^2} \lesssim C(A_i) \frac{\e}{\lambda^r}. \label{lem B estim 2}
\end{align}
\end{lem}

\begin{proof}
We only need to prove \eqref{lem B estim 1}. As explained above, we rely on the equation satisfied by $\h$, that is \eqref{eq h}.  We have
\begin{align}
\l \nabla^r \tBox_g \h \r_{L^2_{\delta+1+r}} & \lesssim \sum_{r_1+r_2=r}\frac{\l \nabla^{r_1}\tBox_g\F\r_{L^2}}{\lambda^{r_2}} + \lambda\l\nabla^r \tBox_g g^{(3)} \r_{L^2} \label{estim nabla r box h}
\\&\quad + \frac{\e}{\lambda}\sum_{r_1+r_2=r}\frac{  \l \nabla^{r_1} \Pi_\geq \left(\h  \right)  \r_{L^2}  }{\lambda^{r_2}} + \l \nabla^r \mathcal{R}(g)\r_{L^2_{\delta+1+r}}\nonumber
\\& = \vcentcolon I + II + III + IV,\nonumber
\end{align}
where the $\e$ in front of $\Pi_\geq \left(\h  \right)$ comes from $F^{(1)}$.  For $I$ we simply use \eqref{BA box F k} to get 
\begin{align}
I \lesssim \frac{A_2\e}{\lambda^r}.\label{lem B I}
\end{align}
For $II$, we first note that $g^{(3)}$ is oscillating and therefore one looses potentially two powers of $\lambda$ in $\tBox_g g^{(3)}$. However this is not the case since if $\mathrm{T}$ is a trigonometric function and $f$ a scalar function then we have schematically
\begin{align}
\tBox_g\left( \lambda \mathrm{T}\left(\frac{u_0}{\lambda}\right)f\right) & = \frac{1}{\lambda}g^{-1}(\d u_0, \d u_0) f + g \dr f + \lambda \tBox_g f \nonumber
\\& = g \dr f + \lambda \left( \tBox_g f + g f \right),\label{calcul}
\end{align}
where we used $g^{-1}(\d u_0, \d u_0)=\GO{\lambda^2}$ (see \eqref{g(du_0,du_0)}) and did not write the trigonometric functions (see \eqref{box fct trigo bis 2} for an exact computation). Therefore, we use the decomposition of Lemma \ref{lem g3} and obtain from \eqref{calcul} on the one hand
\begin{align*}
\lambda\l\nabla^r \tBox_g g^{(3,\mathrm{BG})} \r_{L^2}  \lesssim \frac{C(A_i)\e^2 + \e}{\lambda^r},
\end{align*}
and on the other hand
\begin{align*}
\lambda\l\nabla^r \tBox_g g^{(3)}(\F) \r_{L^2} &\lesssim \sum_{r_1+r_2=r}\l \nabla^{r_1}g \dr \nabla^{r_2}\F\r_{L^2} + \lambda \l \nabla^r\tBox_g\F\r_{L^2} ,
\end{align*}
where for the sake of clarity we replaced each $g^{(3,\mathrm{T})}(\F)$ by $\F$ (see \eqref{g3T(F)}).  We study the sum. If $r_1\leq 3$, then we use \eqref{d g-g0 Linfini} and $\l g_0 \r_{L^\infty}\lesssim 1$ and \eqref{BA F k} to write 
\begin{align*}
\l \nabla^{r_1}g\dr \nabla^{r_2}\F \r_{L^2} \lesssim \frac{C(A_i)\e^2 + A_2\e}{\lambda^{r-1}}.
\end{align*}
If $r_1=4$, then we use the decomposition \eqref{splitting g n} and write
\begin{align*}
\l \nabla^4 g \dr \F \r_{L^2} & \lesssim \l \nabla^4 g_{\mathrm{BG}} \dr \F \r_{L^2}  +  \l \nabla^4 \tilde{g} \dr \F \r_{L^2} 
\\&\lesssim \l \nabla^4 g_{\mathrm{BG}} \r_{L^\infty} \l \dr \F \r_{L^2}  +  \l \nabla^4 \tilde{g} \r_{L^2} \l \dr \F \r_{L^\infty}
\\&\lesssim \frac{A_1\e^2}{\lambda^3} + \frac{C(A_i)\e^2}{\lambda^4} ,
\end{align*}
where we used \eqref{BA F k}. Using in addition \eqref{BA box F k} we obtain
\begin{align*}
\lambda\l\nabla^r \tBox_g g^{(3)}(\F) \r_{L^2} &\lesssim \frac{C(A_i)\e^2 + A_2\e}{\lambda^r}.
\end{align*}
We have proved that
\begin{align}
II \lesssim \frac{C(A_i)\e^2 + A_2\e}{\lambda^r}.\label{lem B II}
\end{align}
For $III$, we use \eqref{P high} to obtain
\begin{align}
III & \lesssim \e \sum_{r_1+r_2=r} \frac{  \l \nabla \nabla^{r_1} \h  \r_{L^2}  }{\lambda^{r_2}} \lesssim \frac{A_3\e^2}{\lambda^r},\label{lem B III}
\end{align}
where we used \eqref{BA h k} and $r_1\leq 4$. For $IV$, we simply use \eqref{d R g L2}. Together with \eqref{lem B I}, \eqref{lem B II} and \eqref{lem B III} this concludes the proof.
\end{proof}

\begin{lem}\label{lem C}
We have
\begin{align*}
\l \nabla^r\Pi_\leq \left(\Box_{g-g_0} \h \right) \r_{L^2} \lesssim  C(A_i) \frac{\e^2}{\lambda^r}  .
\end{align*}
\end{lem}

\begin{proof}
We first treat the case $r=0$. We forget about the projector $\Pi_\leq$ and use \eqref{d g-g0 Linfini} and \eqref{BA h k} to obtain
\begin{align*}
\l  \Pi_\leq \left(\left(g-g_0\right)\dr^2 \h \right)  \r_{L^2} & \lesssim \l g-g_0 \r_{L^\infty} \l \dr^2 \h \r_{L^2}
\\&\lesssim \lambda A_3\e \times \frac{A_3\e}{\lambda}
\\&\lesssim C(A_i) \e^2.
\end{align*}
If $r\geq 1$ we use \eqref{P low} to get rid of one derivative and the Leibniz rule to obtain
\begin{align*}
\l  \nabla^r\Pi_\leq \left(\left(g-g_0\right)\dr^2 \h \right) \r_{L^2} & \lesssim \frac{1}{\lambda} \l  \nabla^{r-1} \left(\left(g-g_0\right)\dr^2 \h \right) \r_{L^2} 
\\&\lesssim \frac{1}{\lambda} \sum_{r_1+r_2=r-1} \l \nabla^{r_1}  \left(g-g_0\right)\r_{L^\infty} \l\dr^2\nabla^{r_2} \h \r_{L^2}
\\&\lesssim \frac{1}{\lambda}\times \frac{A_3\e}{\lambda^{r_1-1}}\times\frac{A_3\e}{\lambda^{r_2+1}}
\\&\lesssim C(A_i) \frac{\e^2}{\lambda^{r}},
\end{align*}
where we again use \eqref{d g-g0 Linfini} in addition to \eqref{BA h k} and the fact that $r_i\leq 3$.
\end{proof}

Putting together Lemmas \ref{lem A}, \ref{lem B} and \ref{lem C} and the estimate \eqref{estim Q} we obtain

\begin{lem}\label{lem Cr}
We have
\begin{align*}
\l \mathbf{C}_r \r_{L^2} \lesssim \frac{A_1\e + A_3\e^2}{\lambda^r} +  \l \nabla^r \Box_{g_0}\F \r_{L^2} .
\end{align*}
\end{lem}

We now estimate $\mathbf{B}_r$. 

\begin{lem}\label{lem Br}
We have
\begin{align*}
\l \mathbf{B}_r \r_{L^2} \lesssim \frac{A_1\e + C(A_i)\e^2}{\lambda^r} +  \l \nabla^r \Box_{g_0}\F \r_{L^2} .
\end{align*}
\end{lem}

\begin{proof}
Note that $\mathbf{B}_0=0$ so we can assume $r\geq 1$. We first use \eqref{commut L0 nabla r} to obtain
\begin{align}
\l \mathbf{B}_r \r_{L^2} & \lesssim \sum_{k=0}^{r-1} \l  \nabla^{k+1} \Box_{g_0}\F \r_{L^2} + \sum_{k=0}^{r-1} \l \dr_t \nabla^k \Box_{g_0}\F \r_{L^2}, \label{estim Br inter}
\end{align}
where we distinguished between time and spatial derivatives. For the first sum in \eqref{estim Br inter}, we note that if $k\leq r-2$ we have 
\begin{align*}
 \l  \nabla^{k+1} \Box_{g_0}\F \r_{L^2} \lesssim \sum_{\ell=0}^{k+1} \l \dr \nabla^{\ell}\F \r_{L^2} \lesssim \frac{A_1\e}{\lambda^r},
\end{align*}
where we used \eqref{BA F k}. This gives
\begin{align}
\sum_{k=0}^{r-1} \l  \nabla^{k+1} \Box_{g_0}\F \r_{L^2} \lesssim \l  \nabla^{r} \Box_{g_0}\F \r_{L^2} + \frac{A_1\e}{\lambda^r}.\label{estim Br first sum}
\end{align}
For the second sum in \eqref{estim Br inter} we need to proceed differently since we didn't include time derivative of $\Box_{g_0}\F$ in the bootstrap assumptions. We use the equation satisfied by $\nabla^k \Box_{g_0}\F$ which schematically reads
\begin{align*}
\dr_t \nabla^k \Box_{g_0}\F = \nabla^{k+1} \Box_{g_0}\F + \mathbf{A}_k + \mathbf{B}_k + \mathbf{C}_k.
\end{align*}
Therefore, if we sum these equalities from $0$ to $r-1$ and use \eqref{estim Br first sum} again we obtain
\begin{align}
\l \mathbf{B}_r \r_{L^2} & \lesssim \l  \nabla^{r} \Box_{g_0}\F \r_{L^2} + \frac{A_1\e}{\lambda^r}  +   \sum_{k=0}^{r-1} \left( \l \mathbf{A}_k \r_{L^2} + \l \mathbf{B}_k \r_{L^2} + \l \mathbf{C}_k  \r_{L^2}   \right) .\label{estim Br inter 2}
\end{align}
We can now use Lemmas \ref{lem Ar} and \ref{lem Cr}:
\begin{align*}
\sum_{k=0}^{r-1} \left( \l \mathbf{A}_k \r_{L^2}  + \l \mathbf{C}_k  \r_{L^2}   \right)& \lesssim \sum_{k=0}^{r-1} \left( \frac{A_1\e + A_3\e^2}{\lambda^k} +  \l \nabla^k \Box_{g_0}\F \r_{L^2} \right)
\\&\lesssim  \frac{A_1\e + C(A_i)\e^2}{\lambda^{r-1}},
\end{align*}
where we used the induction assumption \eqref{hypothèse de récurrence} for $\l \nabla^k \Box_{g_0}\F \r_{L^2}$ with $k\leq r-1$. Since we can assume $\lambda\leq \e$ this shows that 
\begin{align*}
\sum_{k=0}^{r-1} \left( \l \mathbf{A}_k \r_{L^2}  + \l \mathbf{C}_k  \r_{L^2}   \right) \lesssim \frac{C(A_i)\e^2}{\lambda^r}.
\end{align*}
Going back to \eqref{estim Br inter 2} this gives 
\begin{align*}
\l \mathbf{B}_r \r_{L^2} & \lesssim \l  \nabla^{r} \Box_{g_0}\F \r_{L^2} + \frac{A_1\e + C(A_i)\e^2}{\lambda^r}  +   \sum_{k=0}^{r-1}  \l \mathbf{B}_k \r_{L^2} .
\end{align*}
Thanks to an iterative argument, this shows that
\begin{align*}
\l \mathbf{B}_r \r_{L^2} & \lesssim \l  \nabla^{r} \Box_{g_0}\F \r_{L^2} + \frac{A_1\e + C(A_i)\e^2}{\lambda^r},
\end{align*}
where we again used the induction assumption \eqref{hypothèse de récurrence} for $\l \nabla^k \Box_{g_0}\F \r_{L^2}$ with $k\leq r-1$ and the fact that $\mathbf{B}_0=0$.
\end{proof}

Putting everything together, we finally obtain the following proposition.

\begin{prop}\label{prop box g0 F n+1}
The following estimate holds
\begin{align}
\max_{r\in\llbracket 0,4\rrbracket}\lambda^r\l \nabla^r \tBox_{g_0}\F \r_{L^2} & \lesssim C(A_i)\e^2+ A_1\e  .\label{estim box g0 F n+1}
\end{align}
\end{prop}

\begin{proof}
We want to apply Lemma \ref{lem L0 h =f} to the equation \eqref{eq box F}. We estimate its RHS, by putting together Lemmas \ref{lem Ar}, \ref{lem Br} and \ref{lem Cr}:
\begin{align*}
\l \text{RHS of \eqref{eq box F}} \r_{L^2} & \lesssim \l \nabla^r \Box_{g_0} \F  \r_{L^2} +  \frac{C(A_i)\e^2+ A_1\e}{\lambda^r} .
\end{align*}
Lemma \ref{lem L0 h =f} then implies 
\begin{align}
\l \nabla^r \Box_{g_0}\F \r_{L^2}(t) \lesssim \l \nabla^r \Box_{g_0}\F \r_{L^2} (0) +  \frac{C(A_i)\e^2+ A_1\e}{\lambda^r} +  \int_0^t \l \nabla^r \Box_{g_0} \F  \r_{L^2} \d s .\label{estim finale nabla r box g0 F}
\end{align}
We estimate the initial data term in \eqref{estim finale nabla r box g0 F}. Since $\F\restriction{\Sigma_0}=0$, the only cases we need to consider are $\dr_t\nabla^{r+1}\F$ and $\dr^2_t\nabla^r \F$.  The initial data for such terms are obtained by differentiating \eqref{ID dt F}. Using the background regularity this gives
\begin{align*}
\l \dr_t\nabla^{r+1}\F \r_{L^2}(0) & \lesssim  \e \sum_{r'\leq r+1} \l \nabla^{r'}  \Pi_\leq \left( \h \right) \r_{L^2}(0),
\end{align*}
where the $\e$ comes from $F^{(1)}$ in front of $\Pi_\leq \left( \h \right)$ in \eqref{ID dt F}. Using $r\leq 4$ and \eqref{BA h k} we obtain
\begin{align*}
\l \dr_t\nabla^{r+1}\F \r_{L^2}(0) & \lesssim \frac{A_3\e^2}{\lambda^r}.
\end{align*}
The proof is similar for $\dr^2_t\nabla^r \F$ and we omit the details. Going back to \eqref{estim finale nabla r box g0 F} we have proved that
\begin{align*}
\l \nabla^r \Box_{g_0}\F \r_{L^2} \lesssim \frac{C(A_i)\e^2+ A_1\e}{\lambda^r} + \int_0^t \l \nabla^r \Box_{g_0} \F  \r_{L^2} \d s.
\end{align*}
Gronwall's inequality then implies that
\begin{align*}
\l \nabla^r \Box_{g_0}\F \r_{L^2} \lesssim \frac{C(A_i)\e^2+ A_1\e}{\lambda^r} .
\end{align*}
This concludes the induction proving that \eqref{hypothèse de récurrence} holds for all values of $r\in \llbracket 0,4 \rrbracket$. Since $\Box_{g_0}\F$ and $\tBox_{g_0}\F$ only differs by lower order terms, this also concludes the proof of \eqref{estim box g0 F n+1}.
\end{proof}

\subsubsection{Estimates for $(g-g_0)\dr^2\F$}\label{section (g-g0)d2F}

We now estimate the second term in \eqref{decomposition box F n+1}.

\begin{prop}\label{prop g1d2F n+1}
For $r\in \llbracket 0,4\rrbracket$ we have
\begin{align}
\max_{r\in\llbracket 0,4\rrbracket}\lambda^r\l \nabla^r\left( \left( g - g_0 \right)\dr^2 \F \right)  \r_{L^2} \lesssim C(A_i) \e^2 \label{BA box F n+1 b}.
\end{align}
\end{prop}

\begin{proof}
We start by noticing that \eqref{BA F k} imply
\begin{align}
\max_{r\in\llbracket 0,4\rrbracket}\lambda^{r+1}\l \dr^2 \nabla^r \F \r_{L^2} \lesssim A_1 \e, \label{estim d2 F n+1 L2}
\end{align}
which in turn implies
\begin{align}
\max_{r\in\llbracket 0,2\rrbracket}\lambda^{r+3}\l \dr^2 \nabla^r \F \r_{L^\infty} \lesssim A_1 \e, \label{estim d2 F n+1 Linfini}
\end{align}
with the usual $H^2\xhookrightarrow{} L^\infty$ embedding. To prove the proposition, we start with the Leibniz rule:
\begin{align*}
\l \nabla^r\left(  \left( g - g_0 \right)\dr^2 \F \right)  \r_{L^2} & \lesssim \sum_{r_1+r_2=r}\l \nabla^{r_1}  \left( g - g_0 \right) \dr^2 \nabla^{r_2}\F  \r_{L^2}.
\end{align*}
If $r_1\leq 3$, then we use \eqref{d g-g0 Linfini} and \eqref{estim d2 F n+1 L2} to obtain
\begin{align*}
\l \nabla^{r_1}  \left( g - g_0 \right) \dr^2 \nabla^{r_2}\F  \r_{L^2} & \lesssim \l\nabla^{r_1}  \left( g - g_0 \right)  \r_{L^\infty} \l  \dr^2 \nabla^{r_2}\F \r_{L^2}
\\&\lesssim \frac{\e}{\lambda^{r_1-1}} \times \frac{A_1\e}{\lambda^{r_2+1}}
\\&\lesssim \frac{A_1\e^2}{\lambda^r},
\end{align*}
where we also used the fact that $r_2\leq 4$. 

\saut
If now $r_1=4$, we have $r=4$ ans $r_2=0$.  In this case we use the decomposition \eqref{splitting g n} to write
\begin{align}
\l \nabla^{4}  \left( g - g_0 \right) \dr^2 \F  \r_{L^2} & \lesssim \l \nabla^{4}  \left( g_\BG - g_0 \right) \dr^2 \F  \r_{L^2} + \l \nabla^{4}  \tilde{g} \dr^2 \F  \r_{L^2}.\label{estim nabla4}
\end{align}
For the first term in \eqref{estim nabla4} we put the background quantities in $L^\infty$:
\begin{align}
\l \nabla^{4}  \left( g_\BG - g_0 \right) \dr^2 \F  \r_{L^2} & \lesssim \l \nabla^{4}  \left( g_\BG - g_0 \right) \r_{L^\infty} \l \dr^2 \F \r_{L^2}\label{estim nabla4 a}
\\&\lesssim C(C_0) \frac{\e}{\lambda^3} \times \frac{A_1\e}{\lambda}\nonumber
\\&\lesssim C(C_0) \frac{A_1\e^2}{\lambda^4}, \nonumber
\end{align}
where we used \eqref{d4 gBG-g0 Linfini} and \eqref{estim d2 F n+1 L2}. For the second term in \eqref{estim nabla4} we use \eqref{def gtilde} to decompose it further:
\begin{align}
\l \nabla^{4}  \tilde{g} \dr^2 \F  \r_{L^2} & \leq \lambda^2 \l \nabla^{4}  \left(\sin\left(\frac{u_0}{\lambda}\right)\F  \right) \dr^2 \F  \r_{L^2} + \lambda^2 \l \nabla^{4}  \h \dr^2 \F  \r_{L^2} + \lambda^3 \l \nabla^{4}  g^{(3)}(\F) \dr^2 \F  \r_{L^2}.\label{d4 gtilde d2 F}
\end{align}
Using \eqref{BA h k} and \eqref{estim d2 F n+1 Linfini} we first obtain
\begin{align}
\lambda^2 \l \nabla^{4}  \h \dr^2 \F  \r_{L^2} & \lesssim \lambda^2 \l \nabla^{4}  \h \r_{L^2} \l \dr^2 \F  \r_{L^\infty} \label{estim nabla4 b1}
\\&\lesssim \lambda^2 \times \frac{A_3\e}{\lambda^3}  \times \frac{A_1\e}{\lambda^3}\nonumber
\\&\lesssim \frac{A_1A_3\e^2}{\lambda^4}.\nonumber
\end{align}
Now, by using the Leibniz rule on the product $\sin\left(\frac{u_0}{\lambda}\right)\F $ and putting aside the case where 4 derivatives hit $\sin\left(\frac{u_0}{\lambda}\right)$ we obtain
\begin{align}
\lambda^2 \l \nabla^{4}  \left(\sin\left(\frac{u_0}{\lambda}\right)\F  \right) \dr^2 \F  \r_{L^2} & \lesssim \frac{1}{\lambda^{2}} \l \F   \dr^2 \F  \r_{L^2}  + \sum_{\substack{a+b=4\\a\leq 3}} \frac{1}{\lambda^{a-2}} \l \nabla^b\F   \dr^2 \F  \r_{L^2}\label{estim nabla4 b}
\\&\lesssim \frac{1}{\lambda^{2}} \l \F \r_{L^\infty}  \l  \dr^2 \F  \r_{L^2}  + \sum_{\substack{a+b=4\\a\leq 3}} \frac{1}{\lambda^{a-2}} \l \nabla^b\F \r_{L^2}  \l  \dr^2 \F  \r_{L^\infty}\nonumber
\\&\lesssim \frac{1}{\lambda^{2}}  \times \frac{A_1\e}{\lambda}  \times \frac{A_1\e}{\lambda} + \sum_{\substack{a+b=4\\a\leq 3}} \frac{1}{\lambda^{a-2}} \times \frac{A_1\e}{\lambda^{b-1}} \times \frac{A_1\e}{\lambda^3}\nonumber
\\&\lesssim \frac{A_1^2\e^2}{\lambda^4}, \nonumber
\end{align}
where we used \eqref{estim d2 F n+1 L2} and \eqref{estim d2 F n+1 Linfini}. We treat the better term $\lambda^3 \l \nabla^{4}  g^{(3)}(\F) \dr^2 \F  \r_{L^2}$ in \eqref{d4 gtilde d2 F} similarly. Putting \eqref{estim nabla4}, \eqref{estim nabla4 a}, \eqref{estim nabla4 b1} and \eqref{estim nabla4 b}  together gives
\begin{align*}
\l \nabla^{4}  \left( g - g_0 \right) \dr^2 \F  \r_{L^2} \lesssim C(A_i) \frac{\e^2}{\lambda^4}.
\end{align*}
This concludes the proof of the proposition.
\end{proof}

Together with the decomposition \eqref{decomposition box F n+1}, the results of Propositions \ref{prop box g0 F n+1} and \ref{prop g1d2F n+1} finally imply 
\begin{prop}\label{prop BA box F n+1}
The following estimate holds
\begin{align*}
\max_{r\in\llbracket 0,4\rrbracket}\lambda^r\l  \nabla^r \tBox_{g} \F \r_{L^2} & \lesssim C(A_i)\e^2 + A_1\e.
\end{align*}
\end{prop}

\subsection{Estimates for $\h$}\label{section estimate h}

\begin{prop}\label{prop BA h n+1}
We have
\begin{align}
\max_{r\in\llbracket 0,4 \rrbracket} \lambda^{2r} E_{\delta+1+r}\left( \nabla^r \h \right) & \lesssim e^{1+C(A_i)\e}\left(  C(A_i)\e^4 + \left( C^2_{\mathrm{in}} + A_2^2 \right) \e^2 \right) , \label{BA h n+1}
\\ \max_{r\in\llbracket 0,3 \rrbracket}\lambda^{r+1} \l \dr_t^2 \nabla^r \h \r_{L^2_{\delta+r+2}} & \lesssim e^{1+C(A_i)\e}\left(  C(A_i)\e^2 + \left( C_{\mathrm{in}} + A_2 \right) \e \right) .\label{BA dt2 h n+1}
\end{align}
\end{prop}

\begin{proof}
We start with the proof of \eqref{BA h n+1}.  We commute $\nabla^r$ and $\tBox_g$ (with $r\in\llbracket 0,4 \rrbracket$) to obtain the equation satisfied by $\nabla^r \h$ from \eqref{eq h}:
\begin{align}
\tBox_g \nabla^r \h & = \left[ \tBox_g , \nabla^r \right] \h   + \nabla^r \tBox_g \h .\label{eq nabla r h}
\end{align}
From the energy estimate in Lemma \ref{lem EE} we need to estimate the $L^2_{\delta+1+r}$ norm of the RHS of \eqref{eq nabla r h}. We estimate the commutator first. We have
\begin{align*}
\l \left[ \tBox_g , \nabla^r \right] \h \r_{L^2_{\delta+1+r}} & \lesssim  \sum_{r_1+r_2 = r-1} \l \nabla^{r_1+1}g \dr^2 \nabla^{r_2} \h \r_{L^2_{\delta+1+r}}.
\end{align*}
We have $r_1\leq 3$. If $r_1\leq 2$, then mimicking the proof of Lemma \ref{lem g BG gtilde} gives $\l \nabla^{r_1+1} g \r_{L^\infty_{r_1}}\lesssim \frac{ A_3\e}{\lambda^{r_1}}$ which in turn implies
\begin{align*}
\l\nabla^{r_1+1}g  \dr^2 \nabla^{r_2} \h \r_{L^2_{\delta+1+r}} & \lesssim \l \nabla^{r_1+1}g \r_{L^\infty_{r_1}} \l  \dr^2 \nabla^{r_2} \h  \r_{L^2_{\delta+2+r_2}}
\\&\lesssim  \frac{ A_3\e}{\lambda^{r_1}} \times \frac{A_3\e}{\lambda^{r_2+1}}
\\&\lesssim C(A_i)\frac{\e^2}{\lambda^r},
\end{align*}
where we used \eqref{BA h k}. If $r_1=3$, then $r=4$ and $r_2=0$ we use the decomposition \eqref{splitting g n}:
\begin{align*}
\l \nabla^{4}g \dr^2  \h \r_{L^2_{\delta+5}} & \lesssim \l \nabla^{4}g_{\mathrm{BG}} \dr^2  \h \r_{L^2_{\delta+5}} + \l \nabla^{4}\tilde{g} \dr^2  \h \r_{L^2_{\delta+5}} .
\end{align*}
For the first term, we put $\nabla^{4}g_{\mathrm{BG}}$ in $L^\infty_3$ as in the proof of Lemma \ref{lem g BG gtilde} and use \eqref{BA h k}. For the second term, we do as in the proof of Proposition \ref{prop g1d2F n+1}, see \eqref{estim nabla4 b1} and \eqref{estim nabla4 b}. This proves that the commutator in \eqref{eq nabla r h} satisfies
\begin{align}
\l \left[ \tBox_g , \nabla^r \right] \h \r_{L^2_{\delta+1+r}} & \lesssim C(A_i)\frac{\e^2}{\lambda^r}.\label{estim h commute}
\end{align}
We now estimate the second term in \eqref{eq nabla r h}. This has already been done in Lemma \ref{lem B}, we recall \eqref{lem B estim 1}:
\begin{align*}
\l\nabla^r   \tBox_{g}\h \r_{L^2_{\delta+r+1}} \lesssim  \frac{C(A_i)\e^2 + A_2\e}{\lambda^r} .
\end{align*}
Together with \eqref{estim h commute} this implies that
\begin{align}
\l\tBox_{g}\nabla^r   \h \r_{L^2_{\delta+r+1}} \lesssim  \frac{C(A_i)\e^2 + A_2\e}{\lambda^r} . \label{estim box nabla h}
\end{align}
We now apply Lemma \ref{lem EE} to \eqref{eq nabla r h} with $\sigma=\delta+r+1$. Thanks to \eqref{d g-g0 Linfini}, we have $\l \dr g \r_{L^\infty}\leq C(A_i)\e$ and Lemma \ref{lem EE} gives for $t\in[0,T]$ (after multiplication by $\lambda^{2r}$)
\begin{align*}
\lambda^{2r} &E_{\delta+r+1}(\nabla^r\h)(t) 
\\& \lesssim \lambda^{2r} E_{\delta+r+1}(\nabla^r\h)(0)  + (1+C(A_i)\e) \int_0^t \lambda^{2r} E_{\delta+r+1}(\nabla^r\h)(s)\d s + \int_0^t \lambda^{2r} \l  \tBox_g \nabla^r \h  \r^2_{L^2_{\delta+r+1}}(s)\d s
\\& \lesssim \lambda^{2r} E_{\delta+r+1}(\nabla^r\h)(0) + C(A_i)\e^4 + A_2^2 \e^2  + (1+C(A_i)\e) \int_0^t \lambda^{2r} E_{\delta+r+1}(\nabla^r\h)(s)\d s .
\end{align*}
We apply Gronwall's lemma to obtain
\begin{align*}
\lambda^{2r} E_{\delta+r+1}(\nabla^r\h)(t) \lesssim e^{1+C(A_i)\e}\left( \lambda^{2r} E_{\delta+r+1}(\nabla^r\h)(0) + C(A_i)\e^4 + A_2^2 \e^2 \right).
\end{align*}
In order to estimate $E_{\delta+r+1}(\nabla^r\h)(0)$, we use \eqref{estim h initial}:
\begin{align*}
\lambda^{2r} E_{\delta+r+1}(\nabla^r\h)(0)\leq C_{\mathrm{in}}^2\e^2.
\end{align*}
This concludes the proof of \eqref{BA h n+1}.

\saut
We now turn to the proof of \eqref{BA dt2 h n+1}. For this we write the operator $\tBox_g$ in the following schematic way:
\begin{align*}
\tBox_g & = g^{tt}\dr_t^2 + g\dr\nabla.
\end{align*}
Thanks to the background regularity we can divide by $g^{tt}$ and therefore obtain for $r\leq 3$:
\begin{align*}
\l \dr_t^2 \nabla^r \h \r_{L^2_{\delta+r+2}} & \lesssim \sum_{r_1+r_2=r} \l \nabla^{r_1} g \nabla^{r_2}\tBox_g \h \r_{L^2_{\delta+r+2}} +  \sum_{r_1+r_2+r_3=r}\l \nabla^{r_1}g\nabla^{r_2}g \dr\nabla^{r_3+1}\h \r_{L^2_{\delta+r+2}}
\\& =\vcentcolon I + II.
\end{align*}
For $I$, we use \eqref{d g-g0 Linfini}, the background regularity and \eqref{lem B estim 1}:
\begin{align*}
I \lesssim \frac{C(A_i)\e^2 + A_2\e}{\lambda^r}.
\end{align*}
For $II$ we proceed similarly but use \eqref{BA h n+1} instead:
\begin{align*}
II \lesssim \frac{e^{1+C(A_i)\e}\left(  C(A_i)\e^2 + \left( C_{\mathrm{in}} + A_2 \right) \e \right)}{\lambda^{r+1}}.
\end{align*}
This concludes the proof.
\end{proof}

\subsection{Conclusion of the bootstrap argument}

Looking at Propositions \ref{prop BA F n+1}-\ref{prop BA box F n+1}-\ref{prop BA h n+1}, we see that a choice of constants $A_i$ such that $A_1\ll A_2\ll A_3$ together with $\e$ and $\lambda$ small enough compared to 1 (but with $\e$ still independent from $\lambda$) allows us to improve the bootstrap assumptions \eqref{BA F k}-\eqref{BA box F k}-\eqref{BA h k}, as long as the time of existence of $(\F,\h)$ is less than 1. This shows that the time of existence of the solution $(\F,\h)$ is actually equal to 1. We have proved the following theorem.

\begin{thm}\label{theo reduced system}
Given initial data as in Corollary \ref{coro ID} and if $\e$ and $\lambda$ are small enough (with $\e$ independent from $\lambda$), there exists a solution $(\F,\h) $ to the reduced system \eqref{eq F}-\eqref{eq h} on $[0,1]\times \R^3$. Moreover there exists a numerical constant $C>0$ such that \eqref{BA F k}-\eqref{BA box F k}-\eqref{BA h k} hold with the $A_i$ replaced by $C$.
\end{thm}

The estimates that $(\F,\h)$ satisfied, along with the ones satisfied by $F^{(1)}$, $F^{(2,1)}$ and $F^{(2,2)}$ (see Theorem \ref{theo bg system}) prove the main estimate \eqref{estim théorème} in Theorem \ref{theo main} if one defines 
\begin{align*}
\tilde{\h}_\lambda = \h + \lambda g^{(3)}.
\end{align*}
It only remains to prove that $g$ given by \eqref{g} solves the Einstein vacuum equations.

\section{Propagation of polarization and gauge conditions}\label{section solving EVE}

By proving Theorems \ref{theo bg system} and \ref{theo reduced system}, we construct a metric $g$ given by \eqref{g} solving both the background and the reduced systems on the manifold $[0,1]\times \R^3$. In this section, we conclude the proof of Theorem \ref{theo main} by proving that $g$ is actually a solution to the Einstein vacuum equations, i.e that $R_{\mu\nu}(g)=0$ on the same manifold. We will show in Section \ref{section expansion bianchi} that $R_{\mu\nu}(g)$ contains only polarization condition tensors $V^{(2,i)}$ or the generalised wave gauge term $\Upsilon^\rho$. The fact that these terms vanish will be proved in Sections \ref{section propa pola 23} using the contracted Bianchi identities, i.e the fact that the Einstein tensor of any Lorentzian metric is divergence free.

\subsection{Algebraic properties of $F^{(1)}$ and $\F$}\label{section algebraic properties F1}

In this section we prove that $F^{(1)}$ satisfies \eqref{pola F1 theo} and \eqref{energie F1 theo}. To this purpose, we define
\begin{align}
\underline{\zeta}_A &= \frac{1}{2}g(\D_{L_0}\Lb_0,e_A).\label{notation zeta}
\end{align}
Thanks to the properties of the background null frame, we obtain the following decomposition of $\D_{L_0}X$ for $X$ a vector field of the null frame:
\begin{align}
\D_{L_0}L_0 & = 0 , \label{D L L}
\\ \D_{L_0}\Lb_0 & =  2\underline{\zeta}_A e_A,\label{D L Lb}
\\ \D_{L_0} e_A & = \underline{\zeta}_AL_0 + \nabb_{L_0}e_A, \label{D L A}
\end{align}
where $\nabb_{L_0}X$ is the projection of $\D_{L_0}X$ onto $TP_{t,u}$. Note that the regularity of the background stated in Section \ref{section BG} is enough to ensure that $\underline{\zeta}_A$ as well as $g_0(\nabb_{L_0} e_A, e_B)$ are bounded on $\mathcal{M}$.

\begin{remark}
The notation $\underline{\zeta}$ in \eqref{notation zeta} is standard when working with null frames and goes back to \cite{Christodoulou1993}. The equations \eqref{D L L}, \eqref{D L Lb} and \eqref{D L A} are a subset of the structure equations of the null frame.
\end{remark}

\begin{lem}\label{lem propa pola}
Let $T$ be a symmetric 2-tensor.
\begin{enumerate}
\item[(i)] If $\Pol \left(\Ll_0 T \right)=0 $, then $\Pol(T)$ satisfies the following system
\begin{align}
\Ll_0 \Pol_{L_0}(T) & = 0,\label{eq Pol L T}
\\  \Ll_0 \Pol_{A}(T)& =  -2  \underline{\zeta}_A \Pol_{L_0}(T) - 2  \delta^{BC} g_0(\nabb_{L_0} e_A, e_B)\Pol_{C}(T),\label{eq Pol A T}
\\  \Ll_0 \Pol_{\Lb_0}(T) & =  - 4 \delta^{AB} \underline{\zeta}_A \Pol_{B}(T).\label{eq Pol Lb T}
\end{align}
\item[(ii)] If $(\Ll_0 T)_{L_0\Lb_0}=0$, then $T_{L_0\Lb_0}$ satisfies
\begin{align}
\Ll_0 T_{L_0\Lb_0} & =   4 \delta^{AB} \underline{\zeta}_A \Pol(T)_B \label{eq T Lb L}.
\end{align}
\item[(iii)] If $\Ll_0 T=0$, $\Pol(T)=0$ and $T_{L_0\Lb_0}=0$, then $\mathcal{E}(T,T)$ satisfies
\begin{align}
(-L_0+\Box_{g_0}u_0)\mathcal{E}(T,T)=0.\label{eq E(T,T)}
\end{align}
\end{enumerate}
\end{lem}

\begin{proof}
Using \eqref{D L L} we have
\begin{align*}
L_0 T_{L_0L_0} & = (\D_{L_0} T)_{L_0L_0} + 2 T_{L_0\alpha} \D_{L_0}L_0^\alpha = (\D_{L_0} T)_{L_0L_0} .
\end{align*}
Using in addition \eqref{D L A} we obtain
\begin{align*}
L_0 T_{L_0 A} & = (\D_{L_0} T)_{L_0A} + T_{L_0\alpha}\D_{L_0} e_A^\alpha + T_{A\alpha}\D_{L_0}L_0^\alpha
\\& = (\D_{L_0} T)_{L_0A}  + \underline{\zeta}_A T_{L_0L_0} + \delta^{BC} g_0(\nabb_{L_0} e_A, e_B) T_{L_0 C}.
\end{align*}
Similarly,  we obtain
\begin{align*}
L_0\left( \delta^{AB}T_{AB}\right) & = \delta^{AB} (\D_{L_0}T)_{AB} + \delta^{AB}T_{A\alpha}\D_{L_0}e_B^\alpha
\\& = \delta^{AB} (\D_{L_0}T)_{AB}  + 2 \delta^{AB} \underline{\zeta}_A T_{L_0B}  ,
\end{align*}
where we also use
\begin{equation*}
\delta^{AB}\delta^{CD} T_{AC}g_0(\nabb_{L_0}e_B,e_D) = 0,
\end{equation*}
which holds because $T$ is symmetric and $g_0(\nabb_{L_0}\cdot,\cdot)$ is antisymmetric. Recalling \eqref{def Ll_0}, this gives the following identities
\begin{align*}
\Ll_0 T_{L_0L_0} & = (\Ll_0T)_{L_0L_0} ,
\\ \Ll_0 T_{L_0A} & = (\Ll_0T)_{L_0A} -2  \underline{\zeta}_A T_{L_0L_0} -2  \delta^{BC} g_0(\nabb_{L_0} e_A, e_B)T_{L_0C},
\\  \Ll_0\left( \delta^{AB} T_{AB} \right)& =  \delta^{AB}(\Ll_0T)_{AB} - 4 \delta^{AB} \underline{\zeta}_A T_{BL_0}.
\end{align*}
Now, thanks to \eqref{Pol L}, \eqref{Pol A} and \eqref{Pol Lb}, the assumption $\Pol \left( \Ll_0T\right)=0$ can be used to rewrite the previous system as \eqref{eq Pol L T}, \eqref{eq Pol A T} and \eqref{eq Pol Lb T}. Similarly, we can show that
\begin{align*}
\Ll_0 T_{L_0\Lb_0} & = (\Ll_0T)_{L_0\Lb_0} - 4 \delta^{AB} \underline{\zeta}_A T_{BL_0},
\end{align*}
which becomes \eqref{eq T Lb L} after using the assumption $(\Ll_0T)_{L_0\Lb_0} =0$. We now prove \eqref{eq E(T,T)}. First, note that the assumptions $\Pol(T)=0$ and $T_{L_0\Lb_0}=0$ imply that
\begin{align}
\E\left(  T,T\right) & =  -\left( \left(T_{e_1e_1} \right)^2  +  \left(T_{e_1e_2} \right)^2 \right),\label{energie expression intermédiaire}
\end{align}
where we recall that $\E\left(  T,T\right)$ is defined in \eqref{def energie E}. We now use the additional assumption $\Ll_0 T=0$, the null structure equation \eqref{D L A} and the fact that $g_0(\nabb_{L_0}\cdot,\cdot)$ is antisymmetric to derive a transport equation for the quantities involved in \eqref{energie expression intermédiaire}:
\begin{align*}
- L_0 \left(T_{e_1e_1} \right)^2 + (\Box_{g_0}u_0) \left(T_{e_1e_1} \right)^2  & =  -   4  T_{e_1e_1}T_{e_1e_2} g_0(\nabb_{L_0}e_1,e_2) ,
\\ - L_0 \left(T_{e_1e_2} \right)^2 + (\Box_{g_0}u_0)\left(T_{e_1e_2} \right)^2 & =   2T_{e_1e_2}T_{e_1e_1}g_0(\nabb_{L_0}e_1,e_2) - 2T_{e_1e_2}T_{e_2e_2}g_0(\nabb_{L_0}e_1,e_2) .
\end{align*}
Thanks to $T_{e_2e_2}=-T_{e_1e_1}$ (consequence of $\Pol_{\Lb_0}(T)=0$) the RHS of the second equation becomes $4T_{e_1e_2}T_{e_1e_1}g_0(\nabb_{L_0}e_1,e_2) $ and we obtain \eqref{eq E(T,T)}
\end{proof}

When applied to $F^{(1)}$ and $\F$, this lemma gives the following.

\begin{coro}\label{coro pola F1}
The tensor $F^{(1)}$ satisfies 
\begin{align}
\Pol\left( F^{(1)}\right) & = 0,\label{pola F1}
\\ F^{(1)}_{L_0\Lb_0} & =0,\label{F1 L Lb}
\\  \E\left(  F^{(1)} ,  F^{(1)} \right)& = -4F_0^2.\label{energie}
\end{align}
The tensor $\F$ satisfies
\begin{align}
\Pol\left( \F\right) = 0.\label{pola F ==0}
\end{align}
\end{coro}

\begin{proof}
Since $F^{(1)}$ satisfies the transport equation \eqref{eq F1}, $\Pol\left( F^{(1)}\right)$ and $F^{(1)}_{L_0\Lb_0}$ satisfy the linear and homogeneous system \eqref{eq Pol L T}-\eqref{eq Pol A T}-\eqref{eq Pol Lb T}-\eqref{eq T Lb L}. Since these quantities vanish on $\Sigma_0$ (see Section \ref{section ID bg system}), we obtain \eqref{pola F1} and \eqref{F1 L Lb}. We can then use \eqref{eq F0} and \eqref{eq E(T,T)} and obtain
\begin{align}
\left( - L_0 + \Box_{g_0}u_0 \right) \left( \E\left(  F^{(1)} ,  F^{(1)} \right) + 4F_0^2 \right)=0.\label{eq quantite intermediaire}
\end{align}
Moreover, \eqref{ID energie F1} implies that $\E\left(  F^{(1)} ,  F^{(1)} \right) + 4F_0^2$ vanishes on $\Sigma_0$, equation \eqref{eq quantite intermediaire} then implies \eqref{energie}. The proof of \eqref{pola F ==0} is similar. Thanks to \eqref{pola F1} and \eqref{eq F}, we have $\Pol(\Ll_0 \F)=0$. Therefore $\Pol(\F)$ satisfies the linear and homogeneous system \eqref{eq Pol L T}-\eqref{eq Pol A T}-\eqref{eq Pol Lb T}, while vanishing initially (since $\F\restriction{\Sigma_0}=0$), we thus obtain \eqref{pola F ==0}.
\end{proof}

\begin{remark}
This corollary shows the consistency of the computations of Section \ref{section Ricci}, which were made under the assumptions that \eqref{pola F1} and \eqref{F1 L Lb} hold, since \eqref{pola F1} and \eqref{F1 L Lb} are equivalent to \eqref{assumptions on F1}.
\end{remark}

The transport equations for $F^{(2,i)}$, i.e \eqref{eq F21} and \eqref{eq F22}, are not as simple as \eqref{eq F1} and we can't hope to compute $\Pol\left(\Ll_0 F^{(2,i)}\right)$ easily. Moreover, the polarization conditions for $F^{(2,i)}$ are not simply $\Pol\left( F^{(2,i)}\right)=0$ but rather involve non-linear expressions depending on $F^{(1)}$ (see \eqref{def V 2 1} and \eqref{def V 2 2}). For these two reasons, the propagation of the second polarization, i.e the propagation of the polarization conditions satisfied by $F^{(2,i)}$, can't be proved with Lemma \ref{lem propa pola} and is postponed to Section \ref{section propa pola 23}.

\subsection{The Einstein tensor and the contracted Bianchi identities}\label{section expansion bianchi}

We compute the Ricci tensor of the solution $g$ to the background and reduced systems. This is based on Proposition \ref{prop expansion ricci tensor}, which provides the expression of the terms in the following expansion
\begin{align}
R_{\mu\nu} & = R^{(0)}_{\mu\nu}  +  \lambda R^{(1)}_{\mu\nu} +  \lambda^2 R^{(\geq 2)}_{\mu\nu}  .\label{expansion ricci formal}
\end{align}

\begin{remark}
In the sequel, we will encounter a lot of terms proportional to the polarization terms $V^{(2,j)}$ for which a precise expression is useless. Therefore, we introduce the following practical notation: we denote by $O\left( V^{(2)} \right)$ any combination of the tensors $V^{(2,j)}$ or its derivatives for $j=1,2$. This notation extends to tensors, in the sense that $O^\alpha\left( V^{(2)} \right)$ will denote any 1-tensor with coefficients of the form $O\left( V^{(2)} \right)$.
\end{remark}

\begin{prop}\label{prop final Ricci}
The three terms in the expansion \eqref{expansion ricci formal} of the Ricci tensor of the metric constructed in Theorems \ref{theo bg system} and \ref{theo reduced system} are given by
\begin{align*}
R^{(0)}_{\alpha\beta} & = -\half \sin\left(\frac{u_0}{\lambda}\right) \dr_{(\alpha} u_0 V^{(2,1)}_{\beta)} -2\cos\left(\frac{2u_0}{\lambda}\right) \dr_{(\alpha} u_0 V^{(2,2)}_{\beta)},
\\ R^{(1)}_{\alpha\beta} & = \half \cos\left( \frac{u_0}{\lambda} \right)\D_{(\alpha}V^{(2,1)}_{\beta)}  - \sin\left( \frac{2u_0}{\lambda} \right) \D_{(\alpha}V^{(2,2)}_{\beta)}  
\\&\quad +  O^\rho\pth{V^{(2)}} F^{(1)}_{\rho(\a} \dr_{\b)} u_0  + O\left( V^{(2)} \right) F^{(1)}_{\alpha\beta} 
\\&\quad  - \half \cos\left( \frac{u_0}{\lambda} \right)   \Pi_\geq \left( \h_{L_0L_0} \right) F^{(1)}_{\alpha\beta} + \half \dr_{(\a} u_0 (g_0)_{\rho \b)}  \dr_\theta(\Upsilon^\rho)^{(0)},
\\ R^{(\geq 2)}_{\a\b} & =\frac{1}{2\lambda}\cos\left(\frac{u_0}{\lambda} \right)\Pi_\geq \left(\h_{L_0L_0}  \right)F^{(1)}_{\alpha\beta} +\frac{ 1}{2} \pth{ \Upsilon^\rho \dr_\rho g_{\alpha\beta} + g_{\rho (\alpha}\dr_{\beta)}\Upsilon^\rho }^{(\geq 0)}.
\end{align*}
\end{prop}

We omit the proof of this proposition since it follows the same lines as the recipe presented in Section \ref{section towards the hierarchy}. Note that the term proportional to $\frac{1}{\lambda}\Pi_\geq \left(\h_{L_0L_0} \right)$ in $R^{(\geq 2)}$ cancels out with the same term present in $R^{(1)}$ so that it disappears from the final Einstein tensor computed below. Also note that 
\begin{align}
\frac{\la}{2} \dr_{(\a} u_0 (g_0)_{\rho \b)}  \dr_\theta(\Upsilon^\rho)^{(0)} + \frac{\la^2}{2} \pth{ \Upsilon^\rho \dr_\rho g_{\alpha\beta} + g_{\rho (\alpha}\dr_{\beta)}\Upsilon^\rho }^{(\geq 0)} & = \frac{\la^2}{2} \pth{ \Upsilon^\rho \dr_\rho g_{\alpha\beta} + g_{\rho (\alpha}\dr_{\beta)}\Upsilon^\rho }.\label{blablabla}
\end{align}
The first term on the LHS of \eqref{blablabla} corresponds to the $\Upsilon$ terms in $R^{(1)}_{\a\b}$, while the second one corresponds to the $\Upsilon$ terms in $R^{(\geq 2)}_{\a\b} $. Therefore the RHS of \eqref{blablabla} completely describes the $\Upsilon$ terms in the Ricci tensor.

\saut
As Proposition \ref{prop final Ricci} shows, the Ricci tensor of $g$ contains only gauge terms or the polarization conditions, i.e terms depending only on $\Upsilon$ and $V^{(2,i)}$. In order to show that $g$ is solution of the Einstein vacuum equations, it thus remains to show that they vanish. This is the content of Section \ref{section propa pola 23} and it is proved using the contracted Bianchi identities, i.e the fact that the Einstein tensor of any metric is divergence free, which will give us extra equations satisfied by $V^{(2,i)}$ and $\Upsilon$.  Even though this is standard when working with the (generalised or not) wave coordinates, the high-frequency character of $g$ (which manifests itself at this point by the presence of the polarization conditions tensors $V^{(2,i)}$) changes the situation. Schematically,  we will show that the divergence of $G$ satisfies
\begin{align*}
\dive_g G & = \mathrm{T}\left( \frac{u_0}{\lambda} \right) \Ll_0 V^{(2,i)} + \lambda^2 \tBox_g \Upsilon,
\end{align*}
with $\mathrm{T}$ a trigonometric function. The contracted Bianchi identities precisely state that for all $\lambda\in(0,\lambda_0]$ we have 
\begin{align}
\mathrm{T}\left( \frac{u_0}{\lambda} \right) \Ll_0 V^{(2,i)} + \lambda^2 \tBox_g \Upsilon = 0 ,\label{Bianchi schematic}
\end{align}
with $ \Ll_0 V^{(2,i)} $ independent of $\lambda$ and $\tBox_g \Upsilon$ depending on $\lambda$ through $g$ and $\Upsilon$. Our goal is to extract from \eqref{Bianchi schematic} the two equations $\Ll_0 V^{(2,i)} = 0$ and $ \tBox_g \Upsilon = 0$. However, because of the oscillating term in \eqref{Bianchi schematic} and the fact that $\tBox_g \Upsilon$ depends on $\lambda$ (since $\Upsilon$ depends on $\h$ or $\F$), we can't simply consider the expression in \eqref{Bianchi schematic} as a polynomial in $\lambda$ vanishing on the interval $(0,\lambda_0]$. Instead we multiply \eqref{Bianchi schematic} by $\mathrm{T}\left( \frac{u_0}{\lambda} \right)$ and use the fact that its weak limit in $L^2(K)$ when $\lambda$ tends to 0 is 0, where $K\vcentcolon=\{|x|\leq C_{\mathrm{supp}}R\}$ is the maximal support of all the background perturbations, and in particular the support of $V^{(2,i)}$ (this fact was justified in Section \ref{section BG}). In order to consider the limit $\la$ tends to 0 for $\tBox_g \Upsilon$ we need precise estimates obtained in Lemma \ref{lem dive G upsilon} and which follow from our bootstrap assumptions for $\h$ and $\F$.

\saut
We define the Einstein tensor of $g$ as
\begin{align}
G_{\alpha\beta} & = R_{\alpha\beta} - \half R g_{\alpha\beta},\label{def einstein tensor}
\end{align}
where $R=g^{\mu\nu}R_{\mu\nu}$ is the scalar curvature of $g$.  The terms $V^{(2,i)}$ and $\Upsilon$ will be treated differently so it is useful to decompose $G$ in the following way:
\begin{align}
G_{\alpha\beta} & = G\left(V \right)_{\alpha\beta} + G\left(\Upsilon\right)_{\alpha\beta} , \label{G decomposition}
\end{align}
with
\begin{align}
 G\left(\Upsilon\right)_{\alpha\beta}  & =\frac{\lambda^2}{2}  \left( g_{\rho(\alpha}\dr_{\beta)}\Upsilon^\rho - g_{\alpha\beta} \dr_\rho \Upsilon^\rho+  \Upsilon^\rho \dr_\rho g_{\alpha\beta} -\half g^{\mu\nu} \Upsilon^\rho \dr_\rho g_{\mu\nu} g_{\alpha\beta}\right). \label{G Upsilon}
\end{align}
and $G\left(V \right)$ depending only on $V^{(2,i)}$.  Thanks to Proposition \ref{prop final Ricci}, the tensor $G\left(V \right)$ admits a high-frequency expansion
\begin{align*}
G\left(V \right)_{\alpha\beta} & = G^{(0)}_{\alpha\beta}+ \lambda G^{(1)}_{\alpha\beta}  + \GO{\lambda^2} ,
\end{align*}
where 
\begin{align}
G^{(0)}_{\alpha\beta} & = -\half \sin\left(\frac{u_0}{\lambda}\right)\left(  \dr_{(\alpha} u_0 V^{(2,1)}_{\beta)} + V^{(2,1)}_{L_0} (g_0)_{\alpha\beta} \right) \label{G0}
 \\&\quad -2\cos\left(\frac{2u_0}{\lambda}\right) \left(  \dr_{(\alpha} u_0 V^{(2,2)}_{\beta)} + V^{(2,2)}_{L_0} (g_0)_{\alpha\beta}\right) ,\nonumber
\\ G^{(1)}_{\alpha\beta} & = \half \cos\left( \frac{u_0}{\lambda} \right)\left( \D_{(\alpha}V^{(2,1)}_{\beta)}  - \dive_{g_0}V^{(2,1)}(g_0)_{\alpha\beta}  \right) \label{G1}
\\&\quad - \sin\left( \frac{2u_0}{\lambda} \right) \left( \D_{(\alpha}V^{(2,2)}_{\beta)}  - \dive_{g_0}V^{(2,2)}(g_0)_{\alpha\beta}  \right)\nonumber
\\&\quad  +  O^\rho\pth{V^{(2)}} F^{(1)}_{\rho(\a} \dr_{\b)} u_0  +  O\left( V^{(2)} \right) F^{(1)}_{\alpha\beta} \nonumber.
\end{align}
The (contracted) Bianchi identities precisely reads
\begin{align}
\dive_g G & = 0.\label{Bianchi identities}
\end{align}
Thanks to \eqref{G decomposition}, \eqref{Bianchi identities} reads $\dive_g G(\Upsilon)+\dive_g G(V)=0$. These two terms are studied in Lemmas \ref{lem dive G upsilon} and \ref{lem B(V)} respectively. In the proof of Lemma \ref{lem dive G upsilon}, we only consider the $\dr\h$ and $\dr \F$ terms in $\Upsilon^\rho$ (recall \eqref{def Upsilon h}) since they satisfy the worst estimates. Moreover, since we only care about $L^2(K)$ norms with $K$ compact we can neglect the weights for $\h$.

\begin{lem}\label{lem dive G upsilon}
We have
\begin{align*}
\l \dive_g G\left(\Upsilon\right)_\alpha \r_{L^2(K)} & \lesssim \lambda.
\end{align*}
\end{lem}

\begin{proof}
The expression in coordinates
\begin{align*}
\dive_g G\left(\Upsilon\right)_\alpha & = g^{\mu\nu} \dr_\mu G\left(\Upsilon\right)_{\nu\alpha}  - g^{\mu\nu} \Gamma^\rho_{\mu(\alpha} G\left(\Upsilon\right)_{\rho\nu)},
\end{align*}
and a standard computation first give
\begin{align}
\dive_g G\left(\Upsilon\right)_\alpha & = \frac{\lambda^2}{2} \left( g_{\rho\alpha}\tBox_g\Upsilon^\rho + \tilde{\mathcal{B}}_\alpha\right), \label{divergence G Upsilon}
\end{align}
with
\begin{align*}
 \tilde{\mathcal{B}}_\alpha & = g^{\mu\nu} \dr_{(\alpha}\Upsilon^\rho \dr_\mu g_{\rho\nu)} - g^{\mu\nu}\dr_\rho \Upsilon^\rho \dr_\mu g_{\nu\alpha}   + g^{\mu\nu}\dr_\mu \left(  \Upsilon^\rho \dr_\rho g_{\alpha\nu} -\half g^{\rho\sigma} \Upsilon^\gamma \dr_\gamma g_{\rho\sigma} g_{\alpha\nu} \right)  -  g^{\mu\nu} \Gamma^\rho_{\mu(\alpha} G\left(\Upsilon\right)_{\rho\nu)}.
\end{align*}
We use the notation $L^2=L^2(K)$ for clarity. We start with $\tilde{\mathcal{B}}_\alpha$. Recall that $g$ and $\dr g$ are bounded in $L^\infty$ irrespective of $\lambda$. Moreover, $\dr \h$ and $\dr \F$ are bounded in $L^2$ irrespective of $\lambda$ thanks to \eqref{BA F k} and \eqref{BA h k}. Therefore, the only terms potentially losing one power of $\lambda$ in $\tilde{\mathcal{B}}_\alpha$ are $\dr \Upsilon$ and $\dr^2 g$. Indeed thanks to \eqref{g}, \eqref{BA F k} and \eqref{BA h k} these terms are bounded by $\frac{1}{\lambda}$ in $L^2$, which shows that
\begin{align}
\l \tilde{\mathcal{B}}_\alpha \r_{L^2} \lesssim \frac{1}{\lambda}. \label{estim B tilde}
\end{align}
Let us now look at $\tBox_g \Upsilon^\rho$. We only focus on the main terms in $\Upsilon^\rho$, that is $\dr \h + \sin\left( \frac{u_0}{\lambda}\right)\dr \F$. For $\dr \h$, we use the equation \eqref{eq h} and commute it with one derivative to obtain
\begin{align*}
\tBox_g \dr \h & = \dr \tBox_g  \h  + [\tBox_g,\dr]\h.
\end{align*}
We have 
\begin{align*}
\l [\tBox_g,\dr]\h \r_{L^2} & \lesssim \l \dr g \r_{L^\infty} \l \dr^2 \h \r_{L^2} \lesssim \frac{1}{\lambda},
\end{align*}
where we used \eqref{BA h k}. We now use the equation \eqref{eq h} (again, we only write down the terms depending on $\F$ with the highest derivative):
\begin{align*}
\l \dr \tBox_g  \h \r_{L^2} & \lesssim \frac{1}{\lambda} \l \tBox_g \F \r_{L^2} + \l \nabla \tBox_g \F \r_{L^2} + \l \dr_t \tBox_g \F \r_{L^2} \lesssim \frac{1}{\lambda} + \l \dr_t \tBox_g \F \r_{L^2},
\end{align*}
where we used \eqref{BA box F k}. This estimate can't directly be used to bound $\dr_t \tBox_g \F$, instead we use \eqref{dt F}:
\begin{align*}
\l \dr_t \tBox_g \F \r_{L^2} & \lesssim \l [\dr_t, \tBox_g ] \F \r_{L^2} + \l \tBox_g \dr_t \F \r_{L^2}
\\& \lesssim \l \dr^2 \F \r_{L^2} + \l \nabla \tBox_g  \F \r_{L^2} + \l \dr^2 \h \r_{L^2}
\\& \lesssim \frac{1}{\lambda}.
\end{align*}
We proved that 
\begin{align}
\l \tBox_g \dr \h \r_{L^2} & \lesssim \frac{1}{\lambda}\label{estim box d h}.
\end{align}
For the $\sin\left( \frac{u_0}{\lambda}\right)\dr \F$ term in $\Upsilon^\rho$ we use \eqref{box fct trigo bis 2} to deduce that we need to estimate
\begin{align*}
  \frac{1}{\lambda}\cos\left( \frac{u_0}{\lambda}\right) \Ll_0 \dr \F  + \sin\left( \frac{u_0}{\lambda}\right) \left( \tBox_g \dr \F -\frac{1}{\lambda^2} g^{-1}(\d u_0, \d u_0) \dr \F \right).
\end{align*}
First, using \eqref{g(du_0,du_0)}, \eqref{BA F k} and \eqref{BA box F k} we obtain 
\begin{align*}
\l \tBox_g \dr \F -\frac{1}{\lambda^2} g^{-1}(\d u_0, \d u_0) \dr \F \r_{L^2} \lesssim \frac{1}{\lambda}.
\end{align*}
Moreover, using \eqref{eq F} commuted with one derivative we obtain
\begin{align*}
\l \Ll_0 \dr \F \r_{L^2} & \lesssim \l [\Ll_0, \dr]\F \r_{L^2}  + \l \dr  \h \r_{L^2} \lesssim 1,
\end{align*}
where we used \eqref{BA F k} and \eqref{BA h k}. We proved that
\begin{align}
\l \tBox_g \left( \sin\left( \frac{u_0}{\lambda}\right)\dr \F \right) \r_{L^2} \lesssim \frac{1}{\lambda}.\label{estim box d F}
\end{align}
Together with \eqref{estim box d h}, \eqref{estim box d F} proves that $\l \tBox_g \Upsilon^\rho \r_{L^2}\lesssim \frac{1}{\lambda}$, which concludes the proof.
\end{proof}

\begin{lem}\label{lem B(V)}
We have
\begin{align}
\dive_g G\left(V \right)_\alpha & = \frac{1}{\lambda} \mathcal{B}(V)^{(-1)}_\alpha + \mathcal{B}(V)^{(0)}_\alpha + \GO{\lambda},\label{divergence G V}
\end{align}
with
\begin{align}
\mathcal{B}(V)^{(-1)}_\alpha & =0,\label{B-1}
\\ \mathcal{B}(V)^{(0)}_\alpha & = \half \sin\left(\frac{u_0}{\lambda}\right)\left( -\Ll_0 V^{(2,1)}_\alpha + g_0^{\mu\nu}V^{(2,1)}_{\nu} \D_{[\mu}(L_0)_{\alpha]}   \right)\label{B0}
\\&\quad + 2 \cos\left( \frac{2u_0}{\lambda}\right) \left( -\Ll_0 V^{(2,2)}_\alpha + g_0^{\mu\nu}V^{(2,2)}_{\nu} \D_{[\mu}(L_0)_{\alpha]}   \right).\nonumber
\end{align}
\end{lem}

\begin{proof}
The $\frac{1}{\lambda}$ terms in $\dive_g G\left(V \right)_\alpha$ comes from the differentiation of the oscillating parts of $G^{(0)}_{\nu\alpha}$, i.e
\begin{align*}
\mathcal{B}(V)^{(-1)}_\alpha & = -\half \cos\left(\frac{u_0}{\lambda}\right)\dr^\nu u_0\left(  \dr_{(\alpha} u_0 V^{(2,1)}_{\nu)} + V^{(2,1)}_{L_0} (g_0)_{\alpha\nu} \right) 
\\&\quad +4\sin\left(\frac{2u_0}{\lambda}\right) \dr^\nu u_0\left(  \dr_{(\alpha} u_0 V^{(2,2)}_{\nu)} + V^{(2,2)}_{L_0} (g_0)_{\alpha\nu}\right),
\end{align*}
where we recall \eqref{G0}. If $i=1,2$, we have
\begin{align*}
\dr^\nu u_0\left(  \dr_{(\alpha} u_0 V^{(2,i)}_{\nu)} + V^{(2,i)}_{L_0} (g_0)_{\alpha\nu} \right) & = -\dr_\nu u_0 V^{(2,i)}_{L_0} +  V^{(2,i)}_{L_0}\dr_\nu u_0 = 0,
\end{align*}
which concludes the proof of \eqref{B-1}. Now for \eqref{B0}, from the definition \eqref{divergence G V} we obtain
\begin{align}
\dive_g G\left(V \right)_\alpha & = g^{\mu\nu} \dr_\mu G(V)_{\nu\alpha} - g^{\mu\nu} \Gamma^\rho_{\mu\nu} G(V)_{\rho\alpha} - g^{\mu\nu} \Gamma^\rho_{\mu\alpha} G(V)_{\rho\nu},\label{def Bianchi bis}
\end{align}
where $\Gamma^\rho_{\mu\nu}$ denotes the Christoffel symbols associated to $g$ computed in \eqref{christoffel}. Note that $F^{(1)}_{L_0\alpha}=0$ and $\tr_{g_0}F^{(1)}=0$ imply $g_0^{\mu\nu}(\tilde{\Gamma}^{(0)})^\rho_{\mu\nu}=0$ (see \eqref{gammatilde0}). This gives
\begin{align}
\mathcal{B}(V)_\alpha^{(0)} & = g_0^{\mu\nu}\dr_{\mu} u_0 \dr_\theta G_{\nu\alpha}^{(1)} + g_0^{\mu\nu}\tilde{\D}_\mu G^{(0)}_{\nu\alpha} - g_0^{\mu\nu} (\tilde{\Gamma}^{(0)})^\rho_{\mu\alpha}G^{(0)}_{\rho\nu},\label{inter B0}
\end{align}
where in $\tilde{\D}_\mu G^{(0)}$ we don't differentiate the oscillating functions, as for the corresponding notation $\tilde{\dr}$. We start by looking at the term depending on $G^{(1)}$, whose expression is given by \eqref{G1}. Thanks to $F^{(1)}_{L_0\alpha}=0$ the terms proportional to $O\left( V^{(2)} \right)$ or $O^\rho\left( V^{(2)} \right)$ in \eqref{G1} don't contribute to $\mathcal{B}(V)_\alpha^{(0)}$. We are left with
\begin{align}
g_0^{\mu\nu}\dr_{\mu} u_0 \dr_\theta G_{\nu\alpha}^{(1)} & = \half \sin\left( \frac{u_0}{\lambda} \right)\left( L_0^\nu\D_{(\alpha}V^{(2,1)}_{\nu)}  - \dive_{g_0}V^{(2,1)}(L_0)_\alpha \right) \label{B0 d theta G1}
\\&\quad + 2\cos\left( \frac{2u_0}{\lambda} \right) \left(L_0^\nu \D_{(\alpha}V^{(2,2)}_{\nu)}  - \dive_{g_0}V^{(2,2)}(L_0)_\alpha  \right).\nonumber
\end{align}
Let us now look at the third term in \eqref{inter B0}.  Let $i=1,2$, we expand using \eqref{gammatilde0}:
\begin{align}
 (g_0)^{\mu\nu}(\Tilde{\Gamma}^{(0)})^\rho_{\mu \alpha} &\left(  \dr_{(\rho} u_0 V^{(2,i)}_{\nu)} + V^{(2,i)}_{L_0} (g_0)_{\rho\nu} \right) \label{calcul utile}
 \\&  = -\half \sin\left(\frac{u_0}{\lambda}\right)  g_0^{\mu\nu}g_0^{\rho\sigma}\left( \dr_{(\mu}u_0 F^{(1)}_{\sigma \alpha)} - \dr_\sigma u_0 F^{(1)}_{\mu\alpha}  \right)\dr_{(\rho}u_0 V^{(2,i)}_{\nu)}\nonumber
\\&\quad -\half \sin\left(\frac{u_0}{\lambda}\right) g_0^{\mu\nu}g_0^{\rho\sigma}\left( \dr_{(\mu}u_0 F^{(1)}_{\sigma \alpha)} - \dr_\sigma u_0 F^{(1)}_{\mu\alpha}  \right)  V^{(2,i)}_{L_0} (g_0)_{\rho\nu}\nonumber
\\& = -\half \sin\left(\frac{u_0}{\lambda}\right)  g_0^{\mu\nu}g_0^{\rho\sigma}\left( \dr_{(\mu}u_0 F^{(1)}_{\sigma \alpha)} - \dr_\sigma u_0 F^{(1)}_{\mu\alpha}  \right)\dr_{\nu}u_0 V^{(2,i)}_{\rho}\nonumber
\\& = 0.\nonumber
\end{align}
Therefore the third term in \eqref{inter B0} vanishes. Finally, let us look at the second term in this expression.  From $\D g_0=0$ we obtain
\begin{align}
g_0^{\mu\nu}\D_\mu G^{(0)}_{\nu\alpha} & =  -\half \sin\left(\frac{u_0}{\lambda}\right)\left( - g_0^{\mu\nu}V^{(2,1)}_{\nu} \D_\mu(L_0)_{\alpha}   - V^{(2,1)}_{\alpha}\dive_{g_0}L_0 \right.\label{divergence G0}
\\&\left. \qquad\qquad\qquad\qquad - (L_0)_{\alpha} \dive_{g_0}V^{(2,1)}  - \D_{L_0} V^{(2,1)}_{\alpha} + \dr_\alpha V^{(2,1)}_{L_0}  \right) \nonumber
\\&\quad -2\cos\left(\frac{2u_0}{\lambda}\right)\left( - g_0^{\mu\nu}V^{(2,2)}_{\nu} \D_\mu(L_0)_{\alpha}   - V^{(2,2)}_{\alpha}\dive_{g_0}L_0 \right.\nonumber
\\&\left. \qquad\qquad\qquad\qquad - (L_0)_{\alpha} \dive_{g_0}V^{(2,2)}  - \D_{L_0} V^{(2,1)}_{\alpha} + \dr_\alpha V^{(2,2)}_{L_0}  \right) .\nonumber
\end{align}
Adding this to \eqref{B0 d theta G1} concludes the proof.
\end{proof}

\subsection{Conclusion}\label{section propa pola 23}

Before proving the propagation of the second polarization, we first prove that $V_{\Lb_0}^{(2,i)}=0$ on $\Sigma_0$. Recall that the tangential components of $V^{(2,i)}$ already vanish on $\Sigma_0$ thanks to our choice of initial data for $F^{(2,i)}$, see \eqref{ID V2i=0}.  Since $V_{\Lb_0}^{(2,i)}$ only involves the projection of the metric on $\Sigma_0$ (see \eqref{Pol Lb}), $V_{\Lb_0}^{(2,i)}=0$ can't be ensured by well-chosen initial data, it has to be directly satisfied by the solution of the constraint equations. This means that we could have performed this computation in \cite{Touati2023a}. However, it seems more consistent with the splitting between \cite{Touati2023a} and the present article to study $V_{\Lb_0}^{(2,i)}$ here. Moreover, we will benefit from the reduced form of the Einstein tensor obtained above, and the fact that the constraint equations are solved on the initial hypersurface.

\begin{lem}\label{lem V=0 initialement}
We have 
\begin{align*}
V^{(2,i)}\restriction{\Sigma_0}=0.
\end{align*}
\end{lem}

\begin{proof}
Thanks to \eqref{initial pola tangential}, it remains to show that $V^{(2,i)}_{\Lb_0}=0$ on $\Sigma_0$. In fact, it is enough to show that 
\begin{align}
V^{(2,i)}_0 \restriction{\Sigma_0} =0 \label{V 0 = 0}.
\end{align}
Indeed, if \eqref{V 0 = 0} holds then $V^{(2,i)}_{\Lb_0}=-V^{(2,i)}_{L_0}$ (where we used that $L_0=|\nabla u_0|_{\bar{g}_0}(\dr_t + N_0)$ and $\Lb_0=|\nabla u_0|_{\bar{g}_0}(\dr_t - N_0)$ on $\Sigma_0$) and \eqref{initial pola tangential} then gives the result. Therefore, let us prove \eqref{V 0 = 0}. First, \eqref{G0} and \eqref{initial pola tangential} imply that on $\Sigma_0$
\begin{align*}
G(V)_{\alpha\beta} & = -\half \sin\left(\frac{u_0}{\lambda}\right) \dr_{(\alpha} u_0 V^{(2,1)}_{\beta)} - 2 \cos\left(\frac{2u_0}{\lambda}\right) \dr_{(\alpha} u_0 V^{(2,2)}_{\beta)} + \GO{\lambda}.
\end{align*} 
Moreover, $\Upsilon^\rho$ contains oscillating terms (see \eqref{def Upsilon h}) and also terms depending on $\lambda$ like $\h$ or $\F$. More precisely, if we only write the top derivatives of the non-background terms we have
\begin{align*}
| G(\Upsilon) | & \lesssim \lambda |\dr\F| + \lambda^2 \left( | \dr^2\F | + |\dr^2\h| \right).
\end{align*}
Estimates \eqref{BA F k} and \eqref{BA h k} then imply that $G(\Upsilon)$ is $\GO{\lambda}$ in $L^2(K)$. Conclusion, we have
\begin{align}
G_{\alpha\beta} & = -\half \sin\left(\frac{u_0}{\lambda}\right) \dr_{(\alpha} u_0 V^{(2,1)}_{\beta)} - 2 \cos\left(\frac{2u_0}{\lambda}\right) \dr_{(\alpha} u_0 V^{(2,2)}_{\beta)} + \GO{\lambda}, \label{G reduced}
\end{align}
where the $\GO{\lambda}$ has to be understood in $L^2(K)$.  Now, Corollary \ref{coro ID} ensures that the constraint equations are satisfied on $\Sigma_0$, which in particular gives $G_{TT}=0$ on $\Sigma_0$, where $T$ is the unit normal to $\Sigma_0$ for $g$.  Thanks to Lemma \ref{lem unit normal}, we have $T=\dr_t + \GO{\lambda^2}$ so the previous identity implies that $G_{00}=\GO{\lambda^2}$ on $\Sigma_0$. Thanks to \eqref{G reduced} and \eqref{dt u_0} this rewrites as
\begin{align}
- \sin\left(\frac{u_0}{\lambda}\right) |\nabla u_0|_{\bar{g}_0} V^{(2,1)}_{0} - 4 \cos\left(\frac{2u_0}{\lambda}\right) |\nabla u_0|_{\bar{g}_0} V^{(2,2)}_{0} = \GO{\lambda}.\label{lem 8 10 inter}
\end{align}
We multiply this identity by $\sin\left(\frac{u_0}{\lambda}\right)$ and take the weak limit in $L^2(K)$ when $\lambda$ tends to 0. Moreover, the weak limit of $\sin^2\left(\frac{u_0}{\lambda}\right)$ and $\sin\left(\frac{u_0}{\lambda}\right)\cos\left(\frac{2u_0}{\lambda}\right)$ are respectively $\half$ and 0, and $V^{(2,1)}_0$ does not depend on $\lambda$ so we obtain $V^{(2,1)}_0=0$. Then \eqref{lem 8 10 inter} becomes 
\begin{align}
 \cos\left(\frac{2u_0}{\lambda}\right) |\nabla u_0|_{\bar{g}_0} V^{(2,2)}_{0} = \GO{\lambda}.\label{lem 8 10 inter bis}
\end{align}
Multiplying this identity by $\cos\left(\frac{2u_0}{\lambda}\right)$ and taking the weak limit in $L^2(K)$ when $\lambda$ tends to 0 gives similarly that $V^{(2,2)}_0=0$. This proves \eqref{V 0 = 0} and concludes the proof of the lemma.
\end{proof}

We can now prove the propagation of the second polarization.

\begin{prop}\label{prop V2i}
The following holds on the whole spacetime
\begin{align}
V^{(2,i)}=0, \label{V2i=0}
\end{align}
for $i=1,2$.
\end{prop}

\begin{proof}
We look at the quantity $\sin\left(\frac{u_0}{\lambda}\right) \dive_g G_\alpha$, which thanks to Lemma \ref{lem B(V)} rewrites
\begin{align}
\sin\left(\frac{u_0}{\lambda}\right) \dive_g G_\alpha  & = \half \sin^2\left(\frac{u_0}{\lambda}\right)\left( -\Ll_0 V^{(2,1)}_\alpha + g_0^{\mu\nu}V^{(2,1)}_{\nu} \D_{[\mu}(L_0)_{\alpha]}   \right)\label{sin B}
\\&\quad + 2 \sin\left(\frac{u_0}{\lambda}\right)\cos\left( \frac{2u_0}{\lambda}\right) \left( -\Ll_0 V^{(2,2)}_\alpha + g_0^{\mu\nu}V^{(2,2)}_{\nu} \D_{[\mu}(L_0)_{\alpha]}   \right)\nonumber
\\&\quad + \GO{\lambda}, \nonumber
\end{align}
where we use the Bianchi identities \eqref{Bianchi identities}, Lemma \ref{lem dive G upsilon} and where the $\GO{\lambda}$ has to be understood in $L^2(K)$. The quantity 
\begin{align*}
-\Ll_0 V^{(2,1)}_\alpha + g_0^{\mu\nu}V^{(2,1)}_{\nu} \D_{[\mu}(L_0)_{\alpha]}  ,
\end{align*} 
only involves the background perturbations so it doesn't depend on $\lambda$. Therefore the weak limit when $\lambda$ tends to 0 of the first two lines in \eqref{sin B} is equal to 
\begin{align}
\frac{1}{4}\left( -\Ll_0 V^{(2,1)}_\alpha + g_0^{\mu\nu}V^{(2,1)}_{\nu} \D_{[\mu}(L_0)_{\alpha]}  \right).\label{weak limit 1}
\end{align}
Thanks to \eqref{Bianchi identities}, the weak limit of \eqref{sin B} is zero and we obtain the following equation for $V^{(2,1)}$:
\begin{align}
\Ll_0 V^{(2,1)}_\alpha = g_0^{\mu\nu}V^{(2,1)}_{\nu} \D_{[\mu}(L_0)_{\alpha]} \label{eq V21}.
\end{align}
If we now multiply $\dive_g G_\alpha$ by $\cos\left(\frac{2u_0}{\lambda}\right)$, Lemma \ref{lem B(V)} and \eqref{eq V21} imply
\begin{align*}
\cos\left(\frac{2u_0}{\lambda}\right)\dive_g G_\alpha & = 2 \cos^2\left( \frac{2u_0}{\lambda}\right) \left( -\Ll_0 V^{(2,2)}_\alpha + g_0^{\mu\nu}V^{(2,2)}_{\nu} \D_{[\mu}(L_0)_{\alpha]}   \right) + \GO{\lambda}.
\end{align*} 
We again take the weak limit in $L^2$ when $\lambda$ tends to 0 and thanks to \eqref{Bianchi identities} we obtain
\begin{align}
\Ll_0 V^{(2,2)}_\alpha = g_0^{\mu\nu}V^{(2,2)}_{\nu} \D_{[\mu}(L_0)_{\alpha]} \label{eq V22}.
\end{align}
Equations \eqref{eq V21} and \eqref{eq V22} are transport equations for the quantities $V^{(2,i)}$, $i=1,2$. Thanks to Lemma \ref{lem V=0 initialement} we know that 
\begin{align}
V^{(2,i)}\restriction{\Sigma_0}=0,\label{blabla}
\end{align}
for $i=1,2$. Together with the energy estimate for the transport operator $L_0$ (see Lemma \ref{lem L0 h =f}) and Gronwall's inequality,  \eqref{eq V22} and \eqref{blabla} lead to \eqref{V2i=0}.

\end{proof}

It remains to prove that $\Upsilon^\rho =0$. The Einstein tensor of $g$ now reduces to $G\left(\Upsilon\right)$, i.e
\begin{align}
G_{\alpha\beta} & = \frac{\lambda^2}{2}  \left( g_{\rho(\alpha}\dr_{\beta)}\Upsilon^\rho - g_{\alpha\beta} \dr_\rho \Upsilon^\rho+  \Upsilon^\rho \dr_\rho g_{\alpha\beta} -\half g^{\mu\nu} \Upsilon^\rho \dr_\rho g_{\mu\nu} g_{\alpha\beta}\right).\label{Einstein last expression}
\end{align}
At this point, it is standard to deduce from the Bianchi identities \eqref{Bianchi identities} the fact that $\Upsilon^\rho$ vanishes if it vanishes initially, we redo the proof for completeness.

\begin{prop}\label{prop Upsilon}
The following holds on the whole spacetime
\begin{align*}
\Upsilon^\rho = 0.
\end{align*}
\end{prop}

\begin{proof}
The main ingredient is already contained in \eqref{divergence G Upsilon}. Since $V^{(2,i)}=0$, we have $G_{\alpha\beta} = G\left(\Upsilon\right)_{\alpha\beta}$ and \eqref{divergence G Upsilon} together with \eqref{Bianchi identities} then implies 
\begin{align}
\tBox_g \Upsilon^\rho + A^{\mu\rho}_\nu \dr_\mu \Upsilon^\nu + B^\rho_\alpha \Upsilon^\alpha = 0 \label{wave Upsilon}.
\end{align}
Therefore, $\Upsilon$ solves \eqref{wave Upsilon}, i.e a system of linear wave equations on the spacetime.  If $\Upsilon^\rho$ and $T \Upsilon^\rho$ vanish on $\Sigma_0$, then $\Upsilon^\rho$ vanishes on the whole spacetime (recall that $T$ is the unit normal to $\Sigma_0$ for $g$, see Lemma \ref{lem unit normal}). From Corollary \ref{coro ID} and our choice of initial data for $g$, we deduce what we need.
\begin{itemize}
\item From Corollary \ref{coro ID} we have $\Upsilon^\rho\restriction{\Sigma_0}=0$. Thanks to \eqref{Einstein last expression} this also implies that on $\Sigma_0$ we have
\begin{align}
G_{\alpha\beta} & = \frac{\lambda^2}{2}  \left( g_{\rho(\alpha}\dr_{\beta)}\Upsilon^\rho - g_{\alpha\beta} \dr_t \Upsilon^0 \right). \label{Einstein last expression on sigma 0}
\end{align}
\item Corollary \ref{coro ID} also states that the constraint equations are solved on $\Sigma_0$, which reads $G_{TT}=0$ and $G_{Ti}=0$.  Since $T=(1+Z^t)\dr_t + \bar{Z}^i\dr_i$ and $g(T,\dr_i)=0$ the previous components of the Einstein tensor rewrite
\begin{align*}
G_{TT} & = 2T^\alpha g_{\rho\alpha} T\Upsilon^\rho + \dr_t \Upsilon^0 = - \frac{T\Upsilon^0}{1+Z^t}  ,
\end{align*}
and
\begin{align*}
G_{Ti} &  = g_{0i} T\Upsilon^0 + \bar{g}_{ij} T\Upsilon^j,
\end{align*}
where we used \eqref{Einstein last expression on sigma 0} and $\Upsilon^\rho\restriction{\Sigma_0}=0$ to cancel any spatial derivative of $\Upsilon^\rho$. Since $Z^t=\GO{\lambda^4}$ (see \eqref{def Zt}) the constraint equations indeed imply that $T\Upsilon^\rho\restriction{\Sigma_0}=0$.
\end{itemize}
\end{proof}

Propositions \ref{prop V2i} and \ref{prop Upsilon} concludes the proof of Theorem \ref{theo main}, by showing that the metric $g$ solves the Einstein vacuum equations.

\appendix

\section{Computations of Section \ref{section high-frequency ansatz}}\label{appendix calcul}

\subsection{Proof of Proposition \ref{prop expansion wave part}}

In this section, we derive the expressions of the coefficients in the expansion of the wave part of the Ricci tensor:
\begin{align}
\tBox_g g_{\alpha\beta} & = W^{(0)}_{\alpha\beta} + \lambda W^{(1)}_{\alpha\beta} + \lambda^2 W^{(\geq 2)}_{\alpha\beta}  .
\end{align}
The main ingredient is the computation \eqref{box fct trigo bis 2}, applied to each oscillating term in \eqref{g}. We also need to expand the eikonal term $g_0^{-1}(\d u_0, \d u_0)$, $ g^{\mu\nu}\dr_\mu u_0 \dr_\nu f$ and $\tBox_g f$ (for various $f$), all appearing in \eqref{box fct trigo bis 2}. For that, we use the expansion of the inverse given in \eqref{inverse g}-\eqref{inverse g2} and \eqref{assumptions on F1}. This gives
\begin{align}
&g^{-1}(\d u_0, \d u_0) \nonumber
\\& = g_0^{-1}(\d u_0,\d u_0) - \lambda F^{(1)}_{L_0L_0} \cos\left(\frac{u_0}{\lambda}\right)\nonumber
\\&\quad - \lambda^2 \left( \F_{L_0L_0} \sin\left(\frac{u_0}{\lambda}\right) +   F^{(2,1)}_{L_0L_0} \sin\left(\frac{u_0}{\lambda}\right) +  F^{(2,2)}_{L_0L_0} \cos\left(\frac{2u_0}{\lambda}\right) + \h_{L_0L_0} \right)   + \la^3(g^{-1})^{(\geq 3)}(\dr u_0)^2 \nonumber
\\ & = - \lambda^2 \left( \F_{L_0L_0} \sin\left(\frac{u_0}{\lambda}\right) +   F^{(2,1)}_{L_0L_0} \sin\left(\frac{u_0}{\lambda}\right) +  F^{(2,2)}_{L_0L_0} \cos\left(\frac{2u_0}{\lambda}\right) + \h_{L_0L_0} \right)   + \la^3(g^{-1})^{(\geq 3)}(\dr u_0)^2 ,\label{g(du_0,du_0)}
\end{align}
where we also used \eqref{eq u0}. The absence of a $\lambda^0$ term in \eqref{g(du_0,du_0)} explains why there is no $\lambda^{-1}$ term in the expansion of $\tBox_g g_{\alpha\beta}$ given by \eqref{expansion wave part}. Similarly we have 
\begin{align}
g^{\mu\nu}\dr_\mu u_0 \dr_\nu f & = -L_0 f + \la^2(g^{-1})^{(\geq 2)} \dr u_0 \dr f, \label{g du0 df}
\end{align}
and
\begin{align}
\tBox_g f & = \tBox_{g_0} f + \la (g^{-1})^{(\geq 1)} \dr^2 f, \label{box g f}
\end{align}
for $f$ a scalar function. Note that \eqref{box g f} also gives the expansion of the wave operator applied to the background term in \eqref{g}:
\begin{align*}
\tBox_g (g_0)_{\a\b} & = \tBox_{g_0} (g_0)_{\a\b} + \la (g^{-1})^{(\geq 1)} \dr^2 g_0.
\end{align*}
We now apply \eqref{box fct trigo bis 2} to each oscillating term in \eqref{g}, using systematically \eqref{g(du_0,du_0)}, \eqref{g du0 df} and \eqref{box g f}. For $\la g^{(1)}$ this gives
\begin{align*}
&\tBox_g \pth{ \la \cos\pth{\frac{u_0}{\la}} F^{(1)}_{\a\b} } 
\\& = - \frac{1}{\la} \cos\pth{\frac{u_0}{\la}}g^{-1}(\d u_0,\d u_0) F^{(1)}_{\a\b} - \sin\pth{\frac{u_0}{\la}}\pth{ 2 g^{\mu\nu}\dr_\mu u_0 \dr_\nu F^{(1)}_{\a\b} + (\tBox_g u_0) F^{(1)}_{\a\b} }
\\&\quad  + \la \cos\pth{\frac{u_0}{\la}} \tBox_g F^{(1)}_{\a\b}
\\& = - \sin\pth{\frac{u_0}{\la}}\pth{ -2L_0 F^{(1)}_{\a\b} + (\tBox_{g_0} u_0) F^{(1)}_{\a\b} } 
\\&\quad + \la \cos\pth{\frac{u_0}{\la}} \left( \F_{L_0L_0} \sin\left(\frac{u_0}{\lambda}\right) +   F^{(2,1)}_{L_0L_0} \sin\left(\frac{u_0}{\lambda}\right) +  F^{(2,2)}_{L_0L_0} \cos\left(\frac{2u_0}{\lambda}\right) + \h_{L_0L_0} \right)    F^{(1)}_{\a\b} 
\\&\quad + \la  \pth{ (g^{-1})^{(\geq 1)} F^{(1)} \dr^2 u_0 + g^{-1}\dr^2 F^{(1)} } + \la^2 \pth{ (g^{-1})^{(\geq 2)}\dr u_0 \dr F^{(1)} + (g^{-1})^{(\geq 3)}(\dr u_0)^2  F^{(1)} }.
\end{align*}
For $\la^2g^{(2)}$ this gives
\begin{align*}
&\tBox_g \pth{ \la^2 \sin\pth{\frac{u_0}{\la}} \F_{\a\b} } 
\\& = -\sin\left( \frac{u_0}{\lambda} \right) g^{-1}(\d u_0, \d u_0) \F_{\a\b} + \la\cos\left( \frac{u_0}{\lambda} \right) \left( 2g^{\mu\nu}\dr_\mu u_0 \dr_\nu \F_{\a\b} + (\tBox_g u_0)\F_{\a\b}  \right)
\\&\quad + \la^2\sin\left( \frac{u_0}{\lambda} \right)\tBox_g \F_{\a\b}
\\& =   \la\cos\left( \frac{u_0}{\lambda} \right) \left( -2L_0 \F_{\a\b} + (\tBox_{g_0} u_0)\F_{\a\b}  \right)  + \la^2\sin\left( \frac{u_0}{\lambda} \right)\tBox_g \F_{\a\b}
\\&\quad + \la^2 \pth{ (g^{-1})^{(\geq 2)}(\dr u_0)^2 \F + (g^{-1})^{(\geq 1)} \dr^2 u_0\F   } +   \la^3(g^{-1})^{(\geq 2)} \dr u_0 \dr \F ,
\end{align*}
and
\begin{align*}
&\tBox_g \pth{ \la^2 \sin\pth{\frac{u_0}{\la}} F^{(2,1)}_{\a\b} } 
\\& =   \la\cos\left( \frac{u_0}{\lambda} \right) \left( -2L_0 F^{(2,1)}_{\a\b} + (\tBox_{g_0} u_0)F^{(2,1)}_{\a\b}  \right) 
\\&\quad + \la^2 \pth{ (g^{-1})^{(\geq 2)}(\dr u_0)^2 F^{(2,1)} + (g^{-1})^{(\geq 1)} \dr^2 u_0 F^{(2,1)} + g^{-1}\dr^2 F^{(2,1)}  } +   \la^3(g^{-1})^{(\geq 2)} \dr u_0 \dr F^{(2,1)}  ,
\end{align*}
and
\begin{align*}
&\tBox_g \pth{ \la^2 \cos\pth{\frac{2u_0}{\la}} F^{(2,2)}_{\a\b} } 
\\& =  -2 \la\sin\left( \frac{2u_0}{\lambda} \right) \left( -2L_0 F^{(2,2)}_{\a\b} + (\tBox_{g_0} u_0)F^{(2,2)}_{\a\b}  \right) 
\\&\quad + \la^2 \pth{ (g^{-1})^{(\geq 2)}(\dr u_0)^2 F^{(2,2)} + (g^{-1})^{(\geq 1)} \dr^2 u_0F^{(2,2)} + g^{-1}\dr^2 F^{(2,2)} } +   \la^3(g^{-1})^{(\geq 2)} \dr u_0 \dr F^{(2,2)} .
\end{align*}
Collecting terms for each power of $\la$ in the previous expressions, we obtain the various expressions of Proposition \ref{prop expansion wave part}.

\subsection{Proof of Proposition \ref{prop expansion quad part}}

The quadratic non-linearity $P(g)(\dr g,\dr g)$ is of the form $(g^{-1}\dr g)^2$, see \eqref{quadratic non-linearity} for an exact expression. Therefore, we have schematically
\begin{align*}
P^{(0)} & = \underbrace{(g_0^{-1} \dr g_0)^2}_{\vcentcolon = A }  + \underbrace{g_0^{-2}\dr g_0 \dr u_0 \dr_\theta g^{(1)}}_{\vcentcolon = B} + \underbrace{(g_0^{-1}\dr u_0 \dr_\theta g^{(1)})^2}_{\vcentcolon = C},
\end{align*}
since $(\dr g)^{(0)}= \dr g_0 + \dr u_0 \dr_\theta g^{(1)}$. Note that the purely background term $A_{\a\b}$ simply corresponds to $P_{\a\b}(g_0)(\dr g_0,\dr g_0)$. For $B_{\a\b}$, we have from \eqref{quadratic non-linearity}
\begin{align*}
B_{\a\b} & =  g_0^{\mu\rho}g_0^{\nu\sigma}\dr_{(\alpha}(g_0)_{\rho\sigma}\dr_\mu u_0 \dr_\theta g^{(1)}_{\beta)\nu}  - \frac{1}{2}g_0^{\mu\rho}g_0^{\nu\sigma}\dr_\alpha (g_0)_{\rho\sigma} \dr_\beta u_0 \dr_\theta g^{(1)}_{\mu\nu} 
\\&\quad - g_0^{\mu\rho}g_0^{\nu\sigma}\dr_\rho (g_0)_{\alpha\nu}\dr_\sigma u_0 \dr_\theta g^{(1)}_{\beta\mu} +  g_0^{\mu\rho}g_0^{\nu\sigma}\dr_\rho (g_0)_{\sigma\alpha} \dr_\mu u_0 \dr_\theta g^{(1)}_{\nu\beta}
\\&\quad + g_0^{\mu\rho}g_0^{\nu\sigma}\dr_{(\alpha}u_0 \dr_\theta g^{(1)}_{\rho\sigma}\dr_\mu (g_0)_{\beta)\nu}  - \frac{1}{2}g_0^{\mu\rho}g_0^{\nu\sigma}\dr_\alpha u_0 \dr_\theta g^{(1)}_{\rho\sigma} \dr_\beta (g_0)_{\mu\nu} 
\\&\quad - g_0^{\mu\rho}g_0^{\nu\sigma}\dr_\rho u_0 \dr_\theta g^{(1)}_{\alpha\nu}\dr_\sigma (g_0)_{\beta\mu} +  g_0^{\mu\rho}g_0^{\nu\sigma}\dr_\rho u_0 \dr_\theta g^{(1)}_{\sigma\alpha} \dr_\mu (g_0)_{\nu\beta}
\\& = - L_0^\rho g_0^{\mu\sigma}\pth{ \dr_{(\alpha}(g_0)_{\rho\sigma} +  \dr_\rho (g_0)_{\sigma(\alpha} - \dr_\si (g_0)_{(\alpha\rho}    }  \dr_\theta g^{(1)}_{\b)\mu}
\\&\quad + g_0^{\mu\rho}g_0^{\nu\sigma}\dr_{(\alpha}u_0 \pth{ \dr_\rho (g_0)_{\beta)\sigma} - \frac{1}{2}\dr_{\b)} (g_0)_{\rho\sigma}   }\dr_\theta g^{(1)}_{\mu\nu}
\\& = - 2 L_0^\rho \Gamma(g_0)^\mu_{(\alpha\rho}      \dr_\theta g^{(1)}_{\b)\mu} + g_0^{\mu\rho}g_0^{\nu\sigma}\dr_{(\alpha}u_0 \pth{ \dr_\rho (g_0)_{\beta)\sigma} - \frac{1}{2}\dr_{\b)} (g_0)_{\rho\sigma}   }\dr_\theta g^{(1)}_{\mu\nu}.
\end{align*}
If we now use \eqref{expression g1} we obtain
\begin{align*}
B_{\a\b} & = -\sin\pth{\frac{u_0}{\la}} \pth{ - 2 L_0^\rho \Gamma(g_0)^\mu_{(\alpha\rho}   F^{(1)}_{\b)\mu} + g_0^{\mu\rho}g_0^{\nu\sigma}\dr_{(\alpha}u_0 \pth{ \dr_\rho (g_0)_{\beta)\sigma} - \frac{1}{2}\dr_{\b)} (g_0)_{\rho\sigma}   }F^{(1)}_{\mu\nu} }.
\end{align*}
For $C_{\a\b}$, we have from \eqref{quadratic non-linearity}
\begin{align*}
C_{\a\b} & =  g_0^{\mu\rho}g_0^{\nu\sigma}\dr_{(\alpha}u_0 \dr_\theta g^{(1)}_{\rho\sigma}\dr_\mu u_0 \dr_\theta g^{(1)}_{\beta)\nu}  - \frac{1}{2}g_0^{\mu\rho}g_0^{\nu\sigma}\dr_\alpha u_0 \dr_\theta g^{(1)}_{\rho\sigma} \dr_\beta u_0 \dr_\theta g^{(1)}_{\mu\nu} 
\\&\quad - g_0^{\mu\rho}g_0^{\nu\sigma}\dr_\rho u_0 \dr_\theta g^{(1)}_{\alpha\nu}\dr_\sigma u_0 \dr_\theta g^{(1)}_{\beta\mu} +  g_0^{\mu\rho}g_0^{\nu\sigma}\dr_\rho u_0 \dr_\theta g^{(1)}_{\sigma\alpha} \dr_\mu u_0 \dr_\theta g^{(1)}_{\nu\beta}
\\& =   - \frac{1}{2} \left| \dr_\theta g^{(1)} \right|^2_{g_0} \dr_\alpha u_0  \dr_\beta u_0
\\& =  - \frac{1}{2} \sin^2\pth{\frac{u_0}{\la}} \left| F^{(1)} \right|^2_{g_0} \dr_\alpha u_0  \dr_\beta u_0,
\end{align*}
where we used \eqref{eq u0}  and \eqref{assumptions on F1}. This proves \eqref{P0}. Similarly, we have schematically
\begin{align*}
P^{(1)} & = (g^{-1} \dr g)^{(0)} (g^{-1} \dr g)^{(1)} 
\\& =  \underbrace{g_0^{-1}(g^{-1})^{(1)} (\dr g_0)^2}_{\vcentcolon = D}   + \underbrace{g_0^{-1}(g^{-1})^{(1)} \dr g_0 \dr u_0 \dr_\theta g^{(1)}}_{\vcentcolon = E}     + \underbrace{g_0^{-1}(g^{-1})^{(1)}  (\dr u_0 \dr_\theta g^{(1)})^2}_{\vcentcolon = F} 
\\&\quad + \underbrace{g_0^{-2}\dr g_0 \tilde{\dr} g^{(1)}}_{\vcentcolon = G}    + \underbrace{g_0^{-2}\dr g_0 \dr u_0 \dr_\theta g^{(2)}}_{\vcentcolon = H}    + \underbrace{g_0^{-2}\dr u_0 \dr_\theta g^{(1)} \tilde{\dr} g^{(1)}}_{\vcentcolon = I}   + \underbrace{g_0^{-2}\dr u_0 \dr_\theta g^{(1)} \dr u_0 \dr_\theta g^{(2)}}_{\vcentcolon = J}   ,
\end{align*}
where we used $(\dr g)^{(1)} = \tilde{\dr} g^{(1)} + \dr u_0 \dr_\theta g^{(2)}$. For $D$, $E$, $G$ and $I$, we simply look for a schematic expression using \eqref{expression g1} and \eqref{inverse g1}:
\begin{align*}
D & = \cos\pth{\frac{u_0}{\la}} g_0^{-3}F^{(1)} (\dr g_0)^2,
\\ E & = \sin\pth{\frac{2u_0}{\la}}g_0^{-3}(F^{(1)})^2 \dr g_0 \dr u_0,
\\ G & = \cos\pth{\frac{u_0}{\la}} g_0^{-2}\dr g_0 \dr F^{(1)},
\\ I & = \sin\pth{\frac{2u_0}{\la}} g_0^{-2}\dr u_0 F^{(1)} \dr F^{(1)}.
\end{align*}
For $F_{\a\b}$, we have from \eqref{quadratic non-linearity}, \eqref{eq u0}, \eqref{inverse g1} and \eqref{assumptions on F1}
\begin{align*}
F_{\a\b} & =  g_0^{\mu\rho}(g^{(1)})^{\nu\sigma} \dr_{(\alpha}u_0 \dr_\theta g^{(1)}_{\rho\sigma}\dr_\mu u_0 \dr_\theta g^{(1)}_{\beta)\nu}    - \frac{1}{2}\pth{ g_0^{\mu\rho}(g^{(1)})^{\nu\sigma} + (g^{(1)})^{\mu\rho}g_0^{\nu\sigma} }\dr_\alpha u_0 \dr_\theta g^{(1)}_{\rho\sigma} \dr_\beta u_0 \dr_\theta g^{(1)}_{\mu\nu}  
\\& = \cos\pth{\frac{u_0}{\la}}\sin^2\pth{\frac{u_0}{\la}}  \dr_{(\alpha}u_0  \bigg( -g_0^{\mu\rho}(F^{(1)})^{\nu\sigma}  F^{(1)}_{\rho\sigma}\dr_\mu u_0 F^{(1)}_{\beta)\nu} 
\\&\left. \hspace{5cm} + \frac{1}{4}\pth{ g_0^{\mu\rho}(F^{(1)})^{\nu\sigma} + (F^{(1)})^{\mu\rho}g_0^{\nu\sigma} } F^{(1)}_{\rho\sigma} \dr_{\beta)} u_0 F^{(1)}_{\mu\nu}  \right).
\end{align*}
We now look at the terms depending on $g^{(2)}$, i.e $H$ and $J$. The term $H$ is actually equivalent to the term $B$ above, therefore we have
\begin{align*}
H_{\a\b} & =  - 2 L_0^\rho \Gamma(g_0)^\mu_{(\alpha\rho}      \dr_\theta g^{(2)}_{\b)\mu} + g_0^{\mu\rho}g_0^{\nu\sigma}\dr_{(\alpha}u_0 \pth{ \dr_\rho (g_0)_{\beta)\sigma} - \frac{1}{2}\dr_{\b)} (g_0)_{\rho\sigma}   }\dr_\theta g^{(2)}_{\mu\nu}
\\& = \cos\pth{\frac{u_0}{\la}} \bigg( - 2 L_0^\rho \Gamma(g_0)^\mu_{(\alpha\rho}    (\F + F^{(2,1)})_{\b)\mu} 
\\&\hspace{4cm}+ g_0^{\mu\rho}g_0^{\nu\sigma}\dr_{(\alpha}u_0 \pth{ \dr_\rho (g_0)_{\beta)\sigma} - \frac{1}{2}\dr_{\b)} (g_0)_{\rho\sigma}   } (\F + F^{(2,1)})_{\mu\nu} \bigg)
\\&\quad -2\sin\pth{\frac{2u_0}{\la}} \pth{ - 2 L_0^\rho \Gamma(g_0)^\mu_{(\alpha\rho}   F^{(2,2)}_{\b)\mu} + g_0^{\mu\rho}g_0^{\nu\sigma}\dr_{(\alpha}u_0 \pth{ \dr_\rho (g_0)_{\beta)\sigma} - \frac{1}{2}\dr_{\b)} (g_0)_{\rho\sigma}   } F^{(2,2)}_{\mu\nu} },
\end{align*}
where we used \eqref{expression g2}. Finally for $J$, we have from \eqref{quadratic non-linearity}, \eqref{eq u0} and \eqref{assumptions on F1}
\begin{align*}
J_{\a\b} & = g_0^{\mu\rho}g_0^{\nu\sigma}  \dr_{(\alpha}u_0 \pth{ \dr_\mu u_0\pth{ \dr_\theta g^{(1)}_{\rho\sigma} \dr_\theta g^{(2)}_{\beta)\nu} +  \dr_\theta g^{(2)}_{\rho\sigma} \dr_\theta g^{(1)}_{\beta)\nu}}   - \half  \dr_{\beta)} u_0  \dr_\theta g^{(1)}_{\rho\sigma}  \dr_\theta g^{(2)}_{\mu\nu}    }
\\& = -\sin\pth{\frac{u_0}{\la}}\cos\pth{\frac{u_0}{\la}} g_0^{\mu\rho}g_0^{\nu\sigma}  \dr_{(\alpha}u_0 \pth{ \dr_\mu u_0\pth{ F^{(1)}_{\rho\sigma} \F_{\beta)\nu} +  \F_{\rho\sigma} F^{(1)}_{\beta)\nu}}   - \half  \dr_{\beta)} u_0  F^{(1)}_{\rho\sigma}  \F_{\mu\nu}    }
\\&\quad -\sin\pth{\frac{u_0}{\la}} \cos\pth{\frac{u_0}{\la}} g_0^{\mu\rho}g_0^{\nu\sigma}  \dr_{(\alpha}u_0 \pth{ \dr_\mu u_0\pth{ F^{(1)}_{\rho\sigma} F^{(2,1)}_{\beta)\nu} +  F^{(2,1)}_{\rho\sigma} F^{(1)}_{\beta)\nu}}   - \half  \dr_{\beta)} u_0  F^{(1)}_{\rho\sigma}  F^{(2,1)}_{\mu\nu}    }
\\&\quad +2\sin\pth{\frac{u_0}{\la}} \sin\pth{\frac{2u_0}{\la}} g_0^{\mu\rho}g_0^{\nu\sigma}  \dr_{(\alpha}u_0 \pth{ \dr_\mu u_0\pth{ F^{(1)}_{\rho\sigma} F^{(2,2)}_{\beta)\nu} +  F^{(2,2)}_{\rho\sigma} F^{(1)}_{\beta)\nu}}   - \half  \dr_{\beta)} u_0  F^{(1)}_{\rho\sigma}  F^{(2,2)}_{\mu\nu}    },
\end{align*}
where we also used \eqref{expression g1} and \eqref{expression g2}. This concludes the proof of \eqref{P1}.

\subsection{Proof of Proposition \ref{prop expansion gauge part}}

The gauge term $H^\rho=g^{\mu\nu}\Gamma^\rho_{\mu\nu}$ is of the form $g^{-2}\dr g$. In order to expand $H^\rho$, we will first expand $\dr g$ and then $g^{-2}$. More precisely, we have
\begin{align*}
H^\rho & = g^{\rho\si}g^{\mu\nu}\pth{ \dr_\mu g_{\si\nu} - \half \dr_\si g_{\mu\nu} }.
\end{align*}
For convenience, let us define
\begin{align*}
H^\rho\left[ T \right] & = g^{\rho\si}g^{\mu\nu}\pth{ \dr_\mu T_{\si\nu} - \half \dr_\si T_{\mu\nu} },
\end{align*}
for $T$ a symmetric 2-tensor. With this notation, we have
\begin{align*}
H^\rho & = H^\rho[g_0] + \la H^\rho\left[ g^{(1)} \right]  + \la^2 H^\rho\left[ g^{(2)} \right] + \la^2 H^\rho\left[ \h \right] + \la^3 H^\rho\left[ g^{(3)} \right].   
\end{align*}
We expand each term in this expression. For $g_0$, we use \eqref{wave condition g0} and obtain
\begin{align*}
H^\rho[g_0] & = g^{\rho\si}g_0^{\mu\nu}\pth{ \dr_\mu (g_0)_{\si\nu} - \half \dr_\si (g_0)_{\mu\nu} } + \la g^{\rho\si}(g^{\mu\nu})^{(\geq 1)}\pth{ \dr_\mu (g_0)_{\si\nu} - \half \dr_\si (g_0)_{\mu\nu} }
\\& = \la g^{\rho\si}(g^{\mu\nu})^{(\geq 1)}\pth{ \dr_\mu (g_0)_{\si\nu} - \half \dr_\si (g_0)_{\mu\nu} }.
\end{align*}
We now use further the expansion of the inverse (especially \eqref{inverse g2}) and obtain
\begin{align*}
H^\rho[g_0] & = \la \cos\pth{\frac{u_0}{\la}} g_0^{-3}F^{(1)}\dr g_0 - \la^2 g_0^{\rho\si}\h^{\mu\nu}\pth{ \dr_\mu (g_0)_{\si\nu} - \half \dr_\si (g_0)_{\mu\nu} }  
\\&\quad  + \la^2 g_0^{-3}\pth{ (1+g_0^{-1}) (F^{(1)})^2 + F^{(2,i)} + \F }\dr g_0
\\&\quad +  \la^3 \pth{ g^{-1}(g^{-1})^{(\geq 3)}\dr g_0 + (g^{-1})^{(\geq 2)}(g^{-1})^{(1)}\dr g_0 +  (g^{-1})^{(\geq 1)}(g^{\mu\nu})^{(2)}\dr g_0 }.
\end{align*}
If $\mathrm{T}$ is an oscillating function and $S$ a symmetric 2-tensor then
\begin{align*}
H^\rho \left[ \mathrm{T}\pth{\frac{u_0}{\la}} S \right] & = \frac{1}{\la} \mathrm{T}'\pth{\frac{u_0}{\la}}g^{\rho\si}g^{\mu\nu}\pth{ \dr_\mu u_0 S_{\si\nu} - \half \dr_\si u_0 S_{\mu\nu} } + \mathrm{T}\pth{\frac{u_0}{\la}} H^\rho\left[ S \right] 
\\& = \frac{1}{\la} \mathrm{T}'\pth{\frac{u_0}{\la}}g^{\rho\si}\Pol_\si(S) +  \mathrm{T}'\pth{\frac{u_0}{\la}}g^{\rho\si}(g^{\mu\nu})^{(\geq 1)}\pth{ \dr_\mu u_0 S_{\si\nu} - \half \dr_\si u_0 S_{\mu\nu} } 
\\&\quad + \mathrm{T}\pth{\frac{u_0}{\la}} H^\rho\left[ S \right] .
\end{align*}
Using also $\Pol\left( F^{(1)} \right)=0$ and $F^{(1)}_{L_0 \a}=0$, this gives
\begin{align*}
\la H^\rho\left[ g^{(1)} \right]  & =  -  \frac{\la}{4} \sin\pth{\frac{2u_0}{\la}}g_0^{\rho\si} \dr_\si u_0 \left| F^{(1)}\right|^2_{g_0}  +\la^2\sin\pth{\frac{u_0}{\la}}g_0^{\rho\si}\h^{\mu\nu}\pth{ \dr_\mu u_0 F^{(1)}_{\si\nu} - \half \dr_\si u_0 F^{(1)}_{\mu\nu} } 
\\&\quad + \la \cos\pth{\frac{u_0}{\la}} g^{-2}\dr F^{(1)}
\\&\quad   + \la^2\pth{ g_0^{-3}\pth{ (F^{(1)})^2 + F^{(2,i)} + \F }\dr u_0 F^{(1)}+ (g^{-1})^{(\geq 1)}(g^{-1})^{( 1)}\dr u_0 F^{(1)} } 
\\&\quad + \la^3 \pth{ g^{-1}(g^{-1})^{(\geq 3)}\dr u_0 F^{(1)}  +  (g^{-1})^{(\geq 1)}(g^{-1})^{(2)}\dr u_0 F^{(1)} },
\end{align*}
and
\begin{align*}
\la^2 H^\rho\left[ g^{(2)} \right]  & = \la \cos\pth{\frac{u_0}{\la}}\pth{ g_0^{\rho\si}\Pol_\si\pth{\F}  + g_0^{\rho\si}\Pol_\si\pth{F^{(2,1)}} } - \la \sin\pth{\frac{2u_0}{\la}}g_0^{\rho\si}\Pol_\si\pth{F^{(2,2)}} 
\\&\quad + \la^2 \sin\pth{\frac{u_0}{\la}} H^\rho\left[ \F \right] + \la^2 \pth{ g^{-2}\dr F^{(2,i)} + (g^{-1})^{(\geq 1)}g^{-1}\dr u_0 (\F + F^{(2,i)})   },
\end{align*}
and
\begin{align*}
\la^3  H^\rho\left[ g^{(3)} \right]  & = \la^2 g^{\rho\si}\Pol_\si\pth{ \dr_\theta g^{(3)} }  + \la^3 H^\rho\left[ \tilde{\dr}g^{(3)} \right] + \la^3 g^{-1}(g^{-1})^{(\geq 1)}\dr u_0 \dr_\theta g^{(3)} .
\end{align*}
Collecting terms for each power of $\la$ in the previous expressions, we obtain the various expressions of Proposition \ref{prop expansion gauge part}, using in particular the fact that
\begin{align*}
\Upsilon^\rho & =  H^\rho\left[ \h \right] + \sin\pth{\frac{u_0}{\la}} H^\rho\left[ \F \right]  +  \la H^\rho\left[ \tilde{\dr}g^{(3)} \right]  
\\&\quad - g_0^{\rho\sigma} \h^{\mu\nu} \left( \dr_\mu (g_0)_{\sigma\nu} - \half \dr_\sigma (g_0)_{\mu\nu} - \sin\left(\frac{u_0}{\lambda}\right)  \left(   \dr_\mu u_0 F^{(1)}_{\sigma\nu} - \half \dr_\sigma u_0 F^{(1)}_{\mu\nu} \right) \right).
\end{align*}

\subsection{Proof of Proposition \ref{prop expansion ricci tensor}}

From \eqref{Ricci generalised wave coordinates} we obtain
\begin{align*}
2R^{(0)}_{\a\b} & = - W^{(0)}_{\a\b} + P^{(0)}_{\a\b} + (g_0)_{\rho(\a} \dr_{\b)}u_0 \dr_\theta (H^{(1)})^\rho,
\end{align*}
where we used \eqref{Hrho formel} to compute $\pth{ H^\rho \dr_\rho g_{\alpha\beta} + g_{\rho (\alpha}\dr_{\beta)}H^\rho }^{(0)}$. We now use \eqref{W0}, \eqref{P0} and \eqref{H1} to obtain
\begin{align*}
2R^{(0)}_{\a\b} & = - \tBox_{g_0}(g_0)_{\alpha\beta} + \sin\left( \frac{u_0}{\lambda} \right) \pth{ -2L_0 F^{(1)}_{\alpha\beta} + (\tBox_{g_0} u_0)F^{(1)}_{\alpha\beta}  }
\\&\quad +P_{\alpha\beta}(g_0)(\dr g_0, \dr g_0) - \frac{1}{4}  \left| F^{(1)} \right|^2_{g_0}  \dr_\alpha u_0 \dr_\beta u_0 
\\&\quad -\sin\pth{ \frac{u_0}{\la}} \pth{ -2 L_0^\rho \Gamma(g_0)^\nu_{(\alpha \rho} F^{(1)}_{\beta)\nu} + \dr_{(\alpha}u_0\hat{P}^{(0,1)}_{\beta)} }  +\frac{1}{4} \cos\pth{\frac{2u_0}{\la}}  \left| F^{(1)} \right|^2_{g_0}  \dr_{\alpha} u_0  \dr_{\beta} u_0
\\&\quad - \sin\pth{ \frac{u_0}{\la}}   \dr_{(\a}u_0  \pth{  \Pol_{\b)} \left( F^{(2,1)} \right) +  \Pol_{\b)}\left(\F\right) + (g_0)_{\rho\b)}(\Tilde{H}^{(1,1)})^\rho   }
\\&\quad + 2 \cos\pth{ \frac{2u_0}{\la}}  \dr_{(\a}u_0  \pth{  -2 \Pol_{\b)} \left( F^{(2,2)} \right)    -\frac{1}{4} \dr_{\b)} u_0 \left| F^{(1)} \right|^2_{g_0} }.
\end{align*}
Regrouping terms with respect to their oscillation behaviour concludes the proof of \eqref{expression Ricci0}. Similarly, from \eqref{Ricci generalised wave coordinates} we obtain
\begin{align*}
2R^{(1)}_{\a\b} & =  - W^{(1)}_{\a\b} + P^{(1)}_{\a\b} + \pth{ H^\rho \dr_\rho g_{\alpha\beta} + g_{\rho (\alpha}\dr_{\beta)}H^\rho }^{(1)}.
\end{align*}
Using \eqref{Hrho formel} we compute the last term:
\begin{align*}
\pth{ H^\rho \dr_\rho g_{\alpha\beta} + g_{\rho (\alpha}\dr_{\beta)}H^\rho }^{(1)} & = (H^{(1)})^\rho\pth{ \dr_\rho (g_0)_{\a\b} - \sin\pth{\frac{u_0}{\la}}\dr_\rho u_0 F^{(1)}_{\a\b} }
\\&\quad + (g_0)_{\rho (\a} \pth{ \tilde{\dr}_{\b)} (H^{(1)})^\rho + \dr_{\b)} u_0 \dr_\theta \pth{ (H^{(2)})^\rho + (\Upsilon^\rho)^{(0)} } }
\\&\quad + \cos\pth{\frac{u_0}{\la}}F^{(1)}_{\rho(\a} \dr_{\b)} u_0 \dr_\theta (H^{(1)})^\rho.
\end{align*}
We now use \eqref{W1}, \eqref{P1}, \eqref{H1} and \eqref{H2} to obtain
\begin{align*}
2R^{(1)}_{\a\b} & =  - \cos\left(\frac{u_0}{\lambda}\right) \pth{ -2 L_0F^{(2,1)}_{\alpha\beta} + (\tBox_{g_0}u_0)F^{(2,1)}_{\alpha\beta} +  \tilde{W}^{(1,1)}_{\alpha\beta}   -2 L_0\F_{\alpha\beta} + (\tBox_{g_0}u_0)\F_{\alpha\beta}   + \h_{L_0L_0}F^{(1)}_{\alpha\beta} }  
\\&\quad +2  \sin\left(\frac{2u_0}{\lambda}\right) \pth{  -2 L_0F^{(2,2)}_{\alpha\beta} + (\tBox_{g_0}u_0)F^{(2,2)}_{\alpha\beta} +  \tilde{W}^{(1,2)}_{\alpha\beta}  } 
\\&\quad -\frac{1}{2} \sin\left(\frac{2u_0}{\lambda}\right)\F_{L_0L_0}F^{(1)}_{\alpha\beta} - \cos\left(\frac{u_0}{\lambda}\right)\cos\left(\frac{2u_0}{\lambda}\right) F^{(2,2)}_{L_0L_0} F^{(1)}_{\alpha\beta}
\\&\quad + \cos\pth{\frac{u_0}{\la}}\pth{ -2L_0^\rho \Gamma(g_0)^\nu_{(\alpha\rho}\pth{  \F_{\nu\beta)} +  F^{(2,1)}_{\nu\beta)}} + P^{(1,1)}_{\a\b} + \dr_{(\alpha}u_0 \hat{P}^{(1,1)}_{\beta)} } 
\\&\quad + \sin\pth{\frac{2u_0}{\la}}\pth{ 4L_0^\rho \Gamma(g_0)^\nu_{(\alpha\rho}  F^{(2,2)}_{\nu\beta)}  +P^{(1,2)}_{\alpha\beta} +  \dr_{(\alpha}u_0  \hat{P}^{(1,2)}_{\beta)} } + \cos\left( \frac{3u_0}{\lambda}\right)\dr_{(\alpha} u_0  \hat{P}_{\beta)}^{(1,3)}  
\\&\quad +   \cos\left( \frac{u_0}{\lambda} \right) \left( g_0^{\rho\sigma} \Pol_\sigma \left( F^{(2,1)} \right) + g_0^{\rho\sigma} \Pol_\sigma\left(\F\right) + (\Tilde{H}^{(1,1)})^\rho \right)\pth{ \dr_\rho (g_0)_{\a\b} - \sin\pth{\frac{u_0}{\la}}\dr_\rho u_0 F^{(1)}_{\a\b} }
\\&\quad +  \sin\left( \frac{2u_0}{\lambda} \right) \left( -2 g_0^{\rho\sigma} \Pol_\sigma \left( F^{(2,2)} \right)    -\frac{1}{4}g_0^{\rho\sigma} \dr_\sigma u_0 \left| F^{(1)} \right|^2_{g_0} \right)\pth{ \dr_\rho (g_0)_{\a\b} - \sin\pth{\frac{u_0}{\la}}\dr_\rho u_0 F^{(1)}_{\a\b} }
\\&\quad +  \cos\pth{ \frac{u_0}{\la}} (g_0)_{\rho (\a}   \dr_{\b)}  \left( g_0^{\rho\sigma} \Pol_\sigma \left( F^{(2,1)} \right) + g_0^{\rho\sigma} \Pol_\sigma\left(\F\right) + (\Tilde{H}^{(1,1)})^\rho \right)
\\&\quad + \sin\pth{ \frac{2u_0}{\la}} (g_0)_{\rho (\a}     \dr_{\b)}\left( -2 g_0^{\rho\sigma} \Pol_\sigma \left( F^{(2,2)} \right)    -\frac{1}{4}g_0^{\rho\sigma} \dr_\sigma u_0 \left| F^{(1)} \right|^2_{g_0} \right) 
\\&\quad +   \dr_{(\a} u_0  \pth{  \Pol_{\b)} \left( \dr_\theta^2 g^{(3)} \right) + (g_0)_{\rho \b)}  \dr_\theta(\tilde{H}^{(2)})^\rho + (g_0)_{\rho \b)}  \dr_\theta(\Upsilon^\rho)^{(0)} } 
\\&\quad - \half \sin\pth{ \frac{2u_0}{\la}}  F^{(1)}_{\rho(\a} \dr_{\b)} u_0  \left( g_0^{\rho\sigma} \Pol_\sigma \left( F^{(2,1)} \right) + g_0^{\rho\sigma} \Pol_\sigma\left(\F\right) + (\Tilde{H}^{(1,1)})^\rho \right)
\\&\quad +2 \cos\pth{\frac{u_0}{\la}} \cos\pth{ \frac{2u_0}{\la}}  F^{(1)}_{\rho(\a} \dr_{\b)} u_0  \left( -2 g_0^{\rho\sigma} \Pol_\sigma \left( F^{(2,2)} \right)    -\frac{1}{4}g_0^{\rho\sigma} \dr_\sigma u_0 \left| F^{(1)} \right|^2_{g_0} \right).
\end{align*}
We regroup terms with respect to their oscillation behaviour and use the following fact to conclude the proof of \eqref{expression Ricci1}:
\begin{align}
g_0^{\rho\sigma}\Omega_\sigma \dr_\rho (g_0)_{\alpha\beta} + (g_0)_{\rho(\alpha} \dr_{\beta)} \left( g_0^{\rho\sigma} \Omega_\sigma \right) = \D_{(\alpha} \Omega_{\beta)},\label{simple fact}
\end{align}
where $\Omega$ is any 1-tensor. We conclude the proof by proving \eqref{expression Ricci2}. From \eqref{Ricci generalised wave coordinates} we obtain
\begin{align*}
2R^{(\geq 2)}_{\a\b} & =  - W^{(\geq 2)}_{\a\b} + P^{(\geq 2)}_{\a\b} + \pth{ H^\rho \dr_\rho g_{\alpha\beta} + g_{\rho (\alpha}\dr_{\beta)}H^\rho }^{(\geq 2)},
\end{align*}
and we simply use \eqref{W2} to conclude.

\section{The commutator $[\Box_{g_0},L_0]$ and proof of Lemma \ref{lem commute null}}\label{section background null structure}

This appendix is devoted to the proof of Lemma \ref{lem commute null}, which estimates the commutator $[\Box_{g_0},L_0]$, where $L_0$ is defined by \eqref{def L0} and $g_0$ is the background metric. A general computation using the expression of the wave operator in coordinates shows that if $X$ is any vector field and $f$ a scalar function we have
\begin{align}
[ \Box_{g_0} , X ] f & = 2 g_0^{\alpha\beta} g_0^{\mu\nu}(\D_\alpha X)_\mu \dr_\beta \dr_\nu f  + \dr_\alpha f \left( \Box_{g_0}X^\alpha + X \left( g_0^{\rho\sigma}\Gamma(g_0)^\alpha_{\rho\sigma} \right) \right)\nonumber
\\& = 2 g_0^{\alpha\beta} g_0^{\mu\nu}(\D_\alpha X)_\mu \dr_\beta \dr_\nu f  + \dr_\alpha f  \Box_{g_0}X^\alpha, \label{commute general computation}
\end{align}
where we also used the wave coordinate condition for the background \eqref{wave condition g0}. In the context of Lemma \ref{lem commute null}, we consider $X=L_0$ and $f$ compactly supported and we are only interested in the principal terms in $[ \Box_{g_0} , L_0 ] f$. Indeed the non-principal part simply satisfies
\begin{align}
\left| \dr_\alpha f  \Box_{g_0}L_0^\alpha  \right| & \lesssim | \dr f |.\label{estim l.o.t commutateur}
\end{align}
with a constant only depending on the background spacetime.  It remains to estimate the second derivatives of $f$ in \eqref{commute general computation}, i.e
\begin{align*}
g_0^{\alpha\beta} g_0^{\mu\nu}(\D_\alpha L_0)_\mu \dr_\beta \dr_\nu f,
\end{align*}
by using the expression of the background metric in the background null frame defined in Section \ref{section BG}:
\begin{align}
g_0^{\alpha\beta} & = - \half L_0^{(\alpha}\Lb_0^{\beta)} + \delta^{AB} e_A^\alpha e_B^\beta.\label{g_0 alpha beta}
\end{align}
Using \eqref{g_0 alpha beta} twice and the fact that $L_0$ is geodesic, that is $\D_{L_0}L_0=0$, we obtain
\begin{align*}
g_0^{\alpha\beta} g_0^{\mu\nu}(\D_\alpha L_0)_\mu \dr_\beta \dr_\nu f & = - \half g_0^{\mu\nu}(\D_{L_0} L_0)_\mu \dr_\nu \Lb_0 f - \half  g_0^{\mu\nu}(\D_{\Lb_0} L_0)_\mu  \dr_\nu L_0 f 
\\&\quad + \delta^{AB}  g_0^{\mu\nu}(\D_{e_A} L_0)_\mu  \dr_\nu e_B f
\\& =  - \half  g_0^{\mu\nu}(\D_{\Lb_0} L_0)_\mu  \dr_\nu L_0 f -\half  \delta^{AB}   L_0^{\mu} (\D_{e_A} L_0)_\mu  \Lb_0 e_B f 
\\&\quad -\half  \delta^{AB}   \Lb_0^{\mu} (\D_{e_A} L_0)_\mu  L_0 e_B f   + \delta^{AB}   \delta^{CD} e_C^\mu  (\D_{e_A} L_0)_\mu  e_D e_B f
\\&= - \half  g_0^{\mu\nu}(\D_{\Lb_0} L_0)_\mu  \dr_\nu L_0 f  -\half  \delta^{AB}   \Lb_0^{\mu} (\D_{e_A} L_0)_\mu  L_0 e_B f  
\\&\quad + \delta^{AB}   \delta^{CD} e_C^\mu  (\D_{e_A} L_0)_\mu  e_D e_B f,
\end{align*}
where we also used $L_0^{\mu} (\D_{e_A} L_0)_\mu=0$, which follows from the fact $g_0(L_0,L_0)=0$ is constant. Using now the fact that $\left| L_0 e_B f \right| \lesssim | \dr L_0 f | + |\dr f|$ and recalling \eqref{commute general computation} and \eqref{estim l.o.t commutateur} we obtain
\begin{align}
\l  [ \Box_{g_0} , L_0 ] f   \r_{L^2} & \lesssim \l e_Ae_B f \r_{L^2}  +  \l \dr L_0 f \r_{L^2} + \l \dr f \r_{L^2},\label{commute box inter a}
\end{align}
where we also used the compact support of $f$ to switch from pointwise estimates to $L^2$ norms.
Therefore, in order to prove Lemma \ref{lem commute null} it remains to estimate $\l e_Ae_B f \r_{L^2}$. This follows from elliptic estimates on $P_{t,u}$. If $K$ denotes the Gauss curvature of $P_{t,u}$ and if $\d \mu_{t,u}$ denotes the volume form on $P_{t,u}$ induced by $\mathring{g}_0$, then the scalar Bochner identity (Proposition 3.5 in \cite{Szeftel2018}) reads
\begin{align}
\int_{P_{t,u}} \left|\mathrm{Hess} f\right|^2_{\mathring{g}_0}  \d \mu_{t,u} &  = \int_{P_{t,u}}  (\lap f)^2   \d \mu_{t,u} - \int_{P_{t,u}} K |\nabb f|^2_{\mathring{g}_0} \d \mu_{t,u}.\label{Bochner}
\end{align}
where
\begin{align*}
 \lap f & = \delta^{AB}e_Ae_Bf + \delta^{AB}\delta^{CD}g(\nabb_{e_A} e_C,e_B)e_Df,
\\ \nabb f & = \delta^{AB}(e_A f) e_B .
\end{align*}
Thanks to the regularity assumptions stated in Section \ref{section BG} and the fact that the 2-surfaces $P_{t,u}$ foliates $\Sigma_t$, we can integrate \eqref{Bochner} in the $N_0$ direction and obtain 
\begin{align*}
\l e_A e_B f \r_{L^2} \lesssim \l \lap f \r_{L^2} + \l \dr f \r_{L^2},
\end{align*}
where we recall that $\l \cdot \r_{L^2}$ denotes $\l \cdot \r_{L^2(\Sigma_t)}$ and where the implicit constant depends only on background quantities. Together with \eqref{commute box inter a} this gives
\begin{align}
\l  [ \Box_{g_0} , L_0 ] f   \r_{L^2} & \lesssim  \l \lap f \r_{L^2} +  \l \dr L_0 f \r_{L^2} + \l \dr f \r_{L^2}.\label{commute box inter b}
\end{align}
We now use Lemma 2.5 of \cite{Szeftel2012} which gives the expression of the wave operator in the null frame and implies
\begin{align*}
\l \lap f \r_{L^2} \lesssim \l \Box_{g_0} f\r_{L^2}  + \l \dr L_0 f \r_{L^2} + \l \dr f \r_{L^2}.
\end{align*}
Together with \eqref{commute box inter b} this concludes the proof of the first part of Lemma \ref{lem commute null}. 

\saut
The second part is proved in a similar way. For $r\geq 1$, we apply \eqref{estim commute null} to $\nabla^r f$ and obtain
\begin{align*}
\l [L_0,\Box_{g_0}]\nabla^r f \r_{L^2} & \lesssim  \l \dr L_0 \nabla^r f \r_{L^2} + \l \Box_{g_0} \nabla^r f \r_{L^2} + \l \dr f \r_{H^r}
\\&\lesssim \l \nabla^r \dr L_0  f \r_{L^2} + \l \nabla^r \Box_{g_0}  f \r_{L^2} + \l \dr f \r_{H^r} + \l  [\dr L_0, \nabla^r] f  \r_{L^2}  + \l  [\Box_{g_0},\nabla^r] f  \r_{L^2}.
\end{align*}
Moreover we have
\begin{align*}
\l  [\dr L_0, \nabla^r] f  \r_{L^2}  + \l  [\Box_{g_0},\nabla^r] f  \r_{L^2} \lesssim \l \dr^2 f \r_{H^{r-1}},
\end{align*}
where we recall that $r\geq 1$.  It remains to notice that
\begin{align*}
\l \nabla^r [L_0,\Box_{g_0}] f \r_{L^2} \lesssim \l [L_0,\Box_{g_0}]\nabla^r f \r_{L^2} +\l \dr f \r_{H^r} + \l \dr^2 f \r_{H^{r-1}}.
\end{align*}
This concludes the proof of Lemma \ref{lem commute null}.

\section{The spectral projections}\label{section LP}

In this section, we prove Lemma \ref{lem commute bis}. The proof is based on the dyadic decomposition at the heart of Littlewood-Paley theory, which we present shortly.

\subsection{Littlewood-Paley theory}\label{section littlewood paley}

This presentation of the Littlewood-Paley theory is based on \cite{Bahouri2011}. We start by considering two smooth radial functions $\chi$ and $\ffi$ from $\R^3$ to the interval $[0,1]$ supported in $\left\{|\xi|\leq \frac{4}{3}\right\}$ and $\left\{ \frac{3}{4} \leq|\xi|\leq \frac{8}{3}\right\}$ respectively and such that
\begin{align}
\chi + \sum_{j\geq 0} \ffi\left(2^{-j}\cdot\right) =1,\label{partition de l'unité}
\end{align}
\begin{align*}
|j-i|\geq 2 \Longrightarrow \supp{\ffi\left(2^{-j}\cdot\right)}\cap \supp{\ffi\left(2^{-i}\cdot\right)}=\emptyset,
\end{align*}
\begin{align*}
j\geq 1 \Longrightarrow \supp{\chi} \cap \supp{\ffi\left(2^{-j}\cdot\right)} = \emptyset,
\end{align*}
\begin{align}
\half \leq \chi^2 + \sum_{j\geq 0}  \ffi^2\left(2^{-j}\cdot\right) \leq 1.\label{almost orthogonality}
\end{align}
The existence of $\chi$ and $\ffi$ is the content of Proposition 2.10 in \cite{Bahouri2011}. They define a dyadic partition of unity, with which we can define the dyadic blocks $\dy_j$. If $u$ is a tempered distribution on $\R^3$ we set $\dy_i u=0$ if $i\leq -2$ and
\begin{align*}
\dy_{-1} u & = \mathcal{F}^{-1}\left(  \chi \mathcal{F}(u)   \right),
\\ \dy_{j} u & = \mathcal{F}^{-1}\left(  \ffi\left(2^{-j}\cdot\right)  \mathcal{F}(u)   \right),
\end{align*}
for $j\geq 0$ and where $\mathcal{F}$ denotes the Fourier transform on $\R^3$. We also define
\begin{align*}
S_j u= \sum_{j'\leq j-1} \dy_{j'}u ,
\end{align*}
for $j\geq 0$.  For $u$ a tempered distribution, we use the notation $\spe{u}=\supp{\mathcal{F}(u)}$. Therefore the support properties of $\chi$ and $\ffi$ imply
\begin{align}
\spe{\dy_{-1}u} & \subset \left\{|\xi|\leq \frac{4}{3}\right\},\label{support D-1}
\\ \spe{\dy_j u} & \subset 2^j \left\{ \frac{3}{4} \leq|\xi|\leq \frac{8}{3}\right\},\label{support Dj}
\end{align}
for $j\geq 0$. The following proposition contains all the basic estimates we need on the dyadic blocks.

\begin{prop}\label{prop dyadic blocks}
Let $u$ be a tempered distribution.
\begin{enumerate}
\item[(i)] The so-called Bernstein estimates hold
\begin{align}
\l X_j\nabla^\alpha u  \r_{L^p} & \lesssim 2^{\alpha j} \l X_j u \r_{L^p},\label{Bernstein estim}
\end{align}
for $X_j\in\{S_j,\dy_j\}$, $\alpha=0,1$ and for all $1\leq p \leq + \infty$.
\item[(ii)] Moreover, the following estimates hold
\begin{align}
\l X_j u \r_{L^\infty}\lesssim 2^{\frac{3}{2}j} \l X_j u \r_{L^2},\label{Linfini L2}
\end{align}
for $X_j\in\{S_j,\dy_j\}$.
\end{enumerate}
\end{prop}

\begin{proof}
For $j\geq -1$, we have $\spe{X_j \nabla u}\subset 2^j B$ for $B$ a ball centered at 0 in $\R^3$. Let $\phi$ be a compactly supported function on $\R^3$ such that $\phi\restriction{B}=1$.  This implies 
\begin{align*}
\mathcal{F}\left( X_j  u \right) = \phi \left( 2^{-j}\cdot \right) \mathcal{F}\left( X_j  u \right).
\end{align*}
Therefore, for $\alpha=0,1$ we have
\begin{align}
X_j \nabla^\alpha u = 2^{(3+\alpha)j} \left(\nabla^\alpha\mathcal{F}^{-1}(\phi) \right)\left( 2^{j}\cdot \right) * X_j u,\label{expression convolée}
\end{align}
where the symbol $*$ denotes the convolution between two functions. The Young inequality implies
\begin{align*}
\l X_j \nabla^\alpha u  \r_{L^p} &\lesssim 2^{(3+\alpha)j} \l \left(\nabla^\alpha\mathcal{F}^{-1}(\phi) \right)\left( 2^{j}\cdot \right)\r_{L^1} \l X_j u\r_{L^p}
\\&\lesssim 2^{\alpha j} \l \nabla^\alpha\mathcal{F}^{-1}(\phi)  \r_{L^1} \l X_j u\r_{L^p},
\end{align*}
where we change variable in the last step. This proves \eqref{Bernstein estim}. Now, if $\alpha=0$ in \eqref{expression convolée}, another case of Young inequality implies
\begin{align*}
\l X_j u \r_{L^\infty} & \lesssim 2^{3j} \l \mathcal{F}^{-1}(\phi) \left( 2^{j}\cdot \right) \r_{L^2}  \l X_j u \r_{L^2}
\\&\lesssim 2^{\frac{3}{2}j} \l \mathcal{F}^{-1}(\phi)  \r_{L^2}  \l X_j u \r_{L^2},
\end{align*}
which concludes the proof of the proposition.
\end{proof}

The property \eqref{partition de l'unité} allows us to decompose each function $u$ into an infinite sum of smooth and localized in Fourier space functions:
\begin{align}
u = \sum_{j\geq -1} \dy_j u .\label{decomposition u}
\end{align}
The series \eqref{decomposition u} is \textit{a priori} purely formal, and though each dyadic blocks $\dy_j u $ is smooth the function $u$ might have to be defined merely as a distribution. However, one can link the regularity of $u$ and the summability properties of the series.  The next proposition gives a characterization of the Sobolev spaces in terms of the convergence of the series \eqref{decomposition u}.

\begin{prop}\label{charac Hs}
If $s\in\R$ and $u$ is a tempered distribution, then 
\begin{align*}
u\in H^s\iff \left(2^{js}\l \dy_j u \r_{L^2} \right)_{j\in\N}\in \ell^2(\N).
\end{align*}
Moreover, there exists $C>0$ such that for all tempered distribution $u$ we have
\begin{align}
\frac{1}{C}\sum_{j\geq -1}2^{2js}\l \dy_j u \r_{L^2}^2 \leq \l u \r^2_{H^s} \leq C \sum_{j\geq -1}2^{2js}\l \dy_j u \r_{L^2}^2.\label{equivalence norme}
\end{align}
\end{prop}

\begin{proof}
By definition of the space $H^s$, we have $\l u \r_{H^s}  = \l \langle \cdot\rangle^s \mathcal{F}(u) \r_{L^2}$.  Using the first inequality in \eqref{almost orthogonality} we obtain
\begin{align*}
\l u \r^2_{H^s} &  \leq 2 \left( \int_{\R^3} \langle \xi \rangle^{2s}\chi^2(\xi)|\mathcal{F}(u)(\xi)|^2\d\xi + \sum_{j\geq 0} \int_{\R^3} \langle \xi \rangle^{2s}\ffi^2\left( 2^{-j} \xi\right) |\mathcal{F}(u)(\xi)|^2\d\xi\right)
\\&\lesssim 2^{-2s} \l \dy_{-1}u \r_{L^2}^2+ \sum_{j\geq 0}2^{2js} \l \dy_j u \r^2_{L^2},
\end{align*}
where we used the support properties of $\chi$ and $\ffi$. This proves half of \eqref{equivalence norme}, the other half is proved similarly using the other inequality in \eqref{almost orthogonality}.
\end{proof} 

The decomposition \eqref{decomposition u} is also very useful to study products of functions. Indeed, formally we have
\begin{align*}
uv = \sum_{j,j'\geq -1} \dy_j u \dy_{j'}v
\end{align*}
for $u$ and $v$ two tempered distributions.  As it is well known the product $uv$ could not even be a tempered distribution. In order to understand the product's properties, Bony introduced in \cite{Bony1981} what is now called the Bony decomposition:
\begin{mydef}\label{def bony}
For $u$ and $v$ two tempered distributions, we define the paraproduct $T_uv$ and the remainder $R(u,v)$ by
\begin{align*}
T_uv & = \sum_j S_{j-1}u \dy_jv,
\\ R(u,v) & = \sum_{|j-k|\leq 1}\dy_j u \dy_k v.
\end{align*}
The Bony decomposition of $uv$ is then
\begin{align*}
uv = T_uv + T_vu + R(u,v).
\end{align*}
\end{mydef}

In the sequel, we don't directly use Bony's decomposition but rather the spectral property of the general terms in the series defining $T_uv$ and $R(u,v)$, namely if $i\geq -1$ and $| j-k|\leq 1$:
\begin{align}
\spe{S_{i-1}u \dy_iv} & \subset 2^i \mathcal{C}'\label{support anneau},
\\ \spe{\dy_j u \dy_k v} & \subset 2^j B'\label{support boule},
\end{align}
with $\mathcal{C}'=\left\{ \frac{1}{12}\leq |\xi| \leq 3 \right\}$ and $B' =\left\{ 0\leq |\xi| \leq \frac{32}{3} \right\}$ .

\subsection{Properties of $\Pi_\leq$ and $\Pi_\geq$}\label{section def projectors}

Recall Definition \ref{def P}, where we define the operators $\Pi_\leq$ and $\Pi_\geq$. Their behaviour with respect to derivatives is similar to the one of dyadic blocks (see Proposition \ref{prop dyadic blocks}) thanks to the support property of their symbols.

\begin{lem}\label{lem P high low}
Let $f$ be a scalar function. We have
\begin{align}
\lambda \l \nabla \Pi_{\leq }(f) \r_{L^2} & \lesssim \l f \r_{L^2},\label{P low}
\\ \l \Pi_{\geq }(f) \r_{L^2} & \lesssim \lambda \l \nabla f \r_{L^2}\label{P high}.
\end{align}
\end{lem}

\begin{proof}
To prove \eqref{P low}, we use Parseval's identity twice and the support property of $\chi_\lambda$:
\begin{align*}
 \l \nabla \Pi_{\leq }(f) \r_{L^2}^2 \leq \int_{|\xi|\leq \frac{2}{\lambda} }|\xi|^2  \left| \mathcal{F}(f)(\xi) \right|^2  \d\xi \lesssim \frac{1}{\lambda^2} \int_{\R^4} \left| \mathcal{F}(f)(\xi) \right|^2  \d\xi =  \frac{1}{\lambda^2}  \l f \r_{L^2}.
\end{align*}
The proof of \eqref{P low} is similar:
\begin{align*}
 \l \Pi_{\geq }(f) \r_{L^2}^2 &  \leq \int_{\lambda|\xi|\geq 1}  \left| \mathcal{F}(f)(\xi) \right|^2 \d\xi \lesssim \lambda^2 \int_{\R^4}  |\xi|^2 \left| \mathcal{F}(f)(\xi) \right|^2 \d\xi = \lambda^2 \l \nabla f \r_{L^2}^2.
\end{align*}
\end{proof}

The following lemma is a special case of Lemma 2.97 in \cite{Bahouri2011}. For the sake of completeness, we redo the proof.

\begin{lem}\label{lem commute BCD}
Let $u$ and $v$ be two tempered distributions. We have 
\begin{align}
\l \left[ u , \Pi_\leq  \right] v \r_{L^2} \lesssim \lambda \l\nabla u \r_{L^\infty} \l v\r_{L^2}.\label{commute BCD}
\end{align}
\end{lem}

\begin{proof}
Using the usual notation of pseudo-differential calculus we have $\Pi_\leq = \chi_1\left( \lambda \nabla \right)$. For $f$ a function we have $\chi_1\left( \lambda \nabla \right)f=\lambda^{-3} \left(\mathcal{F}^{-1}\chi_1\right)\left( \lambda^{-1}\cdot\right) * f $. This gives
\begin{align*}
\left( \left[ u ,  \chi_1\left( \lambda \nabla \right)  \right] v \right) (x) & = \frac{1}{\lambda^3}\int_{\R^3} \left(\mathcal{F}^{-1}\chi_1\right)\left( \lambda^{-1}(x-y)\right) (u(y)-u(x))v(y) \d y
\\& = - \frac{1}{\lambda^3}\int_{[0,1]\times\R^3}\left(\mathcal{F}^{-1}\chi_1\right)\left( \lambda^{-1}z\right) \langle \nabla u(x-z\tau),z \rangle v(x-z) \d z\d\tau,
\end{align*}
where we used the Taylor expansion for formula for $u$. Defining $\Psi(z)=z\left(\mathcal{F}^{-1}\chi_1\right)(z)$ we obtain
\begin{align*}
\left| \left( \left[ u ,  \chi_1\left( \lambda \nabla \right)  \right] v \right) (x)  \right| & \leq \frac{\l\nabla u \r_{L^\infty}}{\lambda^2} \int_{[0,1]\times\R^3}  \left| \Psi\left( \lambda^{-1}z \right) \right|    | v(x-z)| \d z\d\tau.
\end{align*}
We take the $L^2$ norm of this inequality to obtain after a last change of variable in $\Psi$:
\begin{align*}
\l \left[ u ,  \chi_1\left( \lambda \nabla \right)  \right] v \r_{L^2} \leq \lambda \l\nabla u \r_{L^\infty} \l v\r_{L^2} \l \Psi \r_{L^1}.
\end{align*}
This concludes the proof, noting that $\l \Psi \r_{L^1}\lesssim \l \chi_1 \r_{C^k}$ for $k$ large enough.
\end{proof}

\subsection{Proof of Lemma \ref{lem commute bis}}\label{section proof lem commute bis}

Let $N\in\N$ the unique integer such that $\frac{1}{\lambda}\leq 2^N < \frac{2}{\lambda}$.  In this proof, we will mainly use the fact that $\chi_\lambda(\xi)=1$ if $|\xi|\leq 2^{N-1}$ and $\chi_\lambda(\xi)=0$ if $|\xi|\geq 2^{N+1}$ (recall that $\chi_\lambda$ is defined in Definition \ref{def P}). To benefit from this, we use the dyadic decompositions of $u$ and $v$ and regroup terms following the paraproduct decomposition (see Definition \ref{def bony}) to obtain:
\begin{align*}
\left[ u , \Pi_\leq  \right]\nabla v  & = \sum_j \left[ S_{j-1} u,\Pi_\leq \right] \dy_j \nabla v  + \sum_j \left[ \dy_j u,\Pi_\leq \right] S_{j-1} \nabla v + \sum_{|i-j|\leq 1} \left[ \dy_i u,\Pi_\leq \right] \dy_j \nabla v
\\& = \vcentcolon A_1 + A_2 + A_3.
\end{align*}
We start with $A_1$.  We have
\begin{align}
\left[ S_{j-1} u,\Pi_\leq \right] \dy_j \nabla v & = S_{j-1} u \Pi_\leq \left(  \dy_j \nabla v \right) - \Pi_\leq \left( S_{j-1} u \dy_j \nabla v \right)\label{A1 terme general}.
\end{align}
If $j$ is such that $2^j\times \frac{1}{12}\geq 2^{N+1}$ then \eqref{support Dj} implies $\Pi_\leq \dy_j =0$ and \eqref{support anneau} implies that $\Pi_\leq \left( S_{j-1} u \dy_j \nabla v \right)=0$.  Therefore in this case the two terms in \eqref{A1 terme general} vanish.  If $3\times 2^j \leq 2^{N-1}$ then $\Pi_\leq \dy_j =\dy_j$ and $\Pi_\leq \left( S_{j-1} u \dy_j \nabla v \right)=S_{j-1} u \dy_j \nabla v$, so the two terms in \eqref{A1 terme general} are equal. Therefore there exists $C>0$ independent from $\lambda$ such that
\begin{align*}
\l A_1 \r_{L^2} & \leq \sum_{\frac{1}{C\lambda}\leq 2^j \leq \frac{C}{\lambda}   }\l  \left[ S_{j-1} u,\Pi_\leq \right] \dy_j \nabla v  \r_{L^2}.
\end{align*}
We now use \eqref{commute BCD} and \eqref{Bernstein estim} to obtain
\begin{align*}
\l A_1 \r_{L^2} & \lesssim \lambda \sum_{\frac{1}{C\lambda}\leq 2^j \leq \frac{C}{\lambda}   }    \l \nabla S_{j-1}u \r_{L^\infty} \l \dy_j\nabla v \r_{L^2}
\\& \lesssim \lambda \l \nabla u \r_{L^\infty} \sum_{\frac{1}{C\lambda}\leq 2^j \leq \frac{C}{\lambda}   }   2^j \l \dy_j v \r_{L^2}.
\end{align*}
Using now the Cauchy-Schwarz inequality for finite sums, the characterization of $L^2$ given in Proposition \ref{charac Hs} and $\left( \sum_{\frac{1}{C\lambda}\leq 2^j \leq \frac{C}{\lambda}   }   2^{pj}\right)^{\frac{1}{p}} \lesssim \lambda^{-1}$ we obtain
\begin{align}
\l A_1 \r_{L^2} & \lesssim \l \nabla u \r_{L^\infty} \l v \r_{L^2}.\label{conclusion A1}
\end{align}

\par\leavevmode\par
For $A_2$, we have
\begin{align}
\left[ \dy_j u,\Pi_\leq \right] S_{j-1} \nabla v & = \dy_j u \Pi_\leq \left( S_{j-1}  \nabla v \right) - \Pi_\leq \left( \dy_j u S_{j-1}  \nabla v \right)\label{A2 terme general}.
\end{align}
If $3\times 2^j \leq 2^{N-1}$ then $\Pi_\leq S_{j-1}=S_{j-1}$ and \eqref{support anneau} still implies $\Pi_\leq \left( \dy_j u S_{j-1}  \nabla v \right)= \dy_j u S_{j-1}  \nabla v $, so the two terms in \eqref{A2 terme general} are equal. If $2^j\times\frac{1}{12}\geq 2^{N+1}$, then $\Pi_\leq S_{j-1}=\Pi_\leq$ and $\Pi_\leq \left( \dy_j u S_{j-1}  \nabla v \right)=0$. Therefore there exists $C$ independent from $\lambda$ such that
\begin{align*}
\l A_2\r_{L^2} & \leq \sum_{\frac{1}{C\lambda}\leq 2^j \leq \frac{C}{\lambda} } \l \left[ \dy_j u,\Pi_\leq \right] S_{j-1} \nabla v \r_{L^2} + \sum_{ \frac{C}{\lambda}\leq 2^j} \l \dy_j u \Pi_\leq \left(  \nabla v \right) \r_{L^2}.
\end{align*}
We treat the first sum as we treated $A_1$:
\begin{align*}
\sum_{\frac{1}{C\lambda}\leq 2^j \leq \frac{C}{\lambda} } \l \left[ \dy_j u,\Pi_\leq \right] S_{j-1} \nabla v \r_{L^2} & \lesssim \lambda  \sum_{\frac{1}{C\lambda}\leq 2^j \leq \frac{C}{\lambda} } \l \nabla \dy_j u \r_{L^\infty} \l S_{j-1} \nabla v \r_{L^2}
\\&\lesssim \l \nabla  u \r_{L^\infty} \l  v \r_{L^2} \lambda  \sum_{\frac{1}{C\lambda}\leq 2^j \leq \frac{C}{\lambda} } 2^j
\\&\lesssim \l \nabla  u \r_{L^\infty} \l  v \r_{L^2}.
\end{align*}
For the second sum,  we use \eqref{P low} and then \eqref{Linfini L2}
\begin{align*}
\sum_{ \frac{C}{\lambda}\leq 2^j} \l \dy_j u \Pi_\leq \left(  \nabla v \right) \r_{L^2} & \lesssim\frac{1}{\lambda} \l  v \r_{L^2} \sum_{ \frac{C}{\lambda}\leq 2^j} \l \dy_j u \r_{L^\infty}
\\& \lesssim\frac{1}{\lambda} \l  v \r_{L^2} \sum_{ \frac{C}{\lambda}\leq 2^j} 2^{-2j} 2^{\frac{7}{2}j} \l \dy_j u \r_{L^2}
\\& \lesssim \lambda \l u \r_{H^{\frac{7}{2}}}\l  v \r_{L^2} ,
\end{align*}
where we used the Cauchy-Schwarz inequality in $\ell^2$, the characterization of $H^{\frac{7}{2}}$ given in Proposition \ref{charac Hs} and $\left( \sum_{ \frac{C}{\lambda}\leq 2^j} 2^{-2j} \right)^\half \lesssim \lambda^2$. We finally obtain
\begin{align}
\l A_2 \r_{L^2} \lesssim \left( \l \nabla  u \r_{L^\infty}  +   \lambda \l u \r_{H^{\frac{7}{2}}}  \right) \l  v \r_{L^2} .\label{conclusion A2}
\end{align}

\par\leavevmode\par
For $A_3$, we define 
\begin{align}
R_j = \sum_{i=j-1}^{j+1} \left[ \dy_i u,\Pi_\leq \right] \dy_j \nabla v =  \sum_{i=j-1}^{j+1}\left(  \dy_iu\Pi_\leq \left( \dy_j\nabla v\right) - \Pi_\leq \left( \dy_i u \dy_j\nabla v \right)  \right),\label{Rj}
\end{align}
so that $A_3=\sum_j R_j$. 
If $2^j\times\frac{32}{3}\leq 2^{N-1}$ then \eqref{support boule} implies $\Pi_\leq \left( \dy_i u \dy_j\nabla v \right)  = \dy_i u \dy_j\nabla v $ and \eqref{support Dj} implies $\Pi_\leq \dy_j = \dy_j$, so that the two terms in \eqref{Rj} are equal and $R_j=0$. Therefore there exists a constant $C>0$ independent from $\lambda$ such that
\begin{align*}
\l A_3 \r_{L^2} \leq \sum_{\frac{C}{\lambda}\leq 2^j} \sum_{i=j-1}^{j+1}  \l \left[ \dy_i u,\Pi_\leq \right] \dy_j \nabla v  \r_{L^2}.
\end{align*}
We use \eqref{commute BCD}, \eqref{Bernstein estim} and \eqref{Linfini L2}
\begin{align*}
\l A_3 \r_{L^2} & \lesssim \lambda \sum_{\frac{C}{\lambda}\leq 2^j} \sum_{i=j-1}^{j+1} 2^{\frac{5}{2}i+j} \l  \dy_i u\r_{L^2} \l \dy_j v \r_{L^2}
\\& \lesssim \lambda \sum_{\frac{C}{\lambda}\leq 2^j} 2^{\frac{7}{2}j} \l  \dy_j u\r_{L^2} \l \dy_j v \r_{L^2},
\end{align*}
where we neglect the sum over $i$. Using again the Cauchy-Schwarz inequality in $\ell^2$, the characterization of $H^{\frac{7}{2}}$ and $L^2$ given in Proposition \ref{charac Hs} we finally obtain
\begin{align*}
\l A_3 \r_{L^2} \lesssim   \lambda \l u \r_{H^{\frac{7}{2}}}   \l  v \r_{L^2}.
\end{align*}
Together with \eqref{conclusion A1} and \eqref{conclusion A2} this last estimate concludes the proof of Lemma \ref{lem commute bis}.

\bibliographystyle{alpha}
\newcommand{\etalchar}[1]{$^{#1}$}

\end{document}